\documentclass[sigconf]{acmart}

\AtBeginDocument{%
  }

\copyrightyear{2024}
\acmYear{2024}
\setcopyright{acmlicensed}
\acmConference[CCS '24] {Proceedings of the 2024 ACM SIGSAC Conference on Computer and Communications Security}{October 14--18, 2024}{Salt Lake City, UT, USA.}
\acmBooktitle{Proceedings of the 2024 ACM SIGSAC Conference on Computer and Communications Security (CCS '24), October 14--18, 2024, Salt Lake City, UT, USA}
\acmISBN{979-8-4007-0636-3/24/10}
\acmDOI{10.1145/3658644.3690334}

\usepackage{color} 
\usepackage{xcolor}
\newcommand{\lwx}[1]{\textcolor{black}{#1}}
\newcommand{\lyf}[1]{\textcolor{black}{#1}}

\newcommand{\louis}[1]{\textcolor{black}{#1}}

\newtheorem{theorem}{Theorem}[section]

\newtheorem{lemma}[theorem]{Lemma}
\newenvironment{proofsketch}{\begin{proof}[Proof sketch]}{\end{proof}}

\usepackage{tikz}
\usepackage{array}
\usepackage{algorithm}
\usepackage{algorithmic}
\usepackage{multirow}
\usepackage{makecell}
\usepackage{booktabs} % 添加美观的水平线
\usepackage{lipsum}
\usepackage{graphicx}
\usepackage{wasysym}
\usepackage{stfloats}
\usepackage{pgf}
\usepackage{CJKutf8}
\usepackage{pifont}
\usepackage[normalem]{ulem}
\usepackage{svg}
\usepackage{makecell}
\usepackage{balance}

\usepackage{url}

\usepackage{enumitem}
\usepackage{adjustbox}
\setlist[enumerate]{leftmargin=*}
\setlist[itemize]{leftmargin=*}
\setenumerate[1]{itemsep=0pt,partopsep=0pt,parsep=\parskip,topsep=5pt}
\setitemize[1]{itemsep=0pt,partopsep=0pt,parsep=\parskip,topsep=5pt}
\setdescription{itemsep=0pt,partopsep=0pt,parsep=\parskip,topsep=5pt}

\usepackage{subcaption} % 图像
\begin{document}

% \title{\textsc{Neural Dehydration}: \louis{Removing} Black-box Watermarks
% from Deep Learning Models with Limited Data}
% \title{\textsc{Neural Dehydration}: \louis{Universal} Erasure of Black-box Watermarks from DNNs with Limited Data}
\title{\textsc{Neural Dehydration}: Effective Erasure of Black-box Watermarks from DNNs with Limited Data}


% \author{
% \IEEEauthorblockN{Yifan Lu, Wenxuan Li, Mi Zhang, Xudong Pan, Min Yang}
% \IEEEauthorblockA{Fudan University, China\\ 
% \{luyf23@m., liwx22@m., mi\_zhang@, xdpan@, m\_yang@\}fudan.edu.cn } 
% }

\author{Yifan Lu}
\orcid{0009-0000-7500-3722}
\affiliation{
  \institution{Fudan University}
  \city{Shanghai}
  \country{China}
}
\email{luyf23@m.fudan.edu.cn}

\author{Wenxuan Li}
\orcid{0009-0005-7636-5190}
\affiliation{
  \institution{Fudan University}
  \city{Shanghai}
  \country{China}
}
\email{liwx22@m.fudan.edu.cn}

\author{Mi Zhang}
\orcid{0000-0003-3567-3478}
\affiliation{
  \institution{Fudan University}
  \city{Shanghai}
  \country{China}
}
\email{mi_zhang@fudan.edu.cn}
\authornotemark[1]

\author{Xudong Pan}
\orcid{0000-0003-1394-0395}
\affiliation{
  \institution{Fudan University}
  \city{Shanghai}
  \country{China}
}
\email{xdpan@fudan.edu.cn}

\author{Min Yang}
\orcid{0000-0001-9714-5545}
\affiliation{
  \institution{Fudan University}
  \city{Shanghai}
  \country{China}
}
\email{m_yang@fudan.edu.cn}
\authornote{Corresponding authors: Mi Zhang and Min Yang.}

% \renewcommand{\shortauthors}{Yifan Lu et al.}

\begin{abstract}
To protect the intellectual property of well-trained deep neural networks (DNNs), black-box watermarks, which are embedded into the prediction behavior of DNN models on a set of specially-crafted samples and extracted from suspect models using only API access, have gained increasing popularity in both academy and industry.
Watermark robustness is usually implemented against attackers who steal the protected model and obfuscate its parameters for watermark removal.
However, current robustness evaluations are primarily performed under moderate attacks or unrealistic settings.
Existing removal attacks could only crack a small subset of the mainstream black-box watermarks, and fall short in four key aspects: incomplete removal, reliance on prior knowledge of the watermark, performance degradation, and high dependency on data.
% However, existing attacks are only effective in removing the watermark when over 30\% private training data are available to the attacker, which is somehow impractical.

% Recent studies empirically prove the robustness of most black-box watermarking schemes against existing watermark removal attacks when the attacker has limited access to the private training data (e.g., $<3\%$)

% oakland思路：现在对水印攻防的探索大多建立在较多数据上，当data limited时之前的attacks都不能cracking

% In this paper, we propose \louis{\textsc{Neural Dehydration} (referred to as \textsc{Dehydra}), a novel watermark removal attack which is effective against most of mainstream black-box DNN watermarking schemes, with almost no dependence on the dataset availability.}

In this paper, we propose a watermark-agnostic removal attack called \textsc{Neural Dehydration} (\textit{abbrev.} \textsc{Dehydra}), which effectively erases all ten mainstream black-box watermarks from DNNs, with only limited or even no data dependence.
In general, our attack pipeline exploits the internals of the protected model to recover and unlearn the watermark message.
We further design target class detection and recovered sample splitting algorithms to reduce the utility loss and achieve data-free watermark removal on five of the watermarking schemes.
We conduct comprehensive evaluation of \textsc{Dehydra} against ten mainstream black-box watermarks on three benchmark datasets and DNN architectures.
Compared with existing removal attacks, \textsc{Dehydra} achieves strong removal effectiveness across all the covered watermarks, preserving at least $90\%$ of the stolen model utility, under the data-limited settings, i.e., less than $2\%$ of the training data or even data-free.
Our work reveals the vulnerabilities of existing black-box DNN watermarks in realistic settings, highlighting the urgent need for more robust watermarking techniques.
To facilitate future studies, we open-source our code in the following repository: \url{https://github.com/LouisVann/Dehydra}.
\end{abstract}





\maketitle

\section{Introduction}

% DNN影响深远……训练开销大，宝贵IP……攻击者或许想窃取牟利……
In recent years, deep neural networks (DNNs) are empowering real-world applications in computer vision~\cite{parkhi2015deep, zhang2019medical}, natural language processing~\cite{hassan2018convolutional, devlin2018bert}, and autonomous driving~\cite{bachute2021autonomous, allodi2016machine}.
However, training a modern DNN from scratch requires time-consuming data collection and high-end computing resources.
To protect the intellectual property of well-trained DNNs, \textit{model watermarking} is an emerging tool for verifying the ownership of DNNs in case of model stealing or abusing.

% 黑盒水印被提出（类似后门、投毒），作为一种抵御模型偷取的产权追溯方法。
% 黑盒水印的植入与验证方案。
% 它们的共性是过拟合到水印样本，但可能各种各样的设计
Generally, in a DNN watermarking scheme, the secret watermark information is embedded into the \textit{target model} at the watermark embedding stage during training.
At the subsequent verification stage, the watermark can be extracted from the \textit{suspect model} to determine its ownership.
Depending on how the suspect model is accessed during verification, current watermark algorithms can be categorized into \textit{white-box} and \textit{black-box} watermarks.
A white-box watermark is usually embedded into the suspect model's internal parameters~\cite{uchida2017embedding, wang2021riga, chen2021lottery} or neuron activation~\cite{darvish2019deepsigns}.
In comparison, a black-box watermark is embedded into a model's prediction behaviour on a set of specially-crafted samples (i.e., \textit{watermark data}) by specifying their expected classification results (i.e., \textit{target classes}) \cite{adi2018turning, zhang2018protecting, jia2021entangled, lee2022evaluating}.
Owing to the weaker access requirement for verification, black-box watermarks have gained increasing popularity in both academia and industry (e.g., IBM~\cite{IBM}), with various designs in watermark data and target class settings spurred.

% 我们聚焦于黑盒水印
% 黑盒水印共性在于一些额外记忆，但有各自的设计

% 水印对于移除攻击（解释、定义）的鲁棒性至关重要。然而，现有的评估对抗性不够。
% 移除攻击分类……缺点介绍：前两类不够强力、破坏acc，且可能需要一些水印假设；extraction虽然强力，但对于数据依赖太高。
The robustness of black-box DNN watermarking schemes against \textit{removal attacks}, which modifies the model parameters to cause verification failure, is crucial for their real-world trustworthiness.
\louis{Unfortunately, current robustness evaluations are mostly performed under moderate attacks or unrealistic settings.}
Specifically, existing attacks are mainly categorized into three types: pruning-based \cite{uchida2017embedding, liu2018finepruning}, finetuning-based \cite{libo2021refit, Shafieinejad2019regularization, Zhong2022distraction} and unlearning-based \cite{Aiken2020Laundering}.
However, these approaches could only crack a small subset of existing black-box watermarks (i.e., reducing the watermark verification success rate to almost random, with most of the utility preserved),
which are unable to pose practical threats when the attacker is watermark-agnostic.
% incompetent to threaten the diverse watermarking schemes for a physical watermark-agnostic attacker.
To be more specific, these attacks may to varied degrees remove the watermark incompletely, require prior knowledge of the target watermark, hamper the model's utility or have a strong dependence on the clean dataset.
% Moreover, these attacks could only crack a subset of existing black-box watermarks, leaving the rest several watermarks still surviving after attack.
% \lyf{The difficulty for removal lies in the special robustness designs in the target watermarks, as well as the complex designs in the choices of the watermark data and target class settings.}
Recently, another type of model extraction-based attack has been demonstrated to effectively remove the watermarks when over 30\% of the training data is available~\cite{tramer2016stealing, lukas2022sok}.
However, as commercial DNNs are usually trained on private data, it is usually impractical to access such an amount of training data.
A general comparison of their attack budgets is summarized in Table~\ref{tab:atk_req} and a more detailed analysis is presented in Section~\ref{sec:back-removal}.

\noindent\textbf{Our Work.}
% In this paper, we propose a novel watermark-agnostic removal attack called \textsc{Neural Dehydration} (\textit{abbrev.} \textsc{Dehydra}), which effectively erases ten existing black-box DNN watermarks from standard deep learning architectures under the data-limited settings.
In this paper, we propose a watermark-agnostic removal attack called \textsc{Neural Dehydration} (\textit{abbrev.} \textsc{Dehydra}), which effectively erases all ten mainstream black-box watermarks from DNNs, with only limited or even no data dependence.
To achieve these objectives, we mainly address the following two challenges:

% https://arxiv.org/pdf/2312.09665.pdf

% \noindent$\bullet$
\louis{\textit{C1: How to design a watermark-agnostic effective removal attack?}}

Despite the complex designs of existing black-box watermarks in watermark data and target class settings, we identify a shared trait: they exploit the over-parameterization of DNNs to especially memorize the watermark data correlated to the target classes.
Therefore, the model internals should be a sufficient source to reconstruct the watermark information.
Inspired by this, we design a general-purpose attack framework to \textbf{recover and unlearn} the watermark data from the protected model, similar to the dehydration reaction in chemistry.
Specifically, at the first stage, we use an aggressive model inversion technique to reverse-engineer class-wise samples close to the real watermark data from a target watermarked model (\S\ref{sec:recover}).
% This recovering method is especially powerful, compared to existing trigger reverse-engineering or general inversion methods.
At the second stage, we intentionally unlearn these samples during the finetuning process (\S\ref{sec:unlearn}).
This basic \textsc{Dehydra} framework is effective against all the ten investigated black-box watermarks.
% However it may affect the model's utility, mainly because the recovered samples could also contain information critical to the main task.

% \lyf{Under this intuition, our work uncovers a new attack surface through recovering watermark samples from what is memorized in a watermarked DNN model.}
% In this paper, we explore the possibility of recovering information close to watermark data from a DNN model as much as possible, and then unlearn it to remove the watermarks.
% \pxd{This attack framework is established on the understanding of the commonality of black-box watermarks, and is general-purpose and watermark-agnostic therefore.

\louis{\textit{C2: How to preserve the model utility and reduce the data dependence during attack?}}

Despite effective, the basic attack might compromise the model's utility, because the recovered samples could also contain information critical to the main task.
To address this issue, we further enhance \textsc{Dehydra} in the following directions.
\ding{182} We improve the recovering stage of \textsc{Dehydra} via incorporating \textbf{target class detection}.
% it is redundant and can even affect utility to simply apply the framework.
For watermarks with fixed target classes (such as~\cite{zhang2018protecting} and~\cite{jia2021entangled}, where all the watermark data are paired with the identical target class), we derive an observation called \textit{watermark smoothness discrepancy} that, \textit{the loss landscape of model is smoother on the target class than other classes}.
We provide both theoretical and empirical analysis on the class-wise smoothness properties, and leverage the discrepancy to distinguish whether the target watermark has a fixed target class, and, if so, detect its target class.
This allows us to specially unlearn the recovered samples belonging to the target class when cracking a fixed-class watermark (\S\ref{sec:recover2}).
\ding{183} We improve the unlearning stage of \textsc{Dehydra} via \textbf{splitting the recovered samples}.
Our design are based on the observed \textit{normal data dominance} phenomenon that, 
\textit{the class-wise recovered samples should be closer to normal samples compared with the watermark samples in that class}.
% among the chaotic neuron activations of the class-wise recovered samples, the consistently activated neurons will exhibit the highest correlation with the normal samples in the target classes.
Based on this, we develop a neuron-level criterion to carefully split the recovered samples into proxy watermark data and proxy normal data, with the proxy watermark data to ensure the removal effectiveness and the proxy normal data to preserve the original utility (\S\ref{sec:unlearn2}).
% (2) Additionally, indiscriminately unlearning all inverted samples might also decrease model utility.
% We found that the watermark task is irrelevant to the main task, and the neurons watermark data heavily rely on are always different from those of normal data.
% Under this observation, we design a split criterion during the second unlearning stage to separate the samples close to watermark data from the inverted ones, leaving the rest as close to normal data.
With the above improvements, \textsc{Dehydra} minimally impacts the model's utility while ensuring its watermark-agnostic removal effectiveness.
% is watermark-agnostic and generally effective, with little impact on the utility.
Moreover, the split proxy normal data can even allow for a data-free removal attack for watermarks identified with fixed target classes.
Table~\ref{tab:atk_req} summarizes the advantages of our attack compared with existing removal approaches.

\louis{We conduct comprehensive evaluation of \textsc{Dehydra} against ten mainstream black-box watermarks on three benchmark datasets and DNN architectures, under three different data settings.
Compared with six baseline removal attacks, our proposed \textsc{Dehydra} achieves strong removal effectiveness across all the covered watermarks, preserving  at least $90\%$ model utility, under the data-limited settings, i.e., less than $2\%$ of the training data or even data-free.
% is proved more effective in watermark removal and has less impact on the model's utility, while requiring more relaxed or even no assumptions on the dataset availability.
}

To facilitate future studies, we open-source our code in \url{https://github.com/LouisVann/Dehydra}.
Also, an extended version of this paper, with appendix included, is available in \cite{lu2024dehyd}.

\noindent\textbf{Our Contribution.}
We summarize the key contributions of this work as follows:
\begin{itemize}
    \item \louis{We propose a watermark-agnostic removal framework, namely \textsc{Dehydra}, which exploits the model internals to recover and unlearn the underlying watermarks.}
    % \item By summarizing the commonality of existing black-box DNN watermarks in especially memorizing watermark data correlated to the target labels, we uncover a new removal attack surface via recovering and unlearning watermark samples from a target watermarked model. Our novel attack \textsc{Dehydra} is watermark-agnostic and effective against almost all the existing black-box watermarks.
    \item \louis{Based on in-depth analysis, theoretical justifications and pilot studies}, we further improve \textsc{Dehydra} with target class detection and recovered sample splitting algorithms, which help preserve the model utility and significantly reduce the dataset dependence.
    % \item Based on pilot studies, analytical results and extensive evaluations on a wide range of black-box watermarks, we further improve the basic \textsc{Dehydra} with target class detection and recovered sample splitting algorithms. The improvements help reduce the utility loss caused by \textsc{Dehydra} on the target model and relax the dependence on in-distribution/transfer datasets when attacking half of the black-box watermarking schemes.
    \item Extensive evaluations on various settings against ten mainstream black-box watermarks show that \textsc{Dehydra} is generally effective and has minimal utility impact, with much more relaxed or even no assumptions on the dataset availability.
    % \item We conduct comprehensive evaluation of \textsc{Dehydra} against ten mainstream black-box watermarks on three benchmark datasets and DNN architectures, under three different data settings. Compared with six baseline removal attacks, our proposed \textsc{Dehydra} is proved more effective in watermark removal and has less impact on the model's utility, while requiring more relaxed or even no assumptions on the dataset availability.
    % To facilitate future studies, we open-source our code in the following repository: 
    % \url{https://anonymous.4open.science/r/Neural-Dehydration-0E6A}.

\end{itemize}

\begin{table}[!t]
\caption{Comparison of attack budgets between our attack and existing watermark removal attacks, where $\CIRCLE / \LEFTcircle / \Circle$ represent high, medium and low (or no) attack budget.}
\label{tab:atk_req}
\vspace{-0.1in}
\centering
\setlength{\tabcolsep}{3pt}
\begin{tabular}{c c c c c} % <-- Alignments: l c w |
\hline
  \textbf{Attack} & \textbf{Removal} & \textbf{Watermark} & \textbf{Utility} & \textbf{Dataset}\\
  \textbf{Type} & \textbf{Ineffectiveness} & \textbf{Knowledge} & \textbf{Loss} & \textbf{Access}\\
  \hline
  \textbf{Prune} & \LEFTcircle & \Circle & \LEFTcircle & \Circle \\
  \textbf{Finetune} & \LEFTcircle & \Circle & \LEFTcircle & \CIRCLE \\
  \textbf{Unlearn} & \Circle & \CIRCLE & \Circle & \LEFTcircle \\
  \textbf{\textsc{Dehydra}} & \Circle & \Circle & \Circle & \Circle \\
  \hline
\end{tabular}
\vspace{-0.1in}
\end{table}

\section{Background}
\subsection{DNN Watermarks}\label{sec:back-wm}
% \noindent$\bullet$\textbf{ Model Watermarks.}\label{sec:back-wm}
\louis{Traditional model protection techniques employ proactive strategies to prevent unauthorized access, such as encrypting model parameters or restricting access to APIs.
However, these methods have been proven by subsequent research to face leakage risks \cite{tramer2016stealing, sun2021mind}.
To further claim and trace the ownership of DNN models, model watermarking technology has emerged.}

% Digital watermarking is originally designed for protecting the copyright of digital media~\cite{kahng1998watermarking} and recently applied to protect the intellectual property of DNN models.
A white-box watermarking scheme typically embeds a sequence of secret messages into the parameters or neural activations of the target model, and thus requires white-box access of the suspect model to extract the watermark~\cite{uchida2017embedding, wang2021riga, darvish2019deepsigns, chen2021lottery}.
On the contrary, a black-box watermark \louis{(sometimes also called backdoor-based \cite{Shafieinejad2019regularization, liu2021wild}, poisoning-based \cite{Guo2020PST} or trigger set-based \cite{lee2022evaluating} watermarks)} is embedded in the model's prediction behavior on  a set of specially-crafted secret samples (i.e., \textit{watermark data}) by specifying their expected prediction results (i.e., \textit{target classes})~\cite{adi2018turning, zhang2018protecting, le2020afs, jia2021entangled}.
Therefore, the verification process of black-box watermarking only requires the access to the prediction API.

To build a trustworthy ownership verification process, an ideal watermarking scheme should satisfy at least the following three key requirements~\cite{lukas2022sok, yan2023rethinking}.
\ding{182} \textbf{Fidelity:} The performance degradation on the target model should be as low as possible after the watermark is embedded.
\ding{183} \textbf{Integrity:} Models which are trained independently without access to the target model should not be verified as containing the watermark.
\ding{184} \textbf{Robustness:} The embedded watermark in the target model should be resistant to potential removal attacks. In the worst cases when the watermark is removed, the utility of the model should decrease dramatically.

We mainly study on the black-box watermarks due to their more realistic settings.
Existing black-box watermark schemes have diverse designs in the choices of the \textit{watermark data} and the \textit{target label}. In terms of watermark data, some work leverages specially-crafted patterns~\cite{zhang2018protecting, jia2021entangled} or random noise~\cite{zhang2018protecting} as watermark patterns.
Other black-box watermarks may use out-of-distribution samples as watermark data~\cite{adi2018turning, zhang2018protecting}.
Besides, in-distribution clean samples~\cite{namba2019ew}, adversarial samples near the decision boundary~\cite{le2020afs} or images with an imperceptible logo embedded by an encoder~\cite{li2019blind} can also be exploited as watermark data.
In terms of the target class setting, existing black-box DNN watermarks can be categorized into fixed-class watermarks, where all the watermark data are paired with the identical target class (such as~\cite{zhang2018protecting, jia2021entangled, li2019piracy}),
and non-fixed-class watermarks, where each watermark data is paired with its own target class (such as~\cite{adi2018turning, le2020afs, namba2019ew, li2019blind, guo2018mark}).
% one line of work simply maps these watermark data to a fixed target class, either manually set~\cite{zhang2018protecting, jia2021entangled} or produced by a cryptographic hashing function~\cite{li2019piracy}.
% Another line uses non-fixed target labels for the watermark data, i.e., each watermark sample is mapped to a specific class~\cite{adi2018turning, le2020afs, namba2019ew,li2019blind, guo2018mark}.
\louis{These complex designs pose additional challenges to a practical attacker, who usually has no knowledge of the underlying watermark schemes.}

% \sout{Consequently, black-box DNN watermarks can be considered as a generalization of DNN backdoors, encompassing more complex designs.} 

\subsection{Watermark Removal Attacks}\label{sec:back-removal}
% \noindent$\bullet$\textbf{ Watermark Removal Attacks.}\label{sec:back-removal}
Black-box model watermarks are to some extent similar to backdoor attacks, as they both establish connections between some specified data and the target labels.
Therefore, classical methods in backdoor defense, such as pruning, finetuning or trigger reverse-engineering, are usually considered potential threats for black-box watermarks~\cite{adi2018turning, zhang2018protecting, jia2021entangled}.
Recently, some attacks especially targeting black-box watermarks are also proposed.

Existing watermark removal attacks can be mainly categorized into three types as follows, with the pros-and-cons of attacks from each category summarized in Table~\ref{tab:atk_req}.
\begin{itemize}
    \item \textbf{Pruning-based Attacks}: They prune the redundant weights~\cite{uchida2017embedding} or neurons~\cite{liu2018finepruning} in DNNs.
    However, to invalidate the underlying watermarks, these methods need to prune a large proportion of weights or neurons, which causes an unacceptable utility loss.
    \item \textbf{Finetuning-based Attacks}: They continue to train the target model for a few epochs with some carefully designed finetuning techniques, such as learning rate schedule~\cite{libo2019leveraging, libo2021refit}, dataset augmentation~\cite{liu2021wild}, weight regularization~\cite{Shafieinejad2019regularization}, or continue learning with attention distraction~\cite{Zhong2022distraction}.
    Unfortunately, these attacks could only breach a subset of watermarks. Furthermore, due to stringent training restrictions, they require a large amount of clean data to maintain utility.
    \sloppy
    \item \textbf{Unlearning-based Attacks}: They are mainly designed for black-box watermarks whose watermark data have the fixed pattern and the identical target class~\cite{Aiken2020Laundering, Wang2021GAN}.
    For instance, the \textit{Laundering} attack~\cite{Aiken2020Laundering} uses the classical trigger reverse-engineering technique in backdoor defense literature \textit{NeuralCleanse}~\cite{wang2019NC}, followed by an unlearning process.
    However, these attacks pose strong assumptions on the forms of underlying watermarks and are therefore not applicable to physical scenarios where attackers have no knowledge of the target watermarks.    
    % conventional backdoor trigger reverse-engineering techniques in e.g., \textit{NeuralCleanse}~\cite{wang2019NC}, which poses rather strong assumptions on the forms of underlying watermarks and are therefore not watermark-agnostic.
    Also, current unlearning-based attacks all need access to source training data, which further limits their applicable scenarios.
\end{itemize}
Note that there also exists another type of attack aiming at training a new surrogate model using the knowledge transferred from the given watermarked model, i.e., \textbf{model extraction attacks}~\cite{Yang2019distill_wm, jia2021entangled, lukas2022sok}.
However, extraction attacks typically require a large query dataset and entail heavy computation costs~\cite{Shafieinejad2019regularization}. Therefore, 
we do not include these attacks in the main experiment, but will make a comparison with them under different data settings in Section~\ref{sec:discussion}.

% Our work focuses on the vulnerabilities revealed by model inversion.
% Note that this has also been studied by Zhang et al.~\cite{zhang2018protecting}.
% However, they used a rather naive inversion method from~\cite{fredrikson2015inversion_conf} and only presented the visually-meaningless inverted samples.
% They neither delved into the neuron coverage of the inverted samples, nor attempted further watermark removal via unlearning the inverted ones.

\section{Security Settings}
% Threat Model
% In this section, we first introduce the target black-box watermark schemes and the mainstream adopted paradigms.
% scenario of a removal attack against black-box model watermarks, together with the notations we use throughout the paper.
% Then we describe our assumptions on attacker's capabilities.

\subsection{Mechanism of Black-box Watermarking}
% \noindent{\textbf{Mechanism of Black-box Watermarking.}}
The true model owner is denoted as $\mathcal{O}$.
% two parties: the true model owner $\mathcal{O}$ and the malicious attacker $\mathcal{A}$.
During the training, $\mathcal{O}$ embeds a watermark into the target model by training on both the clean dataset $X = \{(x_i, y_i)\}_{i=1}^{N_1} \ (y_i \in \{0,1,\dots,C-1\})$ from its main task, where $C$ is the total number of the classification classes,
% and on a set of specially-crafted watermark data $X_w = \{(x_{w_i}, y_{w_i})\}_{i=1}^{N_2} \ (y_{w_i} \in \{0,1,\dots,C-1\})$, where target classes $y_{w_i}$ can be either identical or sample-specific. We denote the watermarked model as $f_{w}$.
and on a set of specially-crafted watermark data $X_w = \{({x_w}_i, {y_w}_i)\}_{i=1}^{N_2} \ ({y_w}_i \in \{0,1,\dots,C-1\})$, where target classes ${y_w}_i$ can be either identical or sample-specific. We denote the watermarked model as $f_{w}$.
In face of potential copyright infringement, $\mathcal{O}$ hopes to verify the model ownership 
by comparing the suspect model's predictions on $X_{w}$ with the pre-defined target labels, via the provided API (i.e., black-box access), as shown in Figure~\ref{fig:atk_setting}.

%% Requirement -> [paradigms] -> observations -> designs
By inspecting the official watermark implementations, we observe the following technical designs which the owner $\mathcal{O}$ will take to meet the three requirements in Section~\ref{sec:back-wm}.
Specifically, \ding{182} to satisfy the \textbf{fidelity} purpose, $\mathcal{O}$ usually choose an over-parameterized model with enough capacity, which is trained on both normal samples and watermark samples, with normal ones being dominant in quantity~\cite{adi2018turning, li2019blind}, i.e., $N_2 \ll N_1$.
\ding{183} To ensure the watermark \textbf{integrity}, $\mathcal{O}$ typically crafts watermark samples $(x_{w_i}, y_{w_i})$ that are significantly different from the normal samples in the target class $y_{w_i}$,
whose size $N_2$ cannot be too small as well, to claim the ownership with high confidence~\cite{jia2021entangled, li2019piracy, lukas2022sok}.
\ding{184} To enhance the \textbf{robustness} of the embedded watermark, $\mathcal{O}$ tends to intentionally increase the involvement of the watermark data in the training process.
For instance, both Adi et al.~\cite{adi2018turning} and Lederer et al.~\cite{lederer2023sba} concatenate the independently sampled normal data and watermark data into a data batch during each training iteration.
Jia et al.~\cite{jia2021entangled} intentionally train the model 
% \dbq{on one batch of watermark data for every several batches of normal data}
% intentionally train the model on every several batches of normal data followed by a batch of watermark data.
on watermark samples and on normal samples alternatively.
Lukas et al.~\cite{lukas2022sok} explicitly boost the ratio of watermark data when constructing the total training set \cite{wmk_robustness}.

\subsection{Threat Model}

\begin{figure}[t]
    \centering
    \includegraphics[width=1\linewidth]{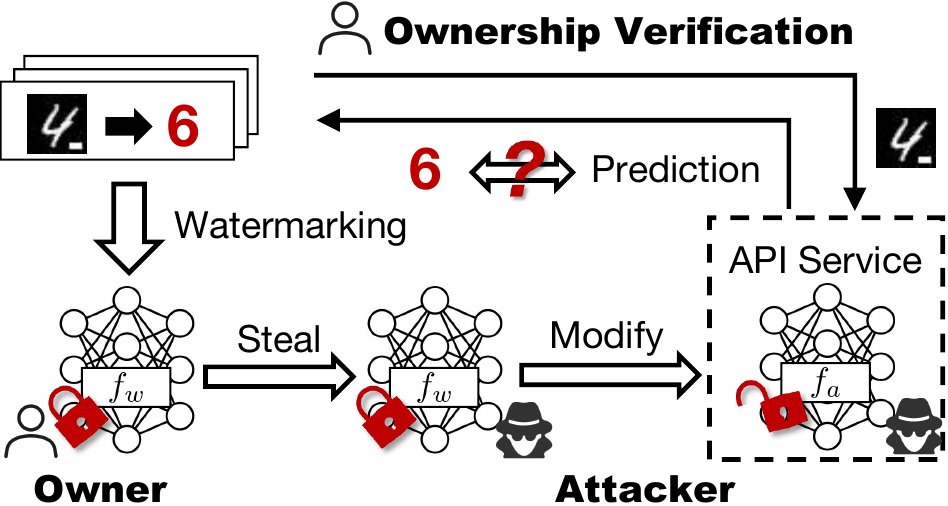}
    \caption{An illustration of our threat model.
    The model owner trains a watermarked model $f_w$ with a secret set of watermark samples, and hopes to perform ownership verification via API queries for suspect models.
    On the other hand, the attacker has managed to acquire the white-box $f_w$, and aims to derive a surrogate model $f_a$ with the watermark removed, evading subsequent ownership verification.}
    \label{fig:atk_setting}
\vspace{-0.1in}
\end{figure}

\noindent \textbf{Attack Scenario.}
In Figure~\ref{fig:atk_setting},
the attacker $\mathcal{A}$ is unwilling to train a model of their own, due to the lack of training data, computation resource or expertise.
Instead, $\mathcal{A}$ attempts to infringe on the intellectual property of $f_w$ trained by $\mathcal{O}$, and then host a similar service illegally or abuse it to make a profit, while evading the subsequent watermark verification process.

Similar to the settings in many watermark robustness studies~\cite{lukas2022sok, lee2022evaluating, libo2021refit, Zhong2022distraction, Aiken2020Laundering}, we assume the attacker has managed to acquire the white-box access to the target model $f_{w}$, i.e., they can observe the model parameters as well as modify them.
\louis{This access could be achieved by, e.g., directly downloading the model after dishonestly agreeing with the official license (on open-source model hubs such as Hugging Face~\cite{hugg}, DNNs might be released under licenses permitting research but prohibiting commercial use~\cite{TheTuringWay}), decrypting the possibly encrypted model from memory (DNNs on end devices) \cite{sun2021mind}, or stealing functional replicas of DNNs via API access (MLaaS) \cite{tramer2016stealing}.}
% by means of, e.g, compromising $\mathcal{O}$'s server or receiving assistance from an insider.  %(which is outside the scope of this paper).
% $\mathcal{A}$ has white-box access to observe the parameters of $f_w$ as well as to modify them, similar to the settings in the recent watermark robustness studies~\cite{lukas2022sok, lee2022evaluating}.
Additionally, we assume that $\mathcal{A}$ is aware of the existence of the watermark in $f_{w}$, and therefore their goal is to derive a surrogate model $f_{a}$ from $f_{w}$ with similar performance while invalidating the ownership verification.
% These settings are similar to those in many recent watermark robustness studies~\cite{lukas2022sok, lee2022evaluating, libo2021refit, Zhong2022distraction}

\noindent \textbf{Attack Budget.}
Specifically, as shown in Table~\ref{tab:atk_req},
\ding{182} $\mathcal{A}$ has no knowledge of the specific watermark algorithm or the watermark data $X_w$ that $\mathcal{O}$ used (i.e., watermark-agnostic).
\ding{183} $\mathcal{A}$ hopes to significantly lower $f_{a}$'s watermark verification accuracy on $X_{w}$, \ding{184} while preserving its clean accuracy on the main task.
To make this attack practical, \ding{185} $\mathcal{A}$ also wants to relax the dataset access requirement as much as possible.

% Specifically, $\mathcal{A}$ hopes to significantly lower $f_{a}$'s watermark verification accuracy on $X_{w}$ while maintaining clean accuracy on its main task as much as possible.
% However, $\mathcal{A}$ neither knows the specific watermark algorithm $\mathcal{O}$ used, nor the knowledge of watermark data $X_{w}$.

% Similar to the settings in~\cite{lukas2022sok, lee2022evaluating}, we assume the attacker $\mathcal{A}$ has white-box access to observe the parameters of $f_w$ as well as to modify the parameters.
% $\mathcal{A}$ is aware of the existence of watermarks in $f_{w}$, but neither knows the specific watermark algorithm $\mathcal{O}$ used, nor the knowledge of watermark data $X_{w}$.

% \lyf{In our setting, the attacker $\mathcal{A}$ has a certain computation resources, but cannot access the whole source training data $X$.
% % Otherwise, he could independently train a non-watermarked model from scratch based on the data
% }
Previous removal attacks may require a substantial subset of the clean training data $X$~\cite{zhang2018protecting, adi2018turning, libo2021refit}, or a considerable size of unlabeled data collected from open sources as proxy data~\cite{Shafieinejad2019regularization, Guo2020PST}.
However, they usually ignore the difficulty of collecting high-quality data at such a scale.
% For instance, Chen et al.~\cite{libo2021refit} assumed the attacker to have at least 20\% of the original training set, which is a rather difficult requirement in practice.
% Also, given that size of a dataset, the attack can already train an adequate model from scratch, with no need to remove the watermark~\cite{Aiken2020Laundering}.
% Collecting samples from open sources seems a good alternative~\cite{Shafieinejad2019regularization, Guo2020PST}, but the distribution of these sample also matters, since finetuning the target model on samples of a different distribution may lead to catastrophic forgetting~\cite{kirkpatrick2017forgetting}.
In this work, we mainly consider the attacker with limited data access, as in~\cite{Aiken2020Laundering, Zhong2022distraction}.
The following three assumptions on dataset availability characterize attackers from being limited in knowledge to almost of zero knowledge:
\begin{itemize}
    \item \textbf{In-distribution}: A limited number of correctly labeled samples (either manually by humans, or automatically by the well-training watermarked model) are available from the same distribution of the target model's main task.
    This setting is realistic since these in-distribution data are often hard to acquire.
    \item \textbf{Transfer}: A small amount of unlabeled data from a transfer distribution is available.
    Note that the distribution of these unlabeled data also matters, since finetuning the target model on samples from a completely different distribution might lead to catastrophic forgetting~\cite{kirkpatrick2017forgetting}.
    \item \textbf{Data-free}: No additional data is available.
    This setting further relaxes the dataset access requirement for the attacker. To the best of our knowledge, our work is the first to formally consider data-free black-box watermark removal settings.
\end{itemize}

\section{General Framework of \textsc{Dehydra}}

\begin{figure*}[t]
    \centering
    \includegraphics[width=0.9\textwidth]{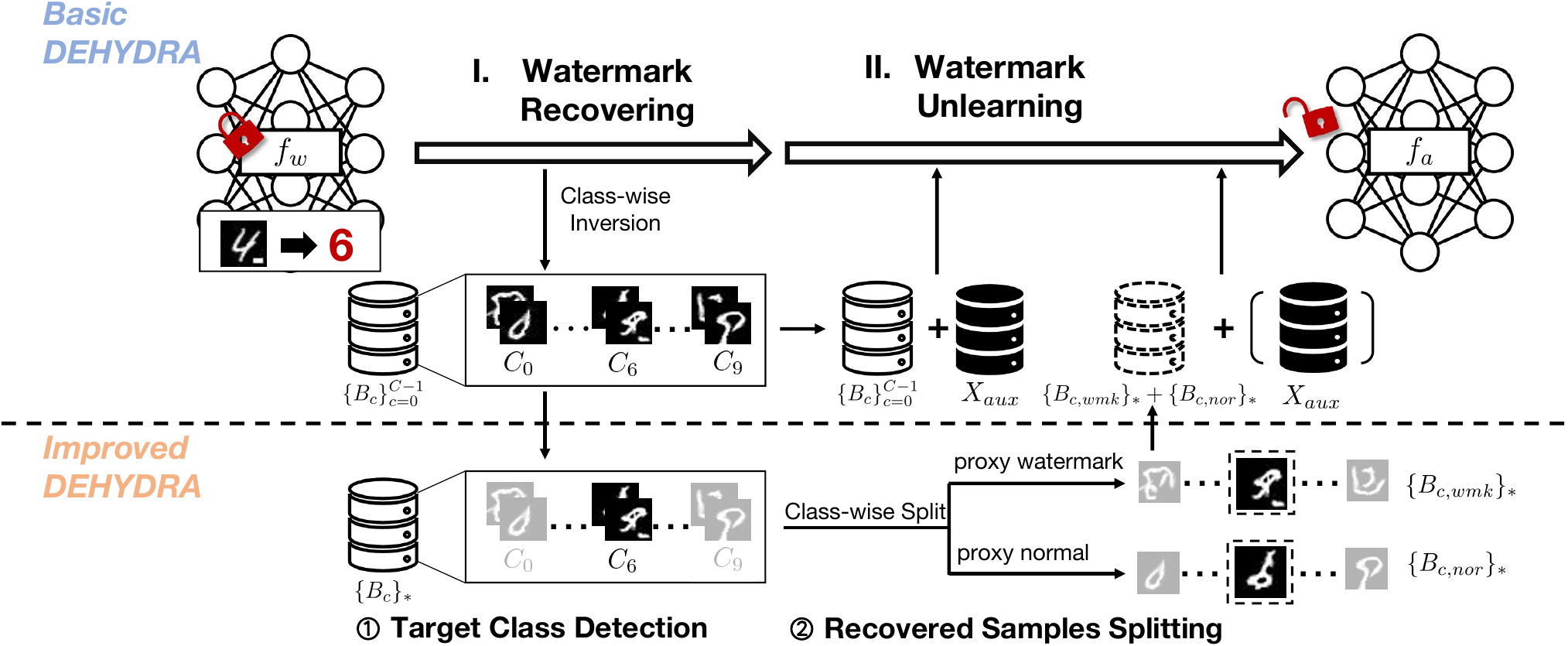}
    \caption{The overview of our \textsc{Dehydra}.
    % to derive a surrogate model $f_a$ from the target model $f_w$, with the underlying watermark erased
    The upper region shows our basic attack, comprising two stages, watermark recovering (\S\ref{sec:recover}), which reconstructs batch samples $\{B_c\}_{c=0}^{C-1}$ close to real watermark data from each class, and watermark unlearning (\S\ref{sec:unlearn}), which unlearns the recovered samples during finetuning, along with the auxiliary dataset $X_{aux}$.
    The lower region shows our improved designs, target class detection (\S\ref{sec:recover2}) and recovered samples splitting (\S\ref{sec:unlearn2}).
    The former detects the target classes after recovering (e.g., $C_6$ in this case), and only retains batches $\{B_c\}_{\ast}$ from those classes. The latter performs class-wise splitting on each $B_c$ before unlearning, to identify
    % proxy watermark data $\{B_{c,wmk}\}_\ast$ and proxy normal data $\{B_{c,nor}\}_\ast$.
    samples closer to watermark or normal data, i.e., $\{B_{c,wmk}\}_\ast$ and $\{B_{c,nor}\}_\ast$, respectively.
    % which detect the possible target class after recovering (\S\ref{sec:recover2}), and identify samples more relevant to the real watermark data before unlearning (\S\ref{sec:unlearn2}).
    }
    \label{fig:main}
\vspace{-0.1in}
\end{figure*}

In general, our \textsc{Dehydra} adopts the following two-stage attack pipeline: \textbf{(1) Watermark Recovering:} We first leverage model inversion technique to recover some samples close to real watermark data (\S\ref{sec:recover}).
\textbf{(2) Watermark Unlearning:} We next deliberately unlearn these samples during the finetuning process, along with the auxiliary dataset $\mathcal{A}$ possesses (\S\ref{sec:unlearn}).
The procedures bear a resemblance to dehydration reactions in chemistry, where water molecules are expelled from compounds, inducing the formation of new substances.
An overview of our general attack framework is shown in the upper region of Figure~\ref{fig:main}.

% This attack makes almost no assumption on the forms of underlying watermarks, i.e., watermark-agnostic, and are effective against various black-box model watermarks. 这里介绍方法就暂时不要说这个效果了

\subsection{Watermark Recovering}\label{sec:recover}
% we use class-wise model inversion technique to directly reverse-engineer information close to watermark data.
\sloppy
\lwx{Despite the complexity of various designs of watermarks, including different watermark patterns and strategies for setting the target classes,}
% Despite the complexity of various underlying watermark forms, including watermark patterns and target labels of watermark data,
we summarize their commonality that, to implement robustness, existing watermarks exploit the over-parameterization property of DNNs to intentionally memorize the watermark data correlated to the pre-defined target labels.
Note that training samples with higher repetition and memorization levels may face greater reconstruction risks~\cite{carlini2019secret}.
% To put it another way, the watermark information is implicitly stored in the watermarked models.
Inspired by this, we believe the model internals should be a sufficient source to reconstruct the watermark information, and hence attempt to use model inversion techniques to reverse-engineer samples close to the real watermark data from a watermarked model.

Concretely, we use optimization-based model inversion to recover the samples hidden in the watermarked model $f_w$.
Given a target class $c$ and a batch of randomly initialized samples $B = \{\hat{x}_i\}_{i=0}^{M-1}$ (\lyf{$\hat{x}_i$ keeps the same shape as the images of $f_w$'s main task,}
% ($\hat{x}_i \in \mathbb{R}^{H\times W \times C}$, $H$, $W$, $C$ being the height, width, channels of the image, 
while $M$ is the number of samples to be optimized for the $c$-th class,
we try to reverse-engineer samples \lwx{which are similar to the watermark data correlated to} the
% close to watermark data correlated to 
target class $c$ by optimizing
\begin{equation}
    \min_{B} (\sum_{\hat{x}_i \in B} \mathcal{L}(f_w(\hat{x}_i), c)) + \mathcal{R}(B), \label{eq:optim}
\end{equation}
where $\mathcal{L}(\cdot)$ is the cross-entropy loss, and the regularizer $\mathcal{R}(\cdot)$ poses our prior knowledge on the watermark data.

% {(put in appendix or simplified)} Specifically, we incorporate two regularization terms $\mathcal{R}_{\ell_2}(\cdot)$ and $\mathcal{R}_{tv}(\cdot)$, as demonstrated in~\cite{mordvintsev2015inceptionism}, \lwx{which correspond to $\ell_2$ distance and total variation, }to steer the optimized samples away from unrealistic ones:
% \begin{equation}
%     \mathcal{R}_{\ell_2}(B) = \sum_{\hat{x}_i \in B} \lVert \hat{x}_i \rVert_2^2,
% \end{equation}
% \begin{equation}
%     \mathcal{R}_{tv}(B) = \sum_{\hat{x}_i \in B} \sum_{(j,k)} \sum_{(j^\prime, k^\prime) \in \delta(j,k)} \lVert \hat{x}_{i_{(j,k)}} - \hat{x}_{i_{(j^\prime,k^\prime)}} \rVert_2^2,
% \end{equation}
% where $\delta(j,k)$ indicates a set of pixels adjacent to the pixel at $(j,k)$ in the image $\hat{x}_i$.

\lyf{In addition to classic priors on natural images (e.g., $\ell_2$ regularization and total variation regularization~\cite{mordvintsev2015inceptionism}), our recovering algorithm mainly exploits the information stored in the batch normalization (BN) layers~\cite{ioffe2015BN}, as is done in~\cite{yin2020dreaming, yoon2022fewshot}.
BN layers are widely adopted in mainstream DNNs~\cite{huang2017densenet, szegedy2016inception, he2016res}, and keep useful statistics of the training data, capturing both low-level and high-level features in the neural network.
Since the owner $\mathcal{O}$ tends to intentionally increase the involvement of the watermark data during the training, 
% , for the robustness purpose described in Section~\ref{sec:target_wms}, 
the watermark data $X_w$ are expected to play an important role in the BN statistics of the target model $f_w$.}
% \lyf{As the usage of batch normalization (BN) layers becomes a widely-adopted conventional practice in training a DNN,}
% \lwx{we assume that the same \lyf{holds} for the target model $f_w$.}
% \lyf{we assume the same for the target model $f_w$.}%
% BN layers often keep useful statistics of the training data, \lwx{capturing both low-level and high-level features in the neural network.}
% Usually, to enhance the robustness of black-box watermark in $f_w$, the owner $O$ usually takes special measures to help $f_w$ memorize these watermark samples, such as boosting their ratio in the total training data as implemented in Watermark-Robustness-Toolbox \footnote{https://github.com/dnn-security/Watermark-Robustness-Toolbox}. Therefore,
% for the watermarked model $f_w$, the watermark data $X_w$ is dominant in the BN statistics.

To better utilize the watermark information hidden in the BN layers,
we employ another feature regularization term $\mathcal{R}_{bn}(\cdot)$ to guide the recovered samples closer to the real watermark data:
\begin{equation}
    \mathcal{R}_{bn}(B) = \sum_{l=1}^{L-2} \lVert\mu_l(B)-\mu_l\rVert_2^2 + \lVert \sigma_l^2(B) - \sigma_l^2 \rVert_2^2, \label{eq:bn}
\end{equation}
where $L$ is the total number of BN layers in $f_w$, $\mu_l$ and $\sigma_l^2$ are the running batch-wise mean and variance stored in the $l$-th BN layer, and $\mu_l(\cdot)$ and $\sigma_l^2(\cdot)$ are current statistics of feature maps before the $l$-th BN layer.
According to~\cite{simonyan2014very, ilyas2019adversarial}, the last several layers of a DNN capture high-level features, which are more correlated with the model's prediction behavior.
To encourage the exploration in class $c$'s decision space and capture more watermark information, we intentionally dismiss the guidance of the last two BN layers.
% Therefore, in the above BN regularization, we intentionally dismiss the guidance of the last two BN layers.
% This design prevents the recovered samples from being pulled towards the center of all
% classes in high-level decision space, and encourages them to explore more space in class $c$'s decision region. % \lyf{, similar to the feature diversity loss proposed in~\cite{yoon2022fewshot}}.

In summary, the regularization term in Equation~\ref{eq:optim} is finally:
\begin{equation}
    \mathcal{R}(B) = \alpha_{\ell_2} \mathcal{R}_{\ell_2}(B) + \alpha_{tv} \mathcal{R}_{tv}(B) + \alpha_{bn} \mathcal{R}_{bn}(B),
\end{equation}
where $\mathcal{R}_{\ell_2}(B)$ and $\mathcal{R}_{tv}(B)$ penalize the $\ell_2$ norm and the total variation of the recovered watermark data. The detailed forms can be found in Appendix~\ref{sec:app:nip}.

% the $\ell_2$ regularization term writes $
%     \mathcal{R}_{\ell_2}(B) = \sum_{\hat{x}_i \in B} \lVert \hat{x}_i \rVert_2^2 $ and the total variation term writes
% $
%     \mathcal{R}_{tv}(B) = \sum_{\hat{x}_i \in B} \sum_{(j,k)} \sum_{(j^\prime, k^\prime) \in \delta(j,k)} \lVert \hat{x}_{i_{(j,k)}} - \hat{x}_{i_{(j^\prime,k^\prime)}} \rVert_2^2.
% $ Here, $\delta(j,k)$ indicates a set of pixels adjacent to the pixel at $(j,k)$ in the image $\hat{x}_i$.

Since the attacker $\mathcal{A}$ has no knowledge of the underlying target classes, $\mathcal{A}$ has to conservatively
% As is set up in Section~\ref{sec:capabilities}, the attacker $\mathcal{A}$ has no knowledge of the target watermark form 
% and is therefore unaware of the target label $c$ associated with the watermark data.
% Additionally, the underlying black-box watermark may not set the same label for all the watermark data. In other words, the target label ${y_w}_i$ in watermark data $X_w$ may be non-fixed and could be any integer in the set $\{0,1,\dots, C-1\}$.
% Therefore, we
repeat the optimization step in Equation~\ref{eq:optim} for all possible classes and finally obtain $C$ batches of recovered samples $\{B_c\}_{c=0}^{C-1}$, where $B_c$ indicates the inverted batch of samples towards class $c$.
These recovered samples are expected to contain enough watermark information of $f_w$.

Note that although the BN layers keep the statistics of the total training set, including samples of all classes, our watermark recovering method is actually using an aggressive class-wise inversion scheme.
This is because our goal is to \textit{recover samples close to the watermark data as much as possible here, instead of to generate realistic images}.
Since the watermark data may constitute a large contribution to the BN statistics during the training, we propose this scheme to greedily absorb information in BN layers, preventing it from diffusing into other classes. This design is especially effective against watermarks with fixed target classes, compared to the arbitrary-class scheme in conventional inversion methods~\cite{yin2020dreaming}. We validate this point in Section~\ref{sec:inv_comp}.

% Note that our goal is to \textit{recover samples close to watermark data as much as possible here, instead of to generate realistic images}.
% Considering that watermark data may constitute a large contribution to BN statistics, we propose to use a class-wise inversion scheme to greedily absorb information in BN layers.
% This design is especially effective against watermarks with fixed target labels, compared with the random label setting in \pxd{DeepInversion where this instructive BN information might be diffused into other labels}.

% These two modifications enable our recovered watermark data \lwx{to be} representative enough of the real watermark data $X_w$.

\subsection{Watermark Unlearning}\label{sec:unlearn}
% Now that we have recovered a set of samples containing sufficient watermark information in the target model $f_w$, we can deliberately unlearn them during the finetuning process.
Directly unlearning these recovered samples $\{B_c\}_{c=0}^{C-1}$ without other constraints may cause catastrophic forgetting, bringing irreversible significant damage to the normal performance.
In the current framework, we consider two data settings for 
$\mathcal{A}$'s auxiliary data $X_{aux}$, in-distribution and transfer. This basic framework is extended to the data-free setting with more improved designs in Section~\ref{sec:improve}.
% , and will extend it to data-free setting in the later improved attack in Section~\ref{sec:improve}.

Assuming the attacker $\mathcal{A}$ has some in-distribution data $X_{aux} = \{(x_i, y_i)\}$ from $f_w$'s main task, they can leverage these labeled samples to preserve model's performance during unlearning.
When $\mathcal{A}$ only has some unlabeled data $\{x_i\}$ from a different distribution, similar goals can be achieved by leveraging model $f_w$ as an oracle to generate pseudo-labels for each sample, yielding the transfer dataset $X_{aux} = \{(x_i, f_w(x_i))\}$.

Starting from $f_w$, the attacker can finetune the surrogate model $f_a$ under the following optimization objective:
\begin{equation}
\begin{aligned}
    \min_{\theta_a} 
    &\underbrace{\sum_{(x_i,y_i)\in X_{aux}} \mathcal{L}(f_a(x_i), y_i)}_
    {\text{preserving original performance}} \\
    + & \underbrace{\alpha_{KL} \sum_{B_c \in \{B_c\}_{c=0}^{C-1}} \sum_{\hat{x}_i \in B_c} KL(f_a(\hat{x}_i), y_{soft}),}_{\text{removing watermark}}
    \label{eq:unlearn}
\end{aligned}
\end{equation}
where $\theta_a$ is the parameters of $f_a$, $KL(\cdot)$ is the Kullback–Leibler divergence, 
$y_{soft}$ is a soft unlearning target of length $C$, where each element equal to $1/C$.
Note that we are using a more consistent unlearning target, compared to the hard random labels, which might lead to gradient conflicts and cancellations during fine-tuning.

The first term in Equation~\ref{eq:unlearn} stabilizes the general prediction behavior of $f_a$ and thus preserves its performance, while the second term encourages $f_a$ to unlearn the watermark-related information.
In this way, we successfully derive a surrogate model $f_a$ to prevent successful watermark verification.

\section{Improved Designs for \textsc{Dehydra}}\label{sec:improve}
In this section, we enhance the general framework of \textsc{Dehydra} in the following directions.
\begin{enumerate}
\item \textbf{Improved Recovering via Target Class Detection}: First, existing black-box watermarks consist of those with fixed target classes and with non-fixed target classes.
For watermarks with fixed target classes (such as \textit{Content}~\cite{zhang2018protecting} and \textit{EWE}~\cite{jia2021entangled}), recovering and unlearning watermark data from other classes bring no benefit to watermark removal, and can even sacrifice the utility of the surrogate model. Therefore, after the original watermark recovering stage, we further distinguish whether the underlying watermark has a fixed target class, and detect its target class if so.
We only retain recovered samples belonging to the target class for a fixed-class watermark (\S\ref{sec:recover2}). % lwx: 这个需要实验验证吗 lyf: 感觉不用，非常直观诶

\item \textbf{Improved Unlearning via Recovered Sample Splitting}: Second, the inverted samples derived by Equation~\ref{eq:optim} may also contain important information of the main task.
Indiscriminately unlearning all the inverted samples from the target classes might also hurt the model utility. Therefore, before the original watermark unlearning stage, we split out samples closer to the real watermark data from the recovered ones and improve the finetuning loss accordingly (\S\ref{sec:unlearn2}).
\end{enumerate}

% Therefore, via a more in-depth analysis on the watermarking mechanism, we improve the two stages of the general \textsc{Dehydra} respectively. 
% At the watermark recovering stage, we distinguish whether the underlying watermark has a fixed target label, and detect its target label if so.
% We only retain recovered samples belonging to the target label for a fixed-class watermark (Section~\ref{sec:recover2}).
% At the watermark unlearning stage, we split out samples closer to real watermark data from recovered ones and improve the finetuning loss accordingly (Section~\ref{sec:unlearn2}).

Figure~\ref{fig:main} shows two improved designs and the complete workflow. The two improvements mainly exploit the recovered samples instead of the auxiliary data, enabling our attack to be extended to a data-free setting for watermarks with a fixed target class. %% 不用那么具体

\subsection{Target Class Detection}\label{sec:recover2}
\noindent$\bullet$\textbf{ Insight: Watermark Smoothness Discrepancy.}
By inspecting existing DNN watermarking implementations \cite{zhang2018protecting, adi2018turning, lukas2022sok, jia2021entangled, lederer2023sba}, we find that for a watermark with a fixed target class, the model tends to learn much more diverse data correlated to the target class during training, compared with other labels, due to the additional incorporation of special watermark samples (for the integrity purpose).
Moreover, the model also learns samples belonging to the target class more prominently, because of the higher participation ratio of these watermark data (for the robustness purpose).

Based on these observations, we propose the \textit{watermark smoothness discrepancy} hypothesis that, for a DNN model embedded with a fixed-class watermark, the loss landscape (i.e., the structure of loss values around the parameter space) is smoother with respect to samples of the target label, compared to those of other labels.
The smoothness property is mathematically defined via the Hessian matrix of loss function~\cite{li2018visualizing, ge2023boosting}, which is usually computationally expensive. We instead approximate the smoothness using the performance of model under perturbations. % which are denoted as input smoothness $S_I$ and parameter smoothness $S_P$ respectively.
A conceptual illustration is shown in Figure~\ref{fig:smoothness_illustration}, where compared with the clean model, the watermarked model demonstrates higher smoothness on Class A and lower on Class B, due to the additional incorporation of watermark samples and the increased involvement during training.

\begin{figure}[t]
\centering
\includegraphics[width=\linewidth]{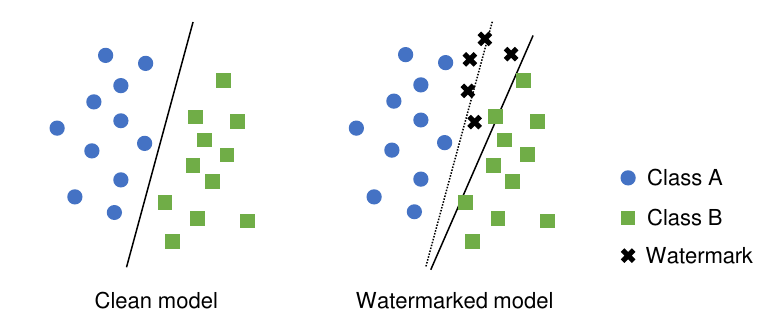}
\vspace{-0.2in}
\caption{An illustration of the watermark smoothness discrepancy hypothesis.
Left: A clean model with its decision boundary.
Right: A fixed-class watermark is embedded, with all watermark samples labeled to Class A.}
\label{fig:smoothness_illustration}
\vspace{-0.2in}
\end{figure}

\noindent$\bullet$ \textbf{Theoretical Justification.}
To provide theoretical evidence, we consider training a watermarked binary classification linear model $f(x) = \text{sign} (\langle \boldsymbol{w}, \boldsymbol{x} \rangle +b)$ on a dataset following mixture Gaussian distributions, together with another set of fixed-class watermark data.
\louis{The detailed problem definition, theorems and proofs are all deferred to Appendix~\ref{sec:proof} due to the page limit.}

Our conclusions are (Remark on Theorem~\ref{th:wm_smooth_complete}):
\textit{For a fixed-class watermark, the target class will exhibit a higher smoothness (both input-level and parameter-level) compared to other classes if the watermark samples are (1) diverse internally and different from the normal ones in the target class, and (2) sampled much more frequently.
During proof, we also find the additive input-level and parameter-level noise would further enlarge the class-wise discrepancies.}

In existing black-box watermarking implementations, the designs of watermark data for the integrity purpose correspond to the first condition, while the training techniques for the robustness purpose satisfy the second condition.
% Moreover, the additive input-level and parameter-level perturbations would further enlarge the class-wise discrepancies.
Therefore, we deduce that the watermark smoothness discrepancy hypothesis should also hold for realistic watermarked DNNs.

Note that another line of work focuses on class-wise robustness discrepancies \cite{tian2021analysis, benz2021robustness} and empirically discovers that, for a normal model trained on a balanced dataset, classes with smaller \textit{inter-class} semantic distances are more vulnerable.
Our findings are exactly complementary to theirs:
% Our findings are consistent with them while presenting a new perspective:
the classes with a larger \textit{intra-class} variance and sampling ratio should be more robust.

% \lyf{\sout{Integrity and robustness are two essential properties of black-box DNN watermarks.
% In order to ensure the integrity of the watermark, owner $\mathcal{O}$ typically crafts watermark samples $(x_{w_i}, y_{w_i})$ that are significantly different from normal samples in the target class $y_{w_i}$.
% Additionally, to enhance robustness of the watermark, $\mathcal{O}$ tends to intentionally increase the involvement of watermark data in the training process.}}

% Existing implementations of black-box watermarks usually craft watermark samples significantly different from the normal ones in the target classes, and sample them more prominently during training, for watermark integrity and robustness respectively.

% Based on the above observations, we propose the \textit{watermark smoothness} hypothesis that, for an over-parameterized model embedded with a fixed-class watermark, the loss landscape (i.e., the structure of loss values around the parameter space of a DNN) is smoother with respect to samples of the target label compared to those of other labels.
% We call this hypothesis \textit{watermark smoothness}, and will conduct a pilot study to validate it.

% \lyf{We note that our theoretical analysis of watermark smoothness is built on simple settings.}

\noindent$\bullet$ \textbf{A Pilot Study.}
To provide a more comprehensive investigation, we empirically evaluate the smoothness on ten realistic models protected by different black-box watermarks.
Following the previous interpretations of smoothness, we define a metric called $SmoothAcc$.
Given a set of samples $X$, we calculate the average prediction accuracy of $(x_i, y_i) \in X$ under perturbations:
\begin{equation}
\begin{aligned}
    &SmoothAcc(X) = \\
    \frac{1}{|X|}& \sum_{(x_i, y_i)\in X} \mathbb{P}(
    \operatorname*{argmax}_j {f_w(x_i+\epsilon_1;\theta_w+\epsilon_2)}_j =y_i),
    \label{eq:smoothacc}
\end{aligned}
\end{equation}
where $\epsilon_1 \in \mathcal{N}(0, \sigma_1^2 I_1)$ is the input-level Gaussian noise and $\epsilon_2 \in \mathcal{N}(0, \sigma_2^2 I_2)$ is the parameter-level Gaussian noise.

We construct ten watermarked ResNet-18~\cite{he2016res} models on CIFAR-10~\cite{cifar10} dataset, embedded with five fixed-class watermarks (with the target class set to 6\footnote{The target classes of the five investigated fixed-class watermarks are set to 6 arbitrarily in this pilot study. We use a randomly chosen label under ten repetitive tests for a systematic study in Section~\ref{sec:exp-target-detection}.}) and five non-fixed-class watermarks respectively (more details are described in Section~\ref{sec:Experimental Settings}).
For each model, to comprehensively evaluate its smoothness on various data categories, we collect 11 batches of samples, including ten batches of normal samples, each from a class in CIFAR-10, and one batch from real watermark samples, with a batch size of 100.
During smoothness analysis, we set the standard deviations of the Gaussian noise $\sigma_1 = 0.5$ and $\sigma_2 = 0.015$\footnote{We at first empirically set the standard deviations of the input-level and parameter-level Gaussian noise. Subsequently, we independently increment each standard deviation until the smoothness discrepancy is evident.}, and repeatedly sample the Gaussian noise 100 times for each sample to approximate the possibility $\mathbb{P}(\cdot)$ in Equation~\ref{eq:smoothacc}.
The results are shown in Figure~\ref{fig:pilot_study}(a).

\begin{figure*}[tb]
    \centering
    \includegraphics[width=\textwidth]{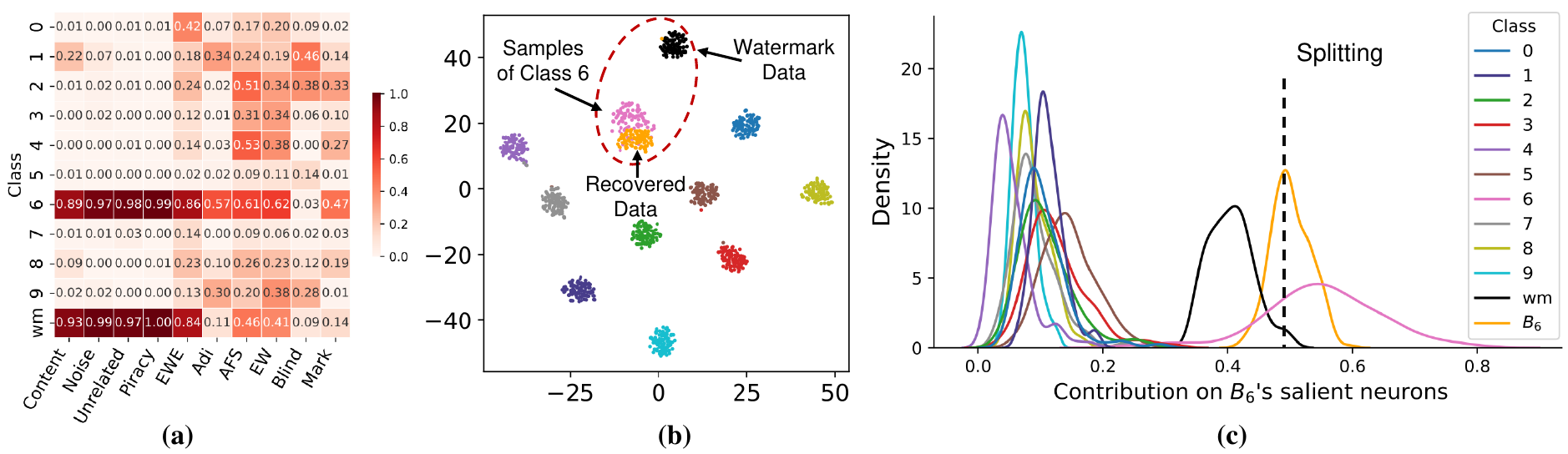}
    \vspace{-0.25in}
    \caption{Pilot study results of the improved designs. (a) Class-wise smoothness analysis of the five fixed-class and five non-fixed-class watermarks on CIFAR-10. (b) Activations visualization of a model protected by the \textit{Content} watermark with target class 6. (c) Distribution of sample contributions per class, on the salient neurons of the recovered samples $B_6$.}
    \label{fig:pilot_study}
\vspace{-0.1in}
\end{figure*}

Evidently, for fixed-class watermarks, the models exhibit much higher smoothness on samples from the target class, including normal samples and watermark samples.
This class-wise difference is less evident for non-fixed-class watermarks, 
where models achieve moderate smoothness on several classes simultaneously.
These results strongly support our hypothesis.

% In summary, the experimental results support the watermark smoothness discrepancy hypothesis.
% For a watermarked DNN model embedded with a fixed-class watermark, the loss landscape is expected to be smoother on the target label compared to other labels.
% In other words, the model is expected to be more robust against input-level and parameter-level perturbations on samples belonging to the target class, including normal samples and watermark samples, compared to those belonging to other classes.

\noindent$\bullet$\textbf{ Technical Designs.}
Based on the observations above, we assume that the watermark smoothness discrepancy phenomenon also exists, when the target model $f_w$ is embedded with a fixed-class watermark, in the class-wise inverted samples $\{B_c\}_{c=0}^{C-1}$.
Supposing the fixed target label is $y_w$, then the class-wise smoothness difference would also be significant in $\{B_c\}_{c=0}^{C-1}$ as the inverted batch $B_{y_w}$ from the target label will cover the space of both normal samples of $y_w$ and watermark samples.

In light of this, we propose to evaluate the smoothness of $f_w$ on $\{B_c\}_{c=0}^{C-1}$ after the watermark recovering, i.e., to calculate $SmoothAcc(B_c)$ for $c \in \{0, 1, \dots, C-1\}$ respectively.
Each $\hat{x}_i \in B_c$ is temporarily labeled as class $c$ during the smoothness evaluation. Next, we sort the batches $\{B_c\}_{c=0}^{C-1}$ according to their $SmoothAcc(B_c)$ values in the descending order, yielding sorted $\{B_{s_0}, B_{s_1}, \dots, B_{s_{C-1}}\}$.
We detect whether $f_w$ is embedded with a fixed-class watermark, and, if so, follow the criterion below to determine the target class: \textbf{a)} If the gap between the recovered two batches evaluated with the highest smoothness, i.e., $SmoothAcc(B_{s_0}) - SmoothAcc(B_{s_1})$, is larger than a certain threshold $T$, we regard the underlying watermark as a fixed-class watermark, and take the target label as $s_0$.
    We then only retain recovered samples $B_{s_0}$ and drop other batches. \textbf{b)} Otherwise, we regard the underlying watermark as a non-fixed-class watermark, and keep all the recovered batches $\{B_c\}_{c=0}^{C-1}$ since each class $c$ might have several correlated watermark samples.
    We additionally record labels of the last two batches, $s_{C-1}$ and $s_{C-2}$, as the least-likely label and second-least-likely label respectively.

In the remainder of this paper, we use $\{B_c\}_{\ast}$ to denote $\{B_{s_0}\}$ if a fixed-class watermark is detected, and $\{B_c\}_{c=0}^{C-1}$ otherwise.
Consequently, after the watermark recovering and target class detection, 
we get batch(es) samples $\{B_c\}_{\ast}$ close to real watermark data $X_w$.

\subsection{Recovered Samples Splitting}\label{sec:unlearn2}
%%% 从BNA的角度解释为啥clean更多（不是因为class-wise的BN层约束，BNA其实没有该约束、但clean甚至比\textsc{Dehydra}更多）
% 1. 模型在clean和wm数据上训练，其中clean是主任务、明显大很多 (Fidelity)
% 2. wm样本分布不同于主要的clean样本 (Integrity) 且被heavily训练 (robustness)，模型很容易对wm完成学习、不管是fixed还是unfixed-class (类似anti-backdoor learning)，后续很少再更新
% 3. 主任务clean样本的各类内部更相似，模型在建模学习它们的过程中需要更多的梯度更新，因此clean信息更丰富、空间更大。
% => normal data dominance
\noindent$\bullet$\textbf{ Insight: Normal Data Dominance.}
We have previously explained the necessity of the class-wise inversion scheme for watermark recovering (\S\ref{sec:recover}), but we find that this scheme would induce competitive objectives during recovering.
Recall the two key components of the recovering objective in Equation \ref{eq:optim}, i.e. the targeted classification loss $\sum_{\hat{x}_i \in B} \mathcal{L}(f_w(\hat{x}_i), c)$ and the BN regularization term $\sum_{l=1}^{L} \lVert\mu_l(B)-\mu_l\rVert_2^2 + \lVert \sigma_l^2(B) - \sigma_l^2 \rVert_2^2$.
For the recovering process on target class $c$, the targeted classification loss will lead all the optimized samples $B_c$ consistently classified to $c$, while the BN regularization term will guide the neuron activations of $B_c$ statistically close to the activations average of the total training set.
Considering that the normal samples constitute the main distribution, different from the watermark distribution, the recovered $B_c$ should be closer to class $c$'s clean data (which belongs to the main distribution) than the watermark samples labeled to $c$.

%%% 从chaotic activations中找一致性，从而split出真正干净的样本 %%%
% 两个loss项是竞争/冲突性的，逆向样本B_c确实都是分到第6类，但BN约束会把B_c拉向全局平均：
%     逆向样本仍旧稳定激活类c的神经元
%     还chaotic地激活了许多其他类的神经元
%     => B_c的显著神经元跟c类强相关
% c类包含干净样本和水印样本，B_c的显著神经元跟谁更相关？
%     BN项把B_c拉向全局平均
%     水印样本通常与主分布不同
%     => B_c的显著神经元跟类c的干净样本更相关
%     ! [需强调] 该“更相关”在面对non-fixed时不那么明显

% As the normal samples and watermark samples tend to rely on different salient neurons~\cite{liu2018finepruning, jia2021entangled, Chen2022LinkBreakerBT},

% \louis{Note that class $c$ may include both normal samples and watermark samples.
% Hence, we take a closer look at the salient neurons of $B_c$ and compare their salience to normal samples and watermark samples in class $c$.
% % Next, we take a closer look at the salient neurons of class $c$ and compare their salience to normal samples and to watermark samples.
% The BN regularization term, as analyzed above, will pull $B_c$ statistically close to the activations center of the training data.
% Considering that the normal samples constitute the main distribution, different from the watermark distribution, $B_c$ should be closer to class $c$'s clean data (belonging to the main distribution) than the watermark samples labeled to $c$.
% Hence, the salient neurons of $B_c$ should be more correlated with normal samples of class $c$.}

Based on the above analysis, we derive another hypothesis named \textit{normal data dominance}, which is intrinsically connected with our class-wise inversion scheme.
Specifically, for a watermarked model (with either fixed or non-fixed target classes), 
the class-wise recovered samples are expected to be closer to normal samples compared with the watermark samples in that class.

\noindent$\bullet$ \textbf{A Pilot Study.}
To gain insights, we analyze the neuron activations of the model protected by \textit{Content}~\cite{zhang2018protecting} watermark here with fixed target class 6.
We first collect 12 batches, including ten batches of normal samples from each class, one batch of watermark samples and one batch of recovered samples $B_6$ from target class 6, each with a batch size 100.
We feed each batch into the watermarked model, extract the activations of the penultimate layer and perform t-SNE dimensionality reduction~\cite{van2008visualizing}.
As shown in Figure~\ref{fig:pilot_study}(b), the recovered samples (the orange cluster) lie slightly close to the global activations center, and hence closer to the normal samples of class 6 (the pink cluster) than the watermark samples (the black cluster, lying separately from the main distribution).
% 这是因为，水印任务跟模型的主任务是完全无关的，从模型参数中恢复水印样本难度极大。尽管我们已经采用策略来提升恢复样本对真实水印样本的覆盖率，但该操作不可避免地会导致逆向样本更靠近全局类中心。
This is because, the watermark task is usually irrelevant to the model's primary task, making it extremely challenging to recover high-quality watermark samples.
Although we have employed the class-wise BN regularization term to enhance the coverage of the recovered $B_6$ over the real watermark samples (\S\ref{sec:recover}), it would also inevitably guide the recovered samples towards the activation average of the total training set, conforming to our analysis above.

% This is because, although we have employed the class-wise BN regularization term to capture more watermark information from model internals (\S\ref{sec:recover}), it would also inevitably guide the recovered samples towards the activation average of the total training set, conforming to our analysis above.
% This conforms to our analysis on the regularization effects of BN statistics.
% The former ensures the coverage of the recovered $B_6$ over the real watermark samples, while the latter calls for a deeper inspection on $B_6$'s activation composition.

% We also find that the recovered samples lie closer to the normal samples of class 6 (the pink cluster) than the watermark samples (the black cluster, lying separately from the main distribution), which also supports our hypothesis. % (similar phenomena are found in the backdoor literature~\cite{Chen2018AC, rajabi2023mdtd})

Next, we delve into the composition of the activations of the recovered samples $B_6$.
Since $B_6$ are pulled statistically close to the activation center of training data of all classes, it is expected that samples of $B_6$ will inevitably activate redundant neurons of class $c$ or neurons important to other classes.
To prevent the noise from these irrelevant activations, we focus on $B_6$'s salient neurons.
Concretely, we utilize the \textit{importance score} for recognizing the salient neurons ~\cite{Chen2022LinkBreakerBT}. Given a batch of samples $X_c$ of class $c$, the score extracts the neuron activations at the $l$-th layer of model $f_w$, i.e., ${f_w}_{(l)}(X_c)$, and calculates an importance score for each $j$-th neuron: $Impt({f_w}_{(l)}(X_c)_j) = \frac{\mu_j}{\sigma_j}$,
where $\mu_j$ and $\sigma_j$ are the mean and the standard deviation of the $j$-th neuron estimated over $X_c$.
% Intuitively, a high importance score implies a significant and stable neuron for $X_c$.
We refer to the neurons with top-5\% importance score as salient neurons of  $X_c$, denoted as $SN(X_c)$.

Next, we capture the salient neurons of the recovered samples $B_6$ and explore the functionality of these consistently activated neurons.
For each sample $x_i \in X_c$ among the 12 batches collected above, we define its total activation values on the salient neurons of $B_6$ as its contribution:
\begin{equation}
    Contribution(x_i) = \sum_{j\in SN(B_6)} {f_w}_{(l)}(x_i)_j.
\end{equation}
Then we plot the distribution of sample contributions of each class, distinguished in color.
In Figure~\ref{fig:pilot_study}(c), 
% samples of classes other than 6 exhibit minimal contribution to $SN(B_6)$, which suggests that despite the chaotic activations of $B_6$, its consistently activated neurons demonstrate a much lower correlation with other classes.
% Moreover, 
for samples belonging to class 6, normal samples contribute the most, followed by samples in $B_6$ itself, and finally the watermark samples.
This indicates that $SN(B_6)$ are more consistently activated by normal samples of class 6 than the watermark samples,
% the recovered samples are closer to the normal samples in comparison to the watermark samples, and the consistently activated neurons of $B_6$ exhibit a stronger correlation with the clean samples from the target classes.
which validates our hypothesis.

More empirical evidence is provided in Figure~\ref{fig:neuron_cross} of Appendix~\ref{app:salient_neuron}, where, in most cases, $SN(B_6)$ exhibit the highest correlation with the normal samples of class 6.
It is worth noting, for non-fixed-class watermarks, $SN(B_6)$ also has a high correlation with the watermark samples besides the normal samples, possibly due to the more complex input-output pairing relations, leaving the watermark data more entangled with the normal data.

% 尽管chaotic activations，这些salient neuron跟其他类不太相关

% 如果对B_6划一刀，那么其中贡献度比较高的部分应该更靠近clean 6 -> 对应到consistency

% \begin{figure}[bt]
%   \begin{subfigure}[t]{0.4\linewidth} %左侧图片
%     \centering
%     \includegraphics[width=\linewidth]{oakland/z_tsne_content_v2.pdf}
%     \subcaption{Visualization of activations from a model protected by the \textit{Content} watermark.}
%     \label{fig:tsne_content}
%     % \vspace{-0.2in}
%   \end{subfigure}%
%   \hspace{1mm}
%   \begin{subfigure}[t]{0.6\linewidth} %右侧图片
%     \centering
%     \includegraphics[width=\linewidth]{oakland/z_displot_v4.pdf}
%     \subcaption{\louis{Distribution of  sample contributions on $SN(B_6)$ of each class.}}
%     \label{fig:displot}
%     % \vspace{-0.2in}
%   \end{subfigure}
%   % \vspace{-0.1in}
%   \caption{Illustration for recovered samples splitting.}
%   \vspace{-0.2in}
% \end{figure}

%%% pilot study
% 直接比较决策空间的clean、wm占用空间是很难的；
% 我们用class-wise逆向的样本来近似决策空间，并巧妙地使用与clean、wm的显著神经元覆盖的角度，来表述两类样本的占用空间
% It is challenging to directly compare the occupied space of normal and watermark samples in each class's decision region, due to the difficulty of depicting the abstract decision region itself.
% To address this issue, we take a slightly different approach and examine the recovered samples from each class, which provides a feasible approximation to the current class's decision region.
% % class $y_w$'s inverted samples $B_{y_w}$ obtained using Equation~\ref{eq:optim}
% Moreover, we tactfully distinguish the space occupied by normal and watermark samples via analysis on their \textit{exclusive salient neurons}.
\noindent$\bullet$\textbf{ Technical Designs.}
As the class-wise recovered samples $B_c$ might also contain important information of the normal samples in class $c$, in this section, we exploit the hypothesis of normal data dominance to split $B_c$ into two parts, namely proxy normal data and proxy watermark data.
We develop a neuron-level criterion: For each class $c$, we extract the salient neurons $SN(B_c)$, and take the subset achieving the largest contributions as proxy normal data, leaving the rest as proxy watermark data.
Algorithm~\ref{alg:split} in Appendix~\ref{app:algo} shows the detailed procedures.

Next, we formally discuss the underlying rationale.
An illustrative example is shown in Figure~\ref{fig:pilot_study}(c), where we split the recovered samples $B_6$ (the orange distribution) according to their contributions on $SN(B_6)$.
After splitting, the right part of $B_6$ are taken as proxy clean data, which are close to the normal samples of class 6 and preserve the normal functionalities specific to class 6.
The left part of $B_6$ are proxy watermark data, either correlating to the real watermark data, or noisily activating other classes' salient neurons (which are inessential for normal performance).

Given the recovered samples $\{B_c\}_{\ast}$ obtained via watermark recovering and target class detection, the algorithm will perform class-wise splitting on each batch $B_c$, returning proxy normal data $\{B_{c,nor}\}_\ast$ and proxy watermark data $\{B_{c,wmk}\}_\ast$.
Finally, we improve the finetuning loss accordingly.

If the target model is detected with a fixed-class watermark (i.e., the target class $s_0$), we only need to address samples from class $s_0$, denoted as $\{B_{c, wmk}\}_\ast = \{B_{s_0, wmk}\}$ and $\{B_{c, nor}\}_\ast = \{B_{s_0, nor}\}$.
In this way, the attacker can finetune the surrogate model $f_a$ under the following objective:
\begin{equation}
\begin{aligned}
    \min_{\theta_a} 
    &\underbrace{\left( \sum_{(x_i,y_i)\in X_{aux}} \mathcal{L}(f_a(x_i), y_i) \right) + 
    \sum_{\hat{x}_i \in B_{s_0, nor}} \mathcal{L}(f_a(\hat{x}_i), s_0)}_{\text{preserving original performance}} \\
    & + \underbrace{\sum_{\hat{x}_i \in B_{s_0, wmk}} \mathcal{L}(f_a(\hat{x}_i), \text{rand}_{s_0}),}_{\text{removing watermark}}
    \label{eq:unlearn_fix}
\end{aligned}
\end{equation}
where $\text{rand}_{s_0}$ denotes a random label in $\{0,1,\dots,C-1\}$ except $s_0$ itself.
This unlearning target is stronger than minimizing $KL(\cdot, y_{soft})$ in Equation~\ref{eq:unlearn}, due to the enhanced specificity brought by target class detection, and also better than maximizing $\mathcal{L}(\cdot, s_0)$, due to the easier convergence.
Note that in the objective function above, we additionally exploit the proxy normal data $B_{s_0,nor}$ to preserve $f_a$'s original performance.
% This is because the normal data dominance phenomenon is evident (Table~\ref{fig:neuron_cross}), and thus the salient neurons of the split proxy normal data will intersect little with those of the real watermark data.
This design enables our \textsc{Dehydra} to be extended to a data-free setting: when the auxiliary data is not available (i.e., the first term enclosed by $\left(\cdot\right)$ in the above equation is optional), we can still perform watermark unlearning only using the recovered samples.
% Note that the first term enclosed by $\left(\cdot\right)$ in the above equation is optional, which means $\mathcal{A}$ can even conduct this attack in a data-free setting.

For models detected with a non-fixed-class watermark, $\{B_{c, wmk}\}_\ast$ contains $C$ batches of proxy watermark data. 
Then $\mathcal{A}$ can finetune $f_a$ under this objective:
\begin{equation}
\begin{aligned}
    \min_{\theta_a}
    &\underbrace{\sum_{(x_i,y_i)\in X_{aux}} \mathcal{L}(f_a(x_i), y_i)}_{\text{preserving original performance}} + \\
    &\underbrace{\sum_{B_{c, wmk} \in \{B_{c, wmk}\}_\ast} \sum_{\hat{x}_i \in B_{c, wmk}} \mathcal{L}(f_a(\hat{x}_i), y_\star)
    ,}_{\text{removing watermark}}
    \label{eq:unlearn_nonfix}
\end{aligned}
\end{equation}
where $y_\star$ denotes the least-likely label $s_{C-1}$ (obtained during the target class detection process) for $B_{c, wmk} \in \{B_{c, wmk}\}_\ast$ and $c \neq s_{C-1}$, while the second-least-likely label $s_{C-2}$ when $c=s_{C-1}$.
This proxy label unlearning target is also stronger than that in Equation~\ref{eq:unlearn} due to the enhanced specificity, while maintaining the consistent descending gradients across the labels.
Different from the design for fixed-class watermarks, we do not use proxy normal data to further preserve the model performance.
This is because the normal data dominance phenomenon is less evident for non-fixed-class watermarks (Figure~\ref{fig:neuron_cross}), and thus the salient neurons of the split proxy normal data might still have some intersection with those of the real watermark data.
Therefore, leveraging these proxy normal data for preserving utility might hinder the watermark removal.

% \begin{equation}
% \begin{aligned}
%     &\min_{\theta_a} \sum_{(x_i,y_i)\in X_{aux}}{\mathcal{L}(f_a(x_i), y_i)} \\
%     &+\alpha_c \sum_{B_{c, nor} \in \{B_{c, nor}\}_\ast} \sum_{\hat{x}_i \in B_{c, nor}} \mathcal{L}(f_a(\hat{x}_i), c)  \\
%     &+\sum_{B_{c, wmk} \in \{B_{c, wmk}\}_\ast} \sum_{\hat{x}_i \in B_{c, wmk}} \mathcal{L}(f_a(\hat{x}_i), y_\star),
%     \label{eq:unlearn2}
% \end{aligned}
% \end{equation}
% % ----- %
% \begin{equation}
% \begin{aligned}
%     \min_{\theta_a} &\sum_{(x_i,y_i)\in X_{aux}} \mathcal{L}(f_a(x_i), y_i) +
%     \sum_{B_{t_w}, B_{t_c} \in \{B_{t_w}, B_{t_c}\}_\star} \\
%     &\left(
%     \sum_{\hat{x}_i \in B_{t_w}} \mathcal{L}(f_a(\hat{x}_i), y_{\star})+
%     \alpha_c \sum_{\hat{x}_i \in B_{t_c}} \mathcal{L}(f_a(\hat{x}_i), t) \right),
%     \label{eq:unlearn2}
% \end{aligned}
% \end{equation}

\section{Experiments}

\subsection{Overview of Evaluation} \label{sec:Experimental Settings}
To evaluate the performance of \textsc{Dehydra}, we perform a comprehensive study on ten mainstream black-box watermarking schemes, under various benchmark datasets, DNN architectures and data availability settings.
Before presenting the detailed evaluation results, we first provide a concise introduction to the experimental setups.

\noindent$\bullet$ \textbf{Datasets and Victim Models.}
Following the settings in~\cite{jia2021entangled, lee2022evaluating, Shafieinejad2019regularization, Aiken2020Laundering}, we construct victim models on three benchmark datasets, namely, MNIST~\cite{mnist}, CIFAR-10 and CIFAR-100~\cite{cifar10}. The respective DNN architectures are LeNet-5~\cite{lecun1998lenet}, ResNet-18~\cite{he2016res} and ResNet-34.
We embed black-box watermarks into the victim models during the training. Then we perform watermark removal attacks, including \textsc{Dehydra} and baseline attacks, to evaluate their effectiveness.

\noindent$\bullet$ \textbf{Target Watermark Schemes.}
Our evaluation covers ten mainstream black-box DNN watermarking schemes published at top-tier conferences (some of which are developed by industry leaders such as IBM): \textit{Content}, \textit{Noise}, \textit{Unrelated}~\cite{zhang2018protecting}, \textit{Piracy}~\cite{li2019piracy}, \textit{EWE}~\cite{jia2021entangled}, \textit{Adi}~\cite{adi2018turning}, \textit{AFS}~\cite{le2020afs}, \textit{EW}~\cite{namba2019ew}, \textit{Blind}~\cite{li2019blind}, \textit{Mark}~\cite{guo2018mark}.
These schemes feature diverse designs in the watermark data and target label settings, and are mostly evaluated in watermark robustness surveys \cite{lukas2022sok, lee2022evaluating, lederer2023sba}.
% are categorized as \textit{fixed-class watermarks} (i.e., all the watermark data are paired with the identical target class), and \textit{non-fixed-class watermarks} (i.e., each watermark data is paired with its own target class). 
For more backgrounds, please refer to Appendix~\ref{sec:back_impt_wms}.

During implementation, we strictly followed the specifications and hyperparameters in the original papers (with the target class set to 6 for those fixed-class watermarks) to prepare the watermarked models. We also referred to some recent open-source replications \cite{lederer2023sba, lukas2022sok} to ensure our implementation is faithful.
For more implementation details and the performance of the watermarked models, please refer to Appendix~\ref{app:impl-wm} and~\ref{app:target-models}.

\noindent$\bullet$ \textbf{Baseline Removal Attacks.}
We compare our \textsc{Dehydra} with the following three types of removal attacks, comprising six baseline attacks in total.
\begin{itemize}
    \item \textit{Pruning-based attacks:}
    \textbf{(1) \textit{Pruning}~\cite{zhang2018protecting, uchida2017embedding}} directly sets a proportion of DNN parameters with the smallest absolute values to zero.
    \textbf{(2) \textit{Fine-pruning}~\cite{liu2018finepruning}} prunes neurons that are infrequently activated by normal data, followed by a finetuning process.
    \item \textit{Finetuning-based attacks:}
    \textbf{(3) \textit{Finetuning}~\cite{libo2019leveraging, libo2021refit}} specifically finetunes the target model using a large learning rate, together with a carefully-designed scheduler.
    \textbf{(4) \textit{Regularization}~\cite{Shafieinejad2019regularization, lukas2022sok}} finetunes the target model with a large L2 regularization on the parameters.
    % {\textbf{(5) \sout{PST (Pattern embedding and Spatial-level Transformation)}~\cite{Guo2020PST}} adds an elaborately designed strong image pre-processing operation PST during the inference, after finetuning the model also with PST augmentation on in-distribution (PST-In) or transfer data (PST-out).}
    \textbf{(5) \textit{Distraction}~\cite{Zhong2022distraction}} finetunes the target model on in-distribution or transfer data, together with another set of lure data that distracts the model's attention away from the watermark.
    \item \textit{Unlearning-based attacks:} \textbf{(6) \textit{Laundering}~\cite{Aiken2020Laundering}} leverages trigger reverse-engineering methods~\cite{wang2019NC} in backdoor defense literature to recover watermark data, followed by neuron resetting and model retraining, to remove backdoor-based watermarks.
\end{itemize}
The implementation details of these baseline attacks are clarified in Appendix~\ref{app:removal}.
% {Among these baseline attacks, PST has a slightly different attack scenario, because the attacker can also manipulate the model inference stage.}

\noindent$\bullet$ \textbf{Dataset Availability Settings.}
We mainly focus on the realistic settings where the attacker has limited data access, including three different scenarios as follows:
\textbf{(1) In-distribution setting.} The attacker has 1000 correctly-labeled normal samples (2\% of the size of the CIFAR-10 training set), following~\cite{Zhong2022distraction}.
\textbf{(2) Transfer setting.} Here the attacker cannot access the source training data, but is assumed to have 2000 unlabeled samples from another distribution (e.g., CIFAR-100 for target models trained on CIFAR-10, following~\cite{Guo2020PST, Zhong2022distraction}).
\textbf{(3) Data-free setting.} In this setting, the attacker cannot access any samples.

\noindent$\bullet$ \textbf{Implementation of \textsc{Dehydra}.}
During watermark recovering, we set hyper-parameters $M=250, \alpha_{\ell_2}=0.01, \alpha_{tv}=0.03, \alpha_{bn}=0.1$. We use $tanh(\cdot)$ to constrain the batch data within a valid range, and optimize the recovering objective with the Adam optimizer of learning rate 0.1.
For target class detection, we set $\sigma_1=1.0, \sigma_2=0.03$ for MNIST and $\sigma_1=0.5, \sigma_2=0.015$ for CIFAR-10 and CIFAR-100 respectively.
The detection threshold $T$ is conservatively set to $0.4$ for MNIST and CIFAR-10, while $0.3$ for CIFAR-100.
For splitting the recovered samples, we set $l$ to be the penultimate layer of the target model, and saliency ratio $\beta=0.95$, split ratio $\gamma=0.5$ in Alg.\ref{alg:split}.
During watermark unlearning, we set $\alpha_{KL}=15$ to ensure the unlearning strength for the basic \textsc{Dehydra} and the uniform loss weights for the improved \textsc{Dehydra}.
The model is finally finetuned for $10$ epochs using the SGD optimizer with the batch size $128$.
The learning rate is set to $0.01$ when auxiliary data is available and $0.003$ when data-free. Noteworthily, the settings above are empirically chosen and lead to strong attack effectiveness uniformly over almost all the covered watermarking schemes, which supports the watermark-agnostic nature of our \textsc{Dehydra}.

\begin{table*}[!tb]
\caption{Comparison of removal attacks against black-box DNN watermarks under the in-distribution setting on CIFAR-10.
x / y denotes the clean accuracy / rescaled watermark accuracy, and values behind ± report their standard deviations respectively in 5 repetitive tests.
The rescaled watermark accuracy values below 50\% are bolded and underlined.
The left five columns are attack results of the fixed-class watermarks, while the right five are of the non-fixed-class watermarks.
}
\label{tab:in_distribution}
\vspace{-0.05in}
\centering
\small
\setlength{\tabcolsep}{2.5pt}
\begin{tabular}{lcccccccccc}
\hline
\textbf{Attacks} & \textbf{\textit{Content}} & \textbf{\textit{Noise}} & \textbf{\textit{Unrelated}} & \textbf{\textit{Piracy}} & \textbf{\textit{EWE}} & \textbf{\textit{Adi}} & \textbf{\textit{AFS}} & \textbf{\textit{EW}} & \textbf{\textit{Blind}} & \textbf{\textit{Mark}} \\
\hline
\textit{None}    & 93.8 / 100.0 & 94.1 / 100.0& 93.7 / 100.0& 93.3 / 100.0 & 94.0 / 100.0& 93.9 / 100.0& 93.0 / 100.0& 92.9 / 100.0& 91.2 / 100.0 & 94.0 / 100.0 \\
\hline
\textit{Fine-pruning} & 86.4 / \underline{\textbf{44.7}}& 87.4 / 67.0& 85.9 / 52.6& 69.5 / \underline{\textbf{37.8}}& 88.9 / 70.2& 85.4 / 59.8& 87.3 / 80.4& 87.6 / 83.3& 77.0 / 53.8& 87.3 / 61.1 \\
% \multirow{2}{*}{Fine-pruning} & 86.4 / \underline{\textbf{44.7}}& 87.4 / 67.0& 85.9 / 52.6& 69.5 / \underline{\textbf{37.8}}& 88.9 / 70.2& 85.4 / 59.8& 87.3 / 80.4& 87.6 / 83.3& 77.0 / 53.8& 87.3 / 61.1 \\
                              % & ±0.7 / ±0.3& ±0.3 / ±2.8& ±0.3 / ±2.9& ±0.9 / ±2.3& ±0.1 / ±0.9& ±0.1 / ±0.5& ±0.1 / ±1.4& ±0.1 / ±1.0& ±0.1 / ±0.3& ±0.2 / ±1.9 \\
% \specialrule{0em}{2pt}{1pt}
\textit{Finetuning}   & 85.5 / \underline{\textbf{46.4}}& 85.1 / 84.9& 84.6 / \underline{\textbf{47.1}}& 66.1 / \underline{\textbf{40.9}}& 86.9 / 81.2& 85.0 / 56.3& 90.4 / 99.4& 90.1 / 98.8& 79.9 / 55.0& 85.8 / 57.1 \\
% \multirow{2}{*}{Finetuning}   & 85.5 / \underline{\textbf{46.4}}& 85.1 / 84.9& 84.6 / \underline{\textbf{47.1}}& 66.1 / \underline{\textbf{40.9}}& 86.9 / 81.2& 85.0 / 56.3& 90.4 / 99.4& 90.1 / 98.8& 79.9 / 55.0& 85.8 / 57.1 \\
%                               & ±0.4 / ±0.3& ±0.3 / ±3.5& ±0.6 / ±0.7& ±1.0 / ±1.1& ±0.1 / ±3.3& ±0.3 / ±0.7& ±0.3 / ±0.3& ±0.2 / ±0.3& ±0.7 / ±3.4& ±0.4 / ±1.7 \\
% \specialrule{0em}{2pt}{1pt}
\textit{Regularization}   & 60.3 / 100.0& 54.7 / 91.5& 82.8 / 100.0& 76.9 / 100.0& 68.2 / \underline{\textbf{47.6}}& 55.4 / 73.8& 72.3 / 81.6& 69.0 / 82.1& 69.8 / 83.6& 75.7 / 87.2 \\
% \multirow{2}{*}{Regularization}   & 60.3 / 100.0& 54.7 / 91.5& 82.8 / 100.0& 76.9 / 100.0& 68.2 / \underline{\textbf{47.6}}& 55.4 / 73.8& 72.3 / 81.6& 69.0 / 82.1& 69.8 / 83.6& 75.7 / 87.2 \\
%                               & ±1.3 / ±0.3& ±1.2 / ±1.9& ±0.3 / ±0.3& ±1.9 / ±1.5& ±1.8 / ±0.2& ±2.5 / ±4.0& ±1.8 / ±2.9& ±1.5 / ±1.2& ±2.4 / ±2.7& ±1.5 / ±1.2 \\
% \specialrule{0em}{2pt}{1pt}
\textit{Distraction}   & 87.4 / \underline{\textbf{47.5}}& 89.5 / 98.1& 88.5 / 51.5& 70.6 / \underline{\textbf{40.9}}& 85.8 / 52.4& 89.4 / 52.2& 92.5 / 99.4& 92.5 / 100.0& 89.8 / 91.2& 86.0 / 51.3 \\
% \multirow{2}{*}{Distraction}   & 87.4 / \underline{\textbf{47.5}}& 89.5 / 98.1& 88.5 / 51.5& 70.6 / \underline{\textbf{40.9}}& 85.8 / 52.4& 89.4 / 52.2& 92.5 / 99.4& 92.5 / 100.0& 89.8 / 91.2& 86.0 / 51.3 \\
%                               & ±1.2 / ±0.9& ±0.2 / ±0.4& ±0.6 / ±2.3& ±4.2 / ±1.6& ±1.0 / ±1.7& ±0.4 / ±1.9& ±0.1 / ±0.3& ±0.1 / ±0.0& ±0.3 / ±1.2& ±1.8 / ±3.3 \\
% \specialrule{0em}{2pt}{1pt}
\textit{Laundering}   & 91.1 / \underline{\textbf{48.1}}& 90.7 / 77.4& 88.8 / 100.0& 83.8 / 69.9& 91.1 / 70.2& 87.0 / 69.7& 89.3 / 100.0& 89.3 / 98.8& 86.3 / 73.1& 91.1 / 94.8 \\
% \multirow{2}{*}{Laundering}   & 91.1 / \underline{\textbf{48.1}}& 90.7 / 77.4& 88.8 / 100.0& 83.8 / 69.9& 91.1 / 70.2& 87.0 / 69.7& 89.3 / 100.0& 89.3 / 98.8& 86.3 / 73.1& 91.1 / 94.8 \\
                              % & ±0.4 / ±0.3& ±0.3 / ±4.7& ±0.4 / ±0.0& ±1.1 / ±0.8& ±0.1 / ±1.1& ±0.3 / ±1.9& ±0.2 / ±0.0& ±0.2 / ±0.5& ±1.0 / ±2.9& ±0.5 / ±2.1 \\
\hline
\textsc{Dehydra}\textsubscript{Basic}  & 88.8 / \underline{\textbf{44.7}}& 88.1 / \underline{\textbf{5.7}}& 88.1 / \underline{\textbf{44.9}}& 83.3 / \underline{\textbf{44.0}}& 89.0 / \underline{\textbf{48.2}}& 84.1 / 51.6& 86.4 / \underline{\textbf{44.9}}& 85.2 / 52.1& 84.1 / 58.5& 83.0 / 50.7 \\

% \multirow{2}{*}{Dehydra (Basic)}   & 88.8 / \underline{\textbf{44.7}}& 88.1 / \underline{\textbf{5.7}}& 88.1 / \underline{\textbf{44.9}}& 83.3 / \underline{\textbf{44.0}}& 89.0 / \underline{\textbf{48.2}}& 84.1 / 51.6& 86.4 / \underline{\textbf{44.9}}& 85.2 / 52.1& 84.1 / 58.5& 83.0 / 50.7 \\

%                               & ±0.9 / ±0.5& ±1.0 / ±0.4& ±0.1 / ±0.5& ±0.4 / ±0.6& ±1.0 / ±0.4& ±0.4 / ±0.5& ±0.2 / ±0.6& ±0.0 / ±0.7& ±0.6 / ±1.0& ±0.1 / ±1.4 \\
% \specialrule{0em}{2pt}{1pt}
\textsc{Dehydra}\textsubscript{Improved}  & 93.1 / \underline{\textbf{45.8}}& 93.1 / \underline{\textbf{5.7}}& 92.7 / \underline{\textbf{44.9}}& 91.4 / \underline{\textbf{41.5}}& 93.3 / \underline{\textbf{47.6}}& 86.0 / \underline{\textbf{49.3}}& 88.1 / \underline{\textbf{44.3}}& 88.2 / \underline{\textbf{46.3}}& 90.1 / \underline{\textbf{46.2}}& 88.5 / \underline{\textbf{47.2}} \\
% \multirow{2}{*}{Dehydra (Improved)}   & 93.1 / \underline{\textbf{45.8}}& 93.1 / \underline{\textbf{5.7}}& 92.7 / \underline{\textbf{44.9}}& 91.4 / \underline{\textbf{41.5}}& 93.3 / \underline{\textbf{47.6}}& 86.0 / \underline{\textbf{49.3}}& 88.1 / \underline{\textbf{44.3}}& 88.2 / \underline{\textbf{46.3}}& 90.1 / \underline{\textbf{46.2}}& 88.5 / \underline{\textbf{47.2}} \\

                              % & ±0.4 / ±0.5& ±0.2 / ±0.9& ±0.4 / ±0.0& ±0.9 / ±0.8& ±1.4 / ±0.5& ±0.5 / ±0.5& ±0.7 / ±0.5& ±0.9 / ±0.8& ±0.2 / ±0.3& ±0.9 / ±0.5 \\
\hline
\end{tabular}
% \footnotesize\raggedright{*The baseline \pxd{in this scenario has slightly different settings}.}
\vspace{-0.05in}
\end{table*}

\noindent$\bullet$ \textbf{Evaluation Metrics.}
We primarily focus on the utility and the watermark retention of the surrogate models obtained via removal attacks. Specifically, we use \textit{clean accuracy}, i.e., the classification accuracy on the test set, to measure the utility, and \textit{watermark accuracy}, i.e., the ratio of watermark samples correctly classified as target labels, to measure the watermark retention.
% To compare the attack effectiveness over different black-box watermark schemes, we follow Lukas et al.~\cite{lukas2022sok} to determine the decision threshold $\theta$ for each watermark on $20$ independently trained clean models, and then calculate a metric called the \textit{rescaled watermark accuracy}.
Further, we follow Lukas et al.~\cite{lukas2022sok} to determine the decision threshold $\theta$ for each watermark on $20$ clean models, and then calculate a metric $S(\cdot;\theta)$ called the \textit{rescaled watermark accuracy}.
The former $\theta$ is used to distinguish the watermark accuracy of the watermarked model from those independently trained models with high confidence, while the latter $S(\cdot;\theta)$ enables a fair robustness comparison over different watermarking schemes.
Specially, the rescaling function is defined as
$S(x;\theta)=\max (0, \frac{1-\theta^\prime}{1-\theta}x+\frac{\theta^\prime-\theta}{1-\theta})$,
linearly rescaling the watermark accuracy $x$ based on the decision threshold $\theta$.
By manually picking a rescaled decision threshold $\theta^\prime=0.5$, a watermark is said to be removed
% we can say a watermark is fully removed
% which degrades to $S(x;\theta)=\max (0, 1-0.5\frac{1-x}{1-\theta})$ when we set $\theta^\prime=0.5$.
if the rescaled watermark accuracy is lower than $50\%$.
Table~\ref{tab:target_models} in Appendix~\ref{app:target-models} summarizes our estimated decision threshold for the ten black-box watermarking schemes.

We define a removal attack as successful (i.e., cracks the target watermark scheme) if the rescaled watermark accuracy is below 50\% and the surrogate model maintains at least 90\% of the original utility.
Note that this criterion is slightly different from~\cite{lukas2022sok} where a maximum $5\%$ utility loss is used.
This is mainly because we consider the more realistic data-limited settings (e.g., $2\%$ of the source training set in most experiments), while~\cite{lukas2022sok} assume the attacker has access to over $30\%$ of the training set.
Nevertheless, we also evaluate the attack performance when more data is available in Section \ref{sec:discussion}.

% 实验设置：参考《On Function-Coupled Watermarks for Deep Neural Networks》
    % 数据集：MNIST, CIFAR-10, CIFAR-100, Tiny-ImageNet
    % 模型结构：LeNet-5, VGG16, ResNet-18
    % 训练水印比例：小于1%训练集
    % 水印移除相关数据
        % 1) pruning：80%的水印acc，90%的clean acc（画了0-90折线图）
        % 2）Fine-tuning：10epoch（也画了10以内折线图）；学习率：1e-4 to 1e-5
        % 3）Transfer learning：10epoch（也画了10以内折线图）；三组：CIFAR-10 to CIFAR-100, CIFAR-10 to MNIST,  CIFAR-100 to Tiny-mageNet
        % 4）Overwriting

\subsection{Attack Performance} \label{sec:attack-performance}
The comparison with baseline removal attacks is organized according to the data settings where the baseline attacks are applicable. Due to the page limit, we mainly present and analyze the attack results on CIFAR-10, while similar results are observed on MNIST and CIFAR-100. Appendix~\ref{app:atk-mnist} and~\ref{app:atk-cifar100} present the omitted results.

% 攻防实验——黑盒水印 vs. 移除攻击
% （Clean/Watermark acc）
% Basic版本：	  labeled / transfer 场景		（without split）
% Improved版本：labeled / transfer / data-free场景	（with split）

% \subsubsection{In-distribution Setting}\label{sec:labeled}

\noindent\textbf{(1) In-distribution Setting.}
Table~\ref{tab:in_distribution} presents the attack results by performing different removal attacks against the ten investigated watermarks under the in-distribution setting.
% \sout{We first evaluate the attack performance of our \textsc{Dehydra} under the in-distribution data setting.}
% , the attacker is assumed to have 1000 normal samples from the source training set (i.e. 2\% of the CIFAR-10 training set).
% Table \ref{tab:in_distribution} presents the clean accuracy and the rescaled watermark accuracy of the surrogate models, after performing different removal attacks on the victim model protected by one of the ten black-box watermarks.
% \noindent \textbf{Results \& Analysis.}
Both the basic \textsc{Dehydra} and improved \textsc{Dehydra} significantly lower the rescaled watermark accuracy of surrogate models, among which the improved \textsc{Dehydra} achieves comprehensive watermark removal against the ten investigated black-box watermarks, with at least 90\% utility preserved (six in ten even with a clean accuracy degradation in 2\%).

The baseline removal attacks are only effective against a subset of watermark schemes.
For instance, \textit{Laundering} is more effective against watermarks similar to backdoors, such as \textit{Content}, \textit{Piracy} and \textit{EWE}, but less effective against other watermarks, because of its dependence on the trigger reverse-engineering method to detect and remove the underlying watermarks.
\textit{Fine-pruning} and \textit{Finetuning} are less effective against \textit{EWE}, \textit{Adi} and \textit{AFS}, possibly due to the more entangled learned representations of the watermark data.
\textit{Distraction} can also only compromise about half of these watermarks' robustness, while being less powerful against watermarks of other forms.
In contrast, our basic and improved \textsc{Dehydra} are both effective regardless of the target watermark schemes, due to the capture of the black-box watermarks' commonality.
Note that our improved \textsc{Dehydra} can effectively remove the \textit{Piracy} watermark, despite the fact that the number of its watermark data (1\% of the training data, i.e., 500, following the original paper) is much larger than the number of unlearned samples in our improved \textsc{Dehydra} ($\gamma\times M=125$).
This indicates that the removal effect of our approach is achieved not just through a sample-level unlearning, but also through a deeper neuron-level fixing.

Additionally, these baseline attacks might impact the surrogate model's utility seriously.
For instance, the finetuning-based attacks, such as \textit{Finetuning}, \textit{Regularization} and \textit{Distraction}, incur a large clean accuracy degradation against \textit{Noise} and \textit{Piracy} watermarks, possibly due to their overly offensive finetuning strategies.
However, our \textsc{Dehydra} attacks could still preserve the surrogate model's performance besides removing the watermark, because of their more specific training (unlearning) objectives.

% \sout{We further compare the attack results of our basic \textsc{Dehydra} and improved \textsc{Dehydra} separately in Figure~\ref{fig:atk_basic_improve}.
% , due to the enhanced removal specificity brought by target class detection and recovered samples splitting algorithms.}

% \begin{figure}[t]
% \centering
% \includegraphics[width=0.9\linewidth]{figs/Attack_effect.pdf}
% \vspace{-0.1in}
% \caption{\sout{Comparison of the basic and the improved \textsc{Dehydra} under the in-distribution setting on CIFAR-10.}}
% \label{fig:atk_basic_improve}
% \vspace{-0.2in}
% \end{figure}

\begin{table*}[htb]
\caption{Comparison of removal attacks against black-box watermarks under the transfer setting on CIFAR-10.
% x/y denotes clean accuracy/rescaled watermark accuracy, and values behind ± report the standard deviation in 5 repetitive tests.
}
\label{tab:transfer}
\vspace{-0.05in}
\centering
\small
\setlength{\tabcolsep}{2.5pt}
\centering

\begin{tabular}{lcccccccccc}
\hline
\textbf{Attacks} & \textbf{\textit{Content}} & \textbf{\textit{Noise}} & \textbf{\textit{Unrelated}} & \textbf{\textit{Piracy}} & \textbf{\textit{EWE}} & \textbf{\textit{Adi}} & \textbf{\textit{AFS}} & \textbf{\textit{EW}} & \textbf{\textit{Blind}} & \textbf{\textit{Mark}} \\
\hline
    \textit{None}    & 93.8 / 100.0& 94.1 / 100.0& 93.7 / 100.0& 93.3 / 100.0& 94.0 / 100.0& 93.9 / 100.0& 93.0 / 100.0& 92.9 / 100.0& 91.2 / 100.0& 94.0 / 100.0 \\ 
\hline
    \textit{Regularization}
            & 49.3 / 99.4& 60.3 / 73.6& 53.1 / 100.0& 60.5 / 100.0& 53.8 / \underline{\textbf{47.6}}& 41.0 / 69.7& 43.4 / 71.5& 55.5 / 76.9& 44.6 / 69.6& 53.4 / 74.5 \\
    % \multirow{2}{*}{Regularization}
    %         & 49.3 / 99.4& 60.3 / 73.6& 53.1 / 100.0& 60.5 / 100.0& 53.8 / \underline{\textbf{47.6}}& 41.0 / 69.7& 43.4 / 71.5& 55.5 / 76.9& 44.6 / 69.6& 53.4 / 74.5 \\

%             & ±1.4 / ±0.3& ±3.6 / ±7.8& ±8.3 / ±0.5& ±4.6 / ±0.0& ±0.3 / ±4.7& ±2.8 / ±1.6& ±2.8 / ±4.9& ±6.5 / ±6.1& ±2.6 / ±3.8& ±1.3 / ±0.3 \\
% \specialrule{0em}{2pt}{1pt}
\textit{Distraction}
            & 89.0 / 50.9& 89.6 / 100.0& 88.7 / \underline{\textbf{49.8}}& 84.2 / \underline{\textbf{49.0}}& 88.8 / 63.4& 87.7 / 60.3& 92.0 / 100.0& 92.3 / 100.0& 88.8 / 91.2& 89.1 / 66.9 \\
% \multirow{2}{*}{Distraction}   
%             & 89.0 / 50.9& 89.6 / 100.0& 88.7 / \underline{\textbf{49.8}}& 84.2 / \underline{\textbf{49.0}}& 88.8 / 63.4& 87.7 / 60.3& 92.0 / 100.0& 92.3 / 100.0& 88.8 / 91.2& 89.1 / 66.9 \\
%             & ±0.2 / ±0.5& ±0.6 / ±1.2& ±0.4 / ±1.7& ±0.3 / ±0.9& ±0.7 / ±1.9& ±0.3 / ±3.0& ±0.2 / ±0.0& ±0.5 / ±0.0& ±0.6 / ±3.6& ±0.3 / ±2.6 \\
\hline
\textsc{Dehydra}\textsubscript{Basic}
            & 85.5 / \underline{\textbf{49.8}}& 85.3 / \underline{\textbf{6.7}}& 80.0 / \underline{\textbf{46.0}}& 82.4 / \underline{\textbf{45.2}}& 88.2 / 50.8& 80.8 / 54.5& 84.6 / \underline{\textbf{48.1}}& 82.1 / 58.5& 80.2 / 62.0& 84.0 / 55.9 \\
% \multirow{2}{*}{Dehydra (Basic)}   
%             & 85.5 / \underline{\textbf{49.8}}& 85.3 / \underline{\textbf{6.7}}& 80.0 / \underline{\textbf{46.0}}& 82.4 / \underline{\textbf{45.2}}& 88.2 / 50.8& 80.8 / 54.5& 84.6 / \underline{\textbf{48.1}}& 82.1 / 58.5& 80.2 / 62.0& 84.0 / 55.9 \\
%             & ±1.0 / ±0.9& ±1.1 / ±0.4& ±0.5 / ±0.3& ±0.4 / ±0.6& ±0.8 / ±0.5& ±0.7 / ±0.5& ±1.1 / ±0.6& ±1.4 / ±1.0& ±0.4 / ±1.0& ±1.2 / ±0.7 \\
% \specialrule{0em}{2pt}{1pt}
\textsc{Dehydra}\textsubscript{Improved}
            & 92.5 / \underline{\textbf{47.0}}& 92.3 / \underline{\textbf{5.7}}& 92.7 / \underline{\textbf{48.2}}& 90.1 / \underline{\textbf{47.1}}& 92.6 / \underline{\textbf{49.7}}& 80.4 / 59.8& 86.2 / \underline{\textbf{47.4}}& 83.5 / \underline{\textbf{49.8}}& 88.8 / \underline{\textbf{46.8}}& 87.0 / \underline{\textbf{49.5}} \\         

% \multirow{2}{*}{Dehydra (Improved)}   
%             & 92.5 / \underline{\textbf{47.0}}& 92.3 / \underline{\textbf{5.7}}& 92.7 / \underline{\textbf{48.2}}& 90.1 / \underline{\textbf{47.1}}& 92.6 / \underline{\textbf{49.7}}& 80.4 / 59.8& 86.2 / \underline{\textbf{47.4}}& 83.5 / \underline{\textbf{49.8}}& 88.8 / \underline{\textbf{46.8}}& 87.0 / \underline{\textbf{49.5}} \\         
%             & ±0.2 / ±0.0& ±0.3 / ±0.0& ±0.3 / ±0.3& ±0.8 / ±0.8& ±0.3 / ±0.5& ±0.4 / ±1.9& ±0.7 / ±0.3& ±0.8 / ±0.3& ±1.5 / ±0.7& ±0.5 / ±0.8 \\
\hline
\end{tabular}
% \footnotesize\raggedright{*The PST-Out attack has a slightly different scenario, as described in Table \ref{tab:in_distribution}.}
\vspace{-0.1in}
\end{table*}

\noindent\textbf{(2) Transfer Setting.}
As we can see from Table~\ref{tab:transfer}, 
% Table~\ref{tab:transfer} presents the attack results under this setting.
% \noindent \textbf{Results \& Analysis.}
the clean accuracy of the surrogate models are generally lower than the in-distribution setting, since the proxy labels of these transfer data may contain noise.
\textit{Regularization} is generally ineffective, and \textit{Distraction} is only effective against a subset of fixed-class watermarks.

On the contrary, both the basic and improved \textsc{Dehydra} significantly lower the rescaled watermark accuracy while preserving the clean accuracy at an acceptable level.
The improved \textsc{Dehydra} still successfully cracks nine of the ten watermarks under this setting.
Note that our \textsc{Dehydra} cannot completely remove the \textit{Adi} watermark (despite the rescaled watermark accuracy being lower than 60\%) under the transfer setting.
This is reasonable because \textit{Adi} uses abstract out-of-distribution images as watermark data, whose salient neurons are more likely to be activated by the transfer dataset, compared to other common watermarks.
% Despite this, our \textsc{Dehydra} still significantly reduces the rescaled watermark accuracy, 

% Despite the improved \textsc{Dehydra} mitigates the impact on the clean accuracy, there still exists minor degradation of utility on non-fixed-class watermarks, such as Adi.
% This is reasonable because \lyf{Adi uses abstract images as watermark data, which is completely different from the original task.
% The transfer dataset is more likely to activate watermark samples' salient neurons in Adi, compared to other watermarks.
% This brings conflict learning objective during \textsc{Dehydra}'s finetuning process.}

\noindent\textbf{(3) Data-free Setting.}
Almost no baseline attacks except for the \textit{Pruning} attack are applicable to the data-free scenario. Therefore, as also explained in Section~\ref{sec:unlearn2}, we mainly compare the improved \textsc{Dehydra}, when attacking five fixed-class watermarks here, with \textit{Pruning}.
The attack results are shown in Table~\ref{tab:data_free}.

\begin{table}[t]
    \caption{Comparison of removal attacks against black-box watermarks under the data-free setting on CIFAR-10.}
    \label{tab:data_free}
    \vspace{-0.05in}
    \centering
    \footnotesize
    \setlength{\tabcolsep}{1.5pt}
    \centering
    \resizebox{\linewidth}{!}{
        \begin{tabular}{cccccc}
        \hline
        \textbf{Attacks} & \textbf{\textit{Content}} & \textbf{\textit{Noise}} & \textbf{\textit{Unrelated}} & \textbf{\textit{Piracy}} & \textbf{\textit{EWE}} \\
        \hline
            \textit{None}        & 93.8 / 100.0& 94.1 / 100.0& 93.7 / 100.0& 93.3 / 100.0& 94.0 / 100.0 \\
        \hline
            \textit{Pruning}     & 81.5 / 78.2& 90.7 / \underline{\textbf{42.5}}& 78.7 / \underline{\textbf{45.4}}& 41.9 / \underline{\textbf{37.8}}& 92.2 / 100.0 \\
        \hline
        
        \makecell{\textsc{Dehydra}\\(Improved)}
                        & \makecell{89.8 / \underline{\textbf{45.8}}} & \makecell{88.6 / \underline{\textbf{5.7}}} & \makecell{88.8 / \underline{\textbf{46.0}}} & \makecell{90.0 / \underline{\textbf{39.6}}} & \makecell{90.7 / \underline{\textbf{48.2}}} \\

        % \multirow{2}{*}{\shortstack[c]{MIRA \\ (Improved)}} 
        %                 & 89.8 / \underline{\textbf{45.8}}& 88.6 / \underline{\textbf{5.7}}& 88.8 / \underline{\textbf{46.0}}& 90.0 / \underline{\textbf{39.6}}& 90.7 / \underline{\textbf{48.2}} \\
                        % & ±0.4 / ±0.5& ±0.7 / ±0.4& ±0.5 / ±0.5& ±1.0 / ±0.3& ±0.7 / ±0.2 \\
        
        \hline
        \end{tabular}
    }
    \vspace{-0.1in}
\end{table}

As is shown, our improved \textsc{Dehydra} achieves surprisingly good performance in the data-free setting, significantly surpassing the baseline \textit{Pruning}.
All the five fixed-class watermarks are completely removed from the protected model, with a clean accuracy degradation of no more than 5.5\%.
This validates the correctness of our splitting algorithm on the recovered samples, since the surrogate modes are finetuned solely on the recovered samples $\{B_c\}_{\ast}=\{B_{s_0}\}$, with split proxy normal data $\{B_{s_0, nor}\}$ to maintain the utility, and proxy watermark data $\{B_{s_0, wmk}\}$ for removal, as in Equation~\ref{eq:unlearn_fix}.

\subsection{Validating Detection of Target Classes} \label{sec:exp-target-detection}
% All the five fixed-class watermarks and five non-fixed-class watermarks in Section~\ref{sec:attack-performance} are successfully identified respectively, with correct detected target class $s_0=6$ for fixed-class watermarks.
In this part, we mainly study the risks of target class detection errors.
We temporarily ignore the case when fixed-class watermarks are correctly identified but given the incorrect target label, since this is extremely rare, as validated by later quantitative experiments.
Here we mainly focus on two types of errors, false positives (non-fixed-class watermarks identified as fixed-class ones) and false negatives (fixed-class watermarks identified as non-fixed-class ones).

Intuitively, we have a lower tolerance for false positives compared to false negatives.
As shown in Table~\ref{tab:fix-non-fix}, for fixed-class watermarks (\textit{Content} and \textit{Piracy}) wrongly identified as non-fixed-class ones, the clean accuracy will degrade a little (e.g., within 3\% on CIFAR-10) and the watermark is still effectively removed.
However, for non-fixed-class watermarks (\textit{Adi} and \textit{AFS}) falsely identified as fixed-class ones, the removal effectiveness will be significantly hampered (e.g, the rescaled watermark accuracy remains over 90\%).
Therefore, we should set the detection threshold conservatively to prevent false positive cases.

\begin{table}[!tb]
\caption{Attack results of improved \textsc{Dehydra} on CIFAR-10 and CIFAR-100, with the target class detection results deliberately \textit{resp}. set to fixed-class and non-fixed-class ones.}
\label{tab:fix-non-fix}
\centering
\vspace{-0.05in}
\small
\setlength{\tabcolsep}{2.5pt}
\begin{tabular}{lc cccc}
\hline
\textbf{Dataset} & \textbf{Detection} & \textbf{\textit{Content}} & \textbf{\textit{Piracy}} & \textbf{\textit{Adi}} & \textbf{\textit{AFS}} \\
\hline
\multirow{2}{*}{CIFAR-10} & fixed & 93.1 / \underline{\textbf{45.8}}& 91.4 / \underline{\textbf{41.5}}& 92.0 / 96.5& 88.8 / 93.7 \\
            & non-fixed & 90.7 / \underline{\textbf{44.2}}& 88.9 / \underline{\textbf{40.3}}& 86.0 / \underline{\textbf{49.3}}& 88.1 / \underline{\textbf{44.3}} \\

\hline
\multirow{2}{*}{CIFAR-100} & fixed & 75.1 / \underline{\textbf{48.7}}& 71.0 / \underline{\textbf{41.8}}& 64.0 / 71.5& 67.1 / 100.0 \\

 & non-fixed & 67.7 / \underline{\textbf{48.7}}& 65.8 / \underline{\textbf{40.0}}& 64.8 / 60.3& 65.4 / 54.8 \\

\hline
\end{tabular}
\vspace{-0.1in}
\end{table}

\begin{figure}[t]
\centering
% \vspace{-0.1in}
\includegraphics[width=1\linewidth]{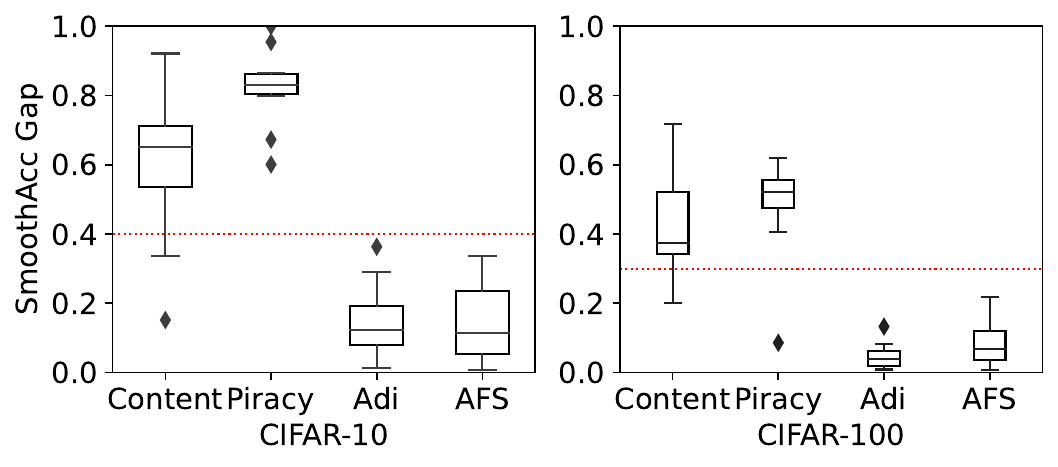}
\vspace{-0.2in}
\caption{SmoothAcc gap distribution under ten random tests.
The box plot shows min/max and quartiles, as well as the estimated outliers.
}
\label{fig:boxplot}
\vspace{-0.1in}
\end{figure}

Next, we further investigate the detection accuracy of our target class detection algorithm.
For each watermark scheme on each dataset, we independently train ten watermarked models, each using the randomly re-generated watermark data with different target classes.
For each target watermarked model, we perform watermark recovering and then target class detection.
We calculate the SmoothAcc gap $SmoothAcc(B_{y_w}) - \max_{i\neq y_w} SmoothAcc(B_i)$ for a fixed-class watermark with target class $y_w$, and $SmoothAcc(B_{s_1}) - SmoothAcc(B_{s_2})$ for a non-fixed-class watermark. As shown in Figure~\ref{fig:boxplot}, the SmoothAcc gaps for fixed-class watermarks are always positive in the ten tests, which means we can always determine a correct target class for those successfully identified fixed-class watermarks. Besides, the SmoothAcc gaps of fixed-class and non-fixed-class watermarks are large enough to find a threshold to determine them, which validates the correctness of our target class detection algorithm.

% Moreover, the difference of the SmoothAcc gaps between fixed-class and non-fixed-class watermarks is obvious, which further validates the effectiveness of our target class detection algorithm.
% On CIFAR-10, using the conservatively set decision threshold $T=0.4$, the false negative rates for \textit{Content} and Piracy are 20\% and 0\%, and the false positive rates for Adi and AFS are both 0\%.
% On CIFAR-100, using $T=0.3$, the false negative rates for \textit{Content} and Piracy are 10\% and 10\%, and the false positive rates for Adi and AFS are both 0\%.

\subsection{Ablation Studies} \label{sec:Ablation}

\subsubsection{Effectiveness of the Two Improvements} \label{sec:ablation_two}

We further validate the effectiveness of the two proposed improvements. As shown in Table \ref{tab:ablation}, target class detection is especially useful to fixed-class watermarks. It is also useful to non-fixed-class watermarks, due to the easier unlearning target of proxy hard labels.
The splitting algorithm is generally beneficial to both fixed-class and non-fixed-class watermarks to preserve normal performance, due to the enhanced removal specificity.
Armed with the two designs, the improved \textsc{Dehydra} achieves the highest clean accuracy, with a 4\% increase in average compared with the basic version.

\begin{table}[bt]
\caption{Ablation study on the two improvements, target class detection and recovered samples splitting, under the in-distribution setting on CIFAR-10.}
\label{tab:ablation}
\vspace{-0.05in}
\centering
\small
\setlength{\tabcolsep}{1.5pt}
\begin{tabular}{c cccc}
\hline
\textbf{Attacks} & \textbf{\textit{Unrelated}} & \textbf{\textit{EWE}} & \textbf{\textit{EW}} & \textbf{\textit{Mark}}\\
\hline
\textsc{Dehydra}\textsubscript{Basic} & 88.1 / \textbf{44.9} & 89.0 / \textbf{48.2} & 85.2 / 52.1 & 83.0 / 50.7 \\
\hline
\textsc{Dehydra} + Detection	& 91.9 / \underline{\textbf{44.9}}	& 92.4 / \underline{\textbf{47.6}}	& 86.8 / \underline{\textbf{46.9}}	& 87.8 / \underline{\textbf{48.4}}	 \\
\textsc{Dehydra} + Splitting 	& 88.3 / \underline{\textbf{44.9}}	& 89.6 / \underline{\textbf{47.6}}	& 87.2 / 51.5	& 85.7 / \underline{\textbf{48.9}}	 \\
\hline
\textsc{Dehydra}\textsubscript{Improved} & 92.7 / \textbf{44.9} & 93.3 / \textbf{47.6} & 88.2 / \textbf{46.3} & 88.5 / \textbf{47.2} \\

% \makecell{\textsc{MIRA}\\(Improved)}
%                 & \makecell{89.8 / \underline{\textbf{45.8}}} & \makecell{88.6 / \underline{\textbf{5.7}}} & \makecell{88.8 / \underline{\textbf{46.0}}} & \makecell{90.0 / \underline{\textbf{39.6}}} & \makecell{90.7 / \underline{\textbf{48.2}}} \\

\hline
\end{tabular}
\vspace{-0.1in}
\end{table}

\subsubsection{Comparison with Other Recovering Methods} \label{sec:inv_comp}
We further prove the effectiveness of the watermark recovering algorithm of \textsc{Dehydra}, by comparing with representative algorithms from \textbf{(1) Trigger Reverse-Engineering} (e.g., an adaptive version of \textit{NeuralCleanse}~\cite{wang2019NC}), which is originally designed to detect and erase backdoors in DNNs,
and \textbf{(2) General Model Inversion} (e.g., \textit{DeepInversion}~\cite{yin2020dreaming}), which is designed to synthesize realistic training samples of the target models for privacy disclosure or knowledge distillation.
Our experiments below focus on the in-distribution data setting on CIFAR-10, and perform watermark removal attacks against three black-box watermarks, \textit{Content}, \textit{EWE} and \textit{Adi}.
More implementation details are deferred to Appendix~\ref{app:comp_recover_impt}.

% \noindent (i) \textit{Trigger Reverse-Engineering.}
% These methods are originally designed to detect and erase backdoors in DNNs.
% They assume the existence of a shortcut from other classes to the target class in a trojaned model's decision space.
% In this part, we compare \textsc{Dehydra} an adaptive version of \textit{NeuralCleanse}~\cite{wang2019NC}, which uses optimization-based method to recover a shortcut trigger pattern.

% \noindent (ii) \textit{Comparison with \dbq{General Inversion Methods}.}
% \dbq{These methods mainly aim at synthesizing realistic training samples of the target models.}
% They mostly exploit the model knowledge and various other prior information during the inversion process.
% Here we compare our \textsc{Dehydra} with DeepInversion~\cite{yin2020dreaming}, a method which
% uses BN information from all layers, and sets the target labels of a data batch randomly before inversion.
% Data-free distillation is mainly designed to synthesize realistic training samples to enable knowledge distillation even when the original training set of the target model is unavailable.  

\begin{figure}[t]
\centering
\includegraphics[width=\linewidth]{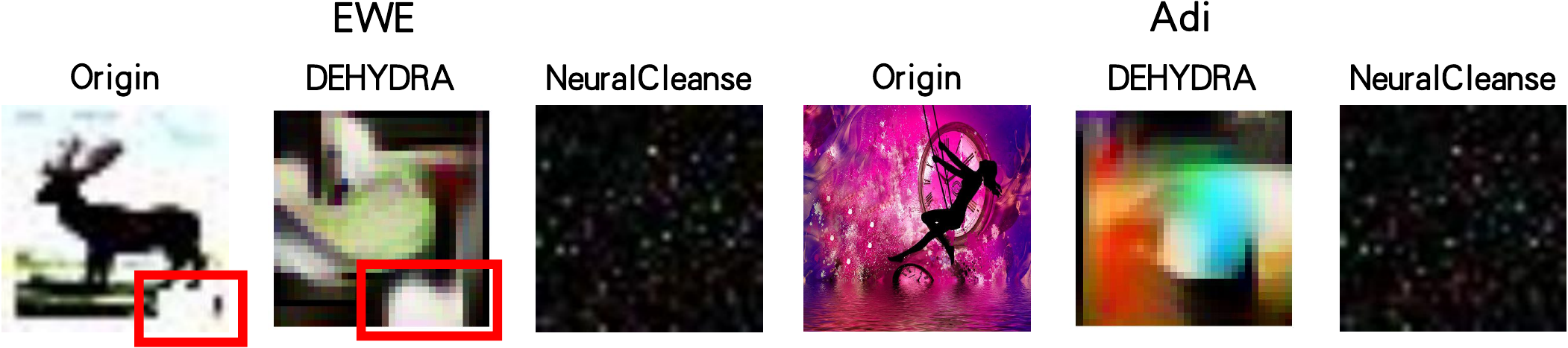}
\caption{Visualization of the recovered watermark images/trigger patterns by \textsc{Dehydra}/\textit{NeuralCleanse}.}
\label{fig:nine}
\vspace{-0.1in}
\end{figure}

\begin{figure}[t]
  \centering
    \includegraphics[width=\linewidth]{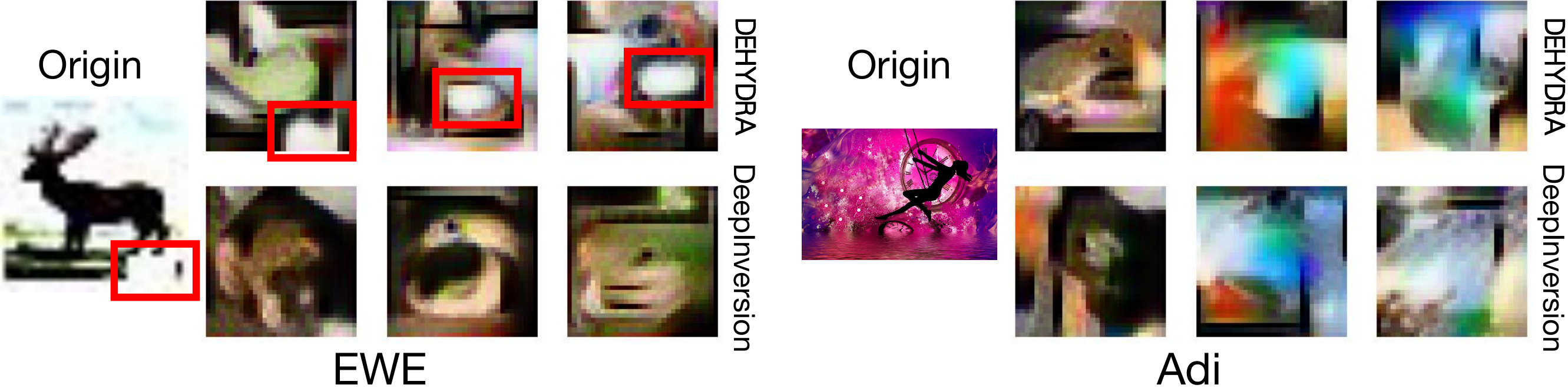} 
  \caption{Visualization of the recovered samples by \textsc{Dehydra} and \textit{DeepInversion}.}
  \label{fig:comp_general_inv}
\end{figure}

\begin{table}[t]
    \caption{Attack results using different recovering methods on CIFAR-10.}
    \vspace{-0.05in}
    \centering
    \small
    \setlength{\tabcolsep}{3pt}
    \begin{tabular}{lccc}
        \hline
        \textbf{Inversion Method} & \textbf{\textit{Content}} & \textbf{\textit{EWE}} & \textbf{\textit{Adi}} \\
        \hline
        \textsc{Dehydra} & 92.6 / \underline{\textbf{44.7}} & 92.5 / \underline{\textbf{48.2}} & 85.1 / \underline{\textbf{48.7}} \\% & 92.6 / 1 & 92.5 / 1 & 85.1 / 12 \\
        \textit{NeuralCleanse} & 92.5 / \underline{\textbf{47.5}} & 92.2 / 98.4& 84.4 / 93.0 \\ % & 92.5 / 6 & 92.2 / 97 & 84.4 / 88 \\ 
        \textit{DeepInversion} & 91.3 / 68.7& 91.2 / 99.0& 84.3 / \underline{\textbf{49.3}} \\ %& 91.3 / 44 & 91.2 / 98 & 84.3 / 13 \\ 
        \hline
    \end{tabular}
    \vspace{-0.05in}
    \label{tab:inverse_comp}
\end{table}

\noindent \textbf{Visual Effects.}
Figure~\ref{fig:nine} visualizes the recovered watermark data by \textsc{Dehydra} and the trigger patterns by \textit{NeuralCleanse}, while Figure~\ref{fig:comp_general_inv} compares \textsc{Dehydra} with \textit{DeepInversion}.
As is shown, the recovered triggers by \textit{NeuralCleanse} are visually meaningless and rather similar to each other.
This implies that the assumed shortcut in the model's decision space is non-existent in these watermarked models, and the recovered triggers are only the universal adversarial perturbations of the CIFAR-10 dataset itself.
Also, \textit{DeepInversion} tends to produce more realistic images belonging to the target classes (e.g., "frogs" for \textit{EWE}), but less informative of the watermark data.
In comparison, most of the recovered watermark data holds the essential features of the real watermark patterns, such as the white patches for \textit{EWE} and colorful fringing for \textit{Adi}.
These visual differences demonstrate the effectiveness of our aggressive inversion scheme.

% Figure~\ref{fig:comp_general_inv} in Appendix~\ref{app:comp_recover} also compares the recovered samples of \textsc{Dehydra} and DeepInversion.

% Note that the black borderlines in its recovered samples for \textit{Content} are simply evidence of the random cropping augmentations usually introduced during training, rather than the actual text shapes of watermark data.

\noindent \textbf{Attack Results.}
More quantitatively, we investigate the removal effectiveness using the three recovering methods, with attack results shown in Table~\ref{tab:inverse_comp}.
As we can see, \textit{NeuralCleanse} is ineffective against \textit{EWE} and \textit{Adi}, due to the inexistence of backdoor shortcuts.
For \textit{DeepInversion}, unlearning the inverted samples from class 6 fails to remove \textit{EWE}, because the BN statistics of the watermark data might diffuse into other classes during its random-class inversion scheme.
The recovered samples are therefore more representative of the normal data, unlearning which incurs a larger utility loss.
Only \textsc{Dehydra} is effective for removing all three watermarks from the protected models, conforming to the previous visual observations.
% which further supports the aggressive recovering of \textsc{Dehydra} does produce samples informative of the watermark information.
% For more detailed experimental settings and analysis, please refer to Appendix~\ref{app:comp_recover}.

\begin{table}[tb]
    \caption{Effectiveness of \textsc{Dehydra} on watermarked models after differential privacy training.}
    \vspace{-0.05in}
    \centering
    \label{tab:dpsgd}
    \small
    \setlength{\tabcolsep}{3pt}
    \centering
    \begin{tabular}{c cc cc}
        \hline
        \multirow{2}{*}{\textbf{Eps}} & \multicolumn{2}{c}{\textbf{\textit{Content}}} & \multicolumn{2}{c}{\textbf{\textit{Adi}}}\\ 
                \cmidrule(r){2-3} \cmidrule(r){4-5}
               & Before & After & Before & After \\ 
        \hline
        \textsc{$\epsilon = 1$}& 57.3 / 100.0	& 51.5 / 100.0	& 53.0 / 100.0	& 45.0 / 68.5	 \\
        \textsc{$\epsilon = 2$}& 62.0 / 100.0	& 55.9 / 100.0	& 60.0 / 99.4	& 53.9 / 62.1	 \\
        \textsc{$\epsilon = 6$}& 79.8 / 100.0	& 63.4 / 83.3	& 77.8 / 99.4	& 62.8 / 57.4	 \\
        \hline
    \end{tabular}
    \vspace{-0.1in}
\end{table}

\begin{table}[tb]
    \caption{The effectiveness of deceiving the detection algorithm against our improved \textsc{Dehydra}.}
    \vspace{-0.05in}
    \centering
    \label{tab:counter_det}
    \small
    \setlength{\tabcolsep}{2pt}
    \begin{tabular}{c ccc ccc}
        \hline
        \multirow{2}{*}{\textbf{$N_{trap}$}} & \multicolumn{3}{c}{\textbf{\textit{Adi}}} & \multicolumn{3}{c}{\textbf{\textit{EW}}}\\ 
                \cmidrule(r){2-4} \cmidrule(r){5-7}
               & Before & After & Gap & Before & After & Gap \\ 
        \hline
        100	& 87.6 / 100.0	& 84.4 / 60.3	& 0.1233	& 90.2 / 100.0	& 87.4 / 51.5	& 0.0766	 \\
        150	& 90.3 / 100.0	& 85.1 / 60.9	& 0.1291	& 92.7 / 100.0	& 89.9 / 92.5	& 0.7825	 \\
        200	& 89.4 / 100.0	& 89.5 / 95.9	& 0.5406	& 90.3 / 100.0	& 90.0 / 98.3	& 0.6226	 \\
        \hline
    \end{tabular}
    \vspace{-0.05in}
\end{table}

\begin{table}[tb]
    \caption{The effectiveness of bypassing the splitting algorithm against our improved \textsc{Dehydra}.}
    \vspace{-0.05in}
    \centering
    \label{tab:counter_split}
    \small
    \setlength{\tabcolsep}{2pt}
    \centering
    \begin{tabular}{c cc cc}
        \hline
        \multirow{2}{*}{\textbf{WM size}} & \multicolumn{2}{c}{\textbf{\textit{Content}}} & \multicolumn{2}{c}{\textbf{\textit{EWE}}}\\ 
                \cmidrule(r){2-3} \cmidrule(r){4-5}
               & Before & After & Before & After \\ 
        \hline
        1000	& 91.9 / 100.0	& 89.0 / 52.0	& 90.3 / 99.5	& 89.2 / 62.3	 \\
        2000	& 91.8 / 100.0	& 90.3 / 53.1	& 90.0 / 100.0	& 89.2 / 74.9	 \\
        5000	& 88.5 / 100.0	& 87.9 / 58.1	& 89.4 / 100.0	& 88.6 / 79.6	 \\
        \hline
    \end{tabular}
    \vspace{-0.05in}
\end{table}

\subsection{Potential Defenses} \label{sec:countermeasure}
The experiments above validate that \textsc{Dehydra} poses real threats to existing black-box watermarks. In this section, we further consider potential defenses against \textsc{Dehydra}. Specifically, we design defenses to resist or bypass the following three key components of \textsc{Dehydra}: the inversion-based framework, target class detection, and recovered samples splitting.

\noindent \textbf{Differential Privacy.}
Differential privacy (DP) is a general approach for privacy protection \cite{abadi2016deep}.
Here, we investigate the effectiveness of differential privacy to reduce the watermark information inverted by \textsc{Dehydra}. Roughly speaking, we use DP-SGD with the privacy budgets $(\epsilon, \delta)$ during the watermark embedding process: $\delta$ is set to be a fixed value of $1\times{10}^{-5}$ and the $\epsilon$ is varied to evaluate the attack effect under different privacy budgets, a smaller $\epsilon$ means adding larger noise to the gradients.
As we can see in Table \ref{tab:dpsgd}, with the decrease of $\epsilon$, 
\textsc{Dehydra} will achieve a worse watermark removal.
This is because the noise added to the gradients during the training phase makes the learned features of watermark information more scattered, thus compromising the effectiveness of the attack.
However, this is brought with a significant accuracy drop of 40\% on average.
% Additionally, we observed that the wm acc of the models trained under DP-SGD fluctuates significantly during the unlearn process, \lyf{indicating...}

 % A smaller $\epsilon$ means adding larger noise to the gradients.
% As we can see in Table \ref{tab:dpsgd}, with the increase of $\epsilon$, \textsc{Dehydra} has a smaller impact on clean acc when removing attacks, and the removal effect is also better. This is because the noise added to the gradients during the training phase affects the distribution of watermark information, thus compromising the effectiveness of the attack. Additionally, we observed that the wm acc of the models trained under DP-SGD fluctuates significantly during the unlearn process, indicating..."

\noindent \textbf{Deceiving the Detection Algorithm.}
Here we try to cause the severer false positive errors (making non-fixed-class watermarks detected as fixed-class ones) in Section \ref{sec:exp-target-detection}.
Concretely, model owners can deliberately change the watermark smoothness by adding $N_{trap}$ watermark samples with the specified identical labels (in addition to the original $100$ watermark samples with non-fixed target classes) during training. In Table \ref{tab:counter_det}, when $N_{trap} = 200$, the SmoothAcc Gap will be larger than the detection threshold, and the watermark will be detected as fixed-class ones.
However, forcing models to memorize more samples will cause an accuracy drop about 3\% on average.

% [def] non-fixed -> fixed: one class with much more samples
% [out] utility loss of wm. \textsc{Dehydra}conserve. still significantly lower if fooled

\noindent \textbf{Bypassing the Splitting Algorithm.}
The model owner may attempt to bypass the recovered samples splitting algorithm, by deliberately adding more watermark samples into the target classes and disrupting the dominance of normal data. This may decrease the attack effectiveness of the improved \textsc{Dehydra} in data-limited settings, as we can see in Table \ref{tab:counter_split}, the larger trigger set size, the worse the removal effect will be.
However, the attacker may skip the splitting procedure and restore the attack effectiveness with in-distribution/transfer data.
% [def] split -> wrong half: add more samples
% [out] utility loss, FP sample. \textsc{Dehydra}not easily fooled, e.g. crack Piracy500 on cifar100

%%% TABLE11 
\begin{table*}[tb]
\caption{Comparison to extraction attacks under different data ratios w.r.t. the training data on CIFAR-10.}
\label{tab:ratio_extract}
\vspace{-0.05in}
\centering
\small
\setlength{\tabcolsep}{3pt}
\begin{tabular}{lccccccccc}
\hline
\multirow{2}{*}{\textbf{Data Ratio}} & \multicolumn{3}{c}{\textbf{\textit{Noise}}} & \multicolumn{3}{c}{\textbf{\textit{EWE}}}& \multicolumn{3}{c}{\textbf{\textit{AFS}}}   \\ 
        % \hline
        \cmidrule(r){2-4} \cmidrule(r){5-7}\cmidrule(r){8-10}
                    & \textsc{Dehydra} & \textit{Extraction} & \textit{Combined} & \textsc{Dehydra} & \textit{Extraction} & \textit{Combined} & \textsc{Dehydra} & \textit{Extraction} & \textit{Combined}  \\ 
\hline
2\%	& 93.1 / \underline{\textbf{5.7}}	& 46.5 / \underline{\textbf{24.6}}	& 45.5 / \underline{\textbf{14.2}}	& 93.3 / \underline{\textbf{47.6}}	& 48.8 / 56.5	& 46.6 / 50.3	& 88.1 / \underline{\textbf{44.3}}	& 48.0 / \underline{\textbf{44.9}}	& 45.8 / \underline{\textbf{45.5}}	 \\
5\%	& 93.3 / \underline{\textbf{10.4}}	& 62.1 / 88.7	& 56.6 / \underline{\textbf{5.7}}	& 93.1 / \underline{\textbf{48.7}}	& 58.7 / 55.0	& 56.0 / \underline{\textbf{49.2}}	& 90.4 / \underline{\textbf{46.2}}	& 63.5 / \underline{\textbf{46.2}}	& 53.5 / \underline{\textbf{48.1}}	 \\
10\%	& 93.4 / \underline{\textbf{8.5}}	& 72.4 / \underline{\textbf{18.0}}	& 65.6 / \underline{\textbf{5.7}}	& 93.3 / \underline{\textbf{49.7}}	& 70.1 / \underline{\textbf{48.2}}	& 66.3 / \underline{\textbf{49.2}}	& 91.2 / \underline{\textbf{46.8}}	& 71.8 / \underline{\textbf{46.8}}	& 67.4 / \underline{\textbf{48.7}}	 \\
33.33\%	& 94.0 / \underline{\textbf{8.5}}	& 86.3 / 81.1	& 78.5 / \underline{\textbf{33.1}}	& 93.2 / \underline{\textbf{49.2}}	& 81.8 / 60.2	& 78.9 / \underline{\textbf{47.6}}	& 92.8 / 50.6	& 86.7 / \underline{\textbf{45.5}}	& 78.5 / \underline{\textbf{50.0}}	 \\
\hline
% \vspace{-0.1in}
\end{tabular}
\vspace{-0.1in}
\end{table*}

\section{Discussion}\label{sec:discussion}
% First, fixed-class watermarks usually collect samples from a certain source class, add watermark patterns and label them to a fixed target class~\cite{zhang2018protecting, jia2021entangled}. In comparison, backdoor attacks apply the trigger to samples from all classes and expect any samples with the trigger to be classified as the target class.
% Also, these watermark samples from the same source class might be added with different watermark patterns~\cite{zhang2018protecting}, compared to the fixed trigger pattern in traditional backdoors.
% Non-fixed-class watermarks have even more differences, such that each watermark sample is paired with its own target class, which makes the watermarked model memorize them without learning other backdoor behaviors.

\noindent \textbf{Black-box DNN Watermarks and Backdoors.}
% The discussion of differences between them are deferred to Appendix~\ref{app:bb_wm}.
In previous literature, black-box watermarks are sometimes also called backdoor-based~\cite{Shafieinejad2019regularization, liu2021wild, Aiken2020Laundering} or trigger set-based watermarks~\cite{lee2022evaluating}.
However, black-box watermarks have substantial differences from DNN backdoors in terms of the watermark data patterns and target class settings (\S\ref{sec:back-wm}), especially for the non-fixed-class watermarks.

Therefore, attacking black-box watermarks has its unique challenges, especially for an attacker who usually has no knowledge of the underlying watermarking schemes.
This also explains why existing removal attacks extended from backdoor defenses, such as \textit{Laundering} (\S\ref{sec:attack-performance}) and \textit{NeuralCleanse} (\S\ref{sec:inv_comp}), could not effectively remove most of the watermarks.

\noindent \textbf{Attack Performance when More Data is Available.}
Our work primarily focuses on the data-limited settings, where only 2\% of the training data is used in most of the experiments (\S\ref{sec:attack-performance}). This limited data availability may have a slight impact on the utility of the surrogate model for several non-fixed-class watermarks.
In this part, we evaluate the performance of the improved \textsc{Dehydra} when more data is available. As shown in Table~\ref{tab:ratio_extract}, for the non-fixed-class watermark \textit{AFS}, the utility loss is mitigated when more data is available. The accuracy drop is controlled within 1\% for the three investigated watermarks, when 33\% of the training data is available.

% Besides, considering the trend of open AI data hubs (e.g., Huggingface Datasets~\cite{hugg}), the attacker is able to find public data from similar domains when performing \textsc{Dehydra} in practice, which can further reduce the utility loss.

\noindent \textbf{Comparison to Model Extraction Attacks.}
Model extraction attacks are effective for watermark removal but typically require a large query dataset. Here we compare the improved \textsc{Dehydra} with the classical \textit{Extraction} attack~\cite{jia2021entangled} and the recently proposed \textit{Combined} attack (which combines transfer learning with label smoothing) that has been reported as effective against all watermarks using one-third of the training data~\cite{lukas2022sok}.
As shown in Table~\ref{tab:ratio_extract}, the surrogate models obtained through extraction attacks demonstrate significantly dependent performance on the available dataset size, resulting in relatively low accuracy when the data ratio is 2\%, 5\%, and 10\%. In contrast, our \textsc{Dehydra} maintains high accuracy while effectively removing the watermark across all data ratios.

\noindent \textbf{Comparison to Data Extraction Attacks.}
Data extraction attacks aim at reconstructing training data from a well-trained DNN model.
% 因为模型水印通常用来保护产权、抵御白盒攻击者，所以在模型参数完全暴露给攻击者的场景下水印数据泄露的风险需要充分评估。
Since DNN watermarks are typically designed to trace ownership after model leakage (\S\ref{sec:back-wm}), it is essential to thoroughly assess the privacy risk of the watermark data when the model parameters are completely exposed.
% Existing data extraction attacks aim at reconstructing training data from a well-trained model.
In fact, previous \textit{Content}~\cite{zhang2018protecting} and \textit{EWE}~\cite{jia2021entangled} watermarks did investigate the watermark privacy risk, while using naive extraction methods, leading to the conclusion that their watermarks are "secure".
On the contrary, \textsc{Dehydra} specially targets the overfitted watermark data, and hence adopts an aggressive class-wise inversion method to better exploit model internals.
Ablation studies in Section~\ref{sec:inv_comp} also demonstrate that our \textsc{Dehydra} captures more meaningful visual features of the watermark samples (Figure~\ref{fig:nine},~\ref{fig:comp_general_inv}) and achieves significantly better removal results (Table~\ref{tab:inverse_comp}) compared to traditional data extraction methods.

% Our work focuses on the vulnerabilities revealed by model inversion.
% Note that this has also been studied by Zhang et al.~\cite{zhang2018protecting}.
% However, they used a rather naive inversion method from~\cite{fredrikson2015inversion_conf} and only presented the visually-meaningless inverted samples.
% They neither delved into the neuron coverage of the inverted samples, nor attempted further watermark removal via unlearning the inverted ones.

\noindent \textbf{Attack Efficiency.}
For a target model with a large number of class labels, the main computation cost of \textsc{Dehydra} comes from the class-wise watermark recovering process, which is proportional to the number of classes.
Similar to the solutions in NeuralCleanse~\cite{wang2019NC}, we find parallel computation on multiple GPUs and early-stopping mechanisms also useful here.
Moreover, we propose a multi-class inversion mechanism while optimizing a mixed data batch, where the BN term is calculated for samples of each class separately.

\noindent \textbf{Limitations and Future Works.}
\louis{Our methodological design and evaluations primarily focus on the image classification task, which is consistent with the mainstream DNN watermarking research \cite{lederer2023sba, lukas2022sok, lee2022evaluating}.
% This is because previous DNN watermarking researches were predominantly based on this classical task, with core concepts transferable to other tasks \cite{lederer2023sba, lukas2022sok, lee2022evaluating}.
We note that there do exist some recent works extending black-box DNN watermarks to other domains, including self-supervised learning for encoders, graph-based tasks, and text-based tasks.
For these new applications, we posit that the watermark information is still deeply embedded within the internals of the watermarked models, suggesting that \textsc{Dehydra} continues to be a promising and effective attack.
The primary challenge to adapt the \textsc{Dehydra} attack lies in developing a suitable recovering method (for instance, prompting-based methods \cite{carlini2021extracting} could be utilized to recover the watermark information stored in a large language model).
Nonetheless, we leave it as future work to explore the further effectiveness of \textsc{Dehydra} on other tasks.
}

\section{Conclusion}

In this paper,
we identify a shared trait of existing black-box DNN watermarks: they exploit the over-parameterization of DNNs to especially memorize the watermark data correlated to the target classes.
Based on this,
we propose a novel watermark-agnostic removal framework, namely \textsc{Dehydra}, which exploits the model internals to recover and unlearn the underlying watermarks.
With in-depth analysis, theoretical justifications and pilot studies, we further improve \textsc{Dehydra} with target class detection and recovered sample splitting algorithms, which help preserve the model utility and significantly reduce the dataset dependence.
Extensive evaluations on various settings against ten mainstream black-box watermarks demonstrate that \textsc{Dehydra} is generally effective and has minimal utility impact, with much more relaxed or even no assumptions on the dataset availability.

% In this paper, we summarize the commonality of existing black-box DNN watermarks in especially memorizing watermark data correlated to the target labels.
% Based on this, we uncover a new removal attack surface via recovering and unlearning watermark samples from a target watermarked model.
% Our novel attack \textsc{Dehydra} is watermark-agnostic and effective against most of mainstream black-box DNN watermarking schemes under the data-limited scenarios.
% In addition, we further improve the basic \textsc{Dehydra}by target class detection and recovered samples splitting, based on in-depth analysis and observations on the existing black-box watermark schemes.
% The improved \textsc{Dehydra} reduces the utility loss and achieves data-free watermark removal on five of the watermark schemes. Compared with existing removal attacks, \textsc{Dehydra} achieves strong removal effectiveness on the covered watermarks, preserving  at least $90\%$ of the stolen model utility,
% under the data-limited settings, i.e., less than $2\%$ of the training data or even data-free.

% \begin{acks}
% hello.
% \end{acks}

\begin{acks}
We would like to thank the anonymous reviewers for their insightful comments that helped improve the quality of the paper.
This work was supported in part by the National Key Research and Development Program (2021YFB3101200), National Natural Science Foundation of China (62472096, 62172104, 62172105, 62102093, 62102091, 62302101, 62202106).
Min Yang is a faculty of Shanghai Institute of Intelligent Electronics \& Systems and Engineering Research Center of Cyber Security Auditing and Monitoring, Ministry of Education, China.
Mi Zhang and Min Yang are the corresponding authors.
\end{acks}

% \bibliographystyle{IEEEtranS}
\bibliographystyle{ACM-Reference-Format}
\balance
\bibliography{lyf_ref_revised}

\newpage

\appendix

\section{Omitted Details of \textsc{Dehydra}}
\subsection{Classical Natural Image Priors}
\label{sec:app:nip}
Specifically, we incorporate two regularization terms $\mathcal{R}_{\ell_2}(\cdot)$ and $\mathcal{R}_{tv}(\cdot)$, as demonstrated in~\cite{mordvintsev2015inceptionism}, which correspond to $\ell_2$ distance and total variation, to steer the optimized samples away from unrealistic ones:
\begin{equation}
    \mathcal{R}_{\ell_2}(B) = \sum_{\hat{x}_i \in B} \lVert \hat{x}_i \rVert_2^2,
\end{equation}
\begin{equation}
    \mathcal{R}_{tv}(B) = \sum_{\hat{x}_i \in B} \sum_{(j,k)} \sum_{(j^\prime, k^\prime) \in \delta(j,k)} \lVert \hat{x}_{i_{(j,k)}} - \hat{x}_{i_{(j^\prime,k^\prime)}} \rVert_2^2,
\end{equation}
where $\delta(j,k)$ indicates a set of pixels adjacent to the pixel at $(j,k)$ in the image $\hat{x}_i$.

\subsection{Splitting Recovered Samples} \label{app:algo}
The detailed pipeline to class-wise split the recovered samples is shown in Alg.\ref{alg:split}.

\begin{algorithm}[!htb]
  \caption{Split recovered samples.}
  \label{alg:split}
  \textbf{Input:} target model $f_w$ and its $l$-th layer to analyse, recovered samples $\{B_c\}_{\ast}$, saliency ratio $\beta$, split ratio $\gamma$.\\
  \textbf{Output:} proxy watermark samples $\{B_{c,wmk}\}_\ast$, proxy normal samples $\{B_{c,nor}\}_\ast$.
  
\begin{algorithmic}[1]
  \STATE $M \gets$ samples' size of each $B_c$ in $\{B_c\}_{\ast}$, \\
  $D \gets$ neurons' size of the $l$-th layer
  \STATE $\{B_{c,wmk}\}_\ast \gets \{\}$, $\{B_{c,nor}\}_\ast \gets \{\}$
  
  \FOR{batch $B_c$ in $\{B_c\}_{\ast}$}
  \STATE Forward propagate and extract activation of the $l$-th layer $act={f_w}_{(l)}(B_c)$ \\%\hfill // $act$ is a matrix of shape $(M, D)$
  /* Find salient neurons */ 
  \STATE $Impt \gets \{\}$
  \FOR{neuron $j=0,1,\dots, D-1$}
  \STATE $Impt[j] \gets \frac{\mu_j}{\sigma_j}$  %\hfill // estimated on $act$
  \ENDFOR
  \STATE $sorted\_neurons \gets \text{argsort}(Impt)$
  \STATE $salient\_neurons \gets sorted\_neurons[: \beta \times D]$ \\
  /* Split recovered batch $B_c$ */ 
  \STATE $Contribution \gets \{\}$
  \FOR{sample $\hat{x}_i$ in $B_c$}
  \STATE $Contribution[i] \gets \text{sum}(act[i, salient\_neurons])$
  \ENDFOR
  \STATE $sorted\_samples \gets \text{argsort}(Contribution)$
  \STATE $proxy\_wmk \gets sorted\_samples[\gamma \times M :]$, \\
  $proxy\_nor \gets sorted\_samples[: \gamma \times M]$
  \STATE $B_{c, wmk} \gets B_c[proxy\_wmk]$, \\
  $B_{c, nor} \gets B_c[proxy\_nor]$
  \STATE Append split data $B_{c,wmk}$ and $B_{c,nor}$ to $\{B_{c,wmk}\}_\ast$ and $\{B_{c,nor}\}_\ast$ respectively
  % \STATE $\{B_{c,wmk}\}_\ast$.append($B_{c,wmk}$), \\
  % $\{B_{c,nor}\}_\ast$.append($B_{c,nor}$)
  \ENDFOR
  \RETURN $\{B_{c,wmk}\}_\ast$, $\{B_{c,nor}\}_\ast$
\end{algorithmic}
\end{algorithm}

\section{Omitted Details in Experiments}
\subsection{More Backgrounds on Existing Black-box Model Watermarking Schemes} \label{sec:back_impt_wms}
%%%%%%%%%%%%%% BEGIN APPENDIX 
Our evaluation covers ten existing black-box DNN watermarks, which are categorized as \textit{fixed-class watermarks} (i.e., all the watermark data are paired with the identical target class), and \textit{non-fixed-class watermarks} (i.e., each watermark data is paired with its own target class).
\begin{itemize}
    \item \textit{Fixed-class watermarks:}
    \textbf{(1) \textit{Content}, (2) \textit{Noise}, (3) \textit{Unrelated}~\cite{zhang2018protecting}.}
    \textit{Content} and \textit{Noise} watermarks draw normal samples from a source class and apply watermark patterns, which are text shapes and random noises respectively, as watermark data.
    \textit{Unrelated} watermark draws samples from another dataset.
    All these selected watermark data are annotated with one fixed target class.
    \textbf{(4) \textit{Piracy}~\cite{li2019piracy}} \lwx{ modifies the values of certain pixels in an image to extreme maximum or minimum values as the trigger, which is then mapped to a fixed class. The position of these pixels and the target class are determined by a signature provided by the user}.
    \textbf{(5) \textit{EWE}~\cite{jia2021entangled}} selects samples from a source class, adds square patterns and annotates with a fixed target class.
    During training, EWE enforces the model to learn entangled features between watermark samples and normal samples of the target class.
    \item \textit{Non-fixed-class watermarks:}
    \textbf{(6) \textit{Adi}~\cite{adi2018turning}} leverages a set of abstract out-of-distribution images as watermark data and annotates them randomly.
    \textbf{(7) \textit{AFS}~\cite{le2020afs}} generates adversarial examples near the classification boundary of a normally trained model and annotates them with the ground-truth label. The model is then finetuned with these samples added to the training set.
    \textbf{(8) \textit{EW}~\cite{namba2019ew}} selects normal samples and annotates them with random labels except the ground-truth.
    During embedding, a normally trained model is applied with an exponential weight operator layer-wise, and then finetuned with watermark samples added to the training set.
    \textbf{(9) \textit{Blind}~\cite{li2019blind}} \lwx{takes ordinary samples and exclusive logos as inputs, uses an encoder to generate watermarks that are almost indistinguishable from the original samples, and uses a discriminator to assist in training. All watermarks are labeled as a fixed class.}.
    \textbf{(10) \textit{Mark}~\cite{guo2018mark}} \lwx{modifies a portion of the dataset based on the user's signature. The pixel values at specific positions of each image are modified to extreme values according to the signature, and the label of each image is mapped to one of the remaining labels in a specific way based on the signature.}.
\end{itemize}
%%%%%%%%%%%%%% END APPENDIX 

\subsection{Implementations of the Target Watermarks}\label{app:impl-wm}
Our implementations mainly referred to some recent open-source watermark framework \cite{lederer2023sba, lukas2022sok}, and strictly followed the specifications and hyperparameters in the original papers.
For \textit{Piracy}~\cite{li2019piracy} and \textit{Blind}~\cite{li2019blind}, we set the size for watermark data as 1\% of the clean training set, following their papers.
For other black-box watermarks, the size of the watermark data is set to 100.

\begin{table}[t]
\caption{The performance of the target watermarked models on three datasets. The threshold is estimated on 20 null models as suggested by Lukas et al. \cite{lukas2022sok}}
\label{tab:target_models}
\vspace{-0.05in}
\centering
\setlength{\tabcolsep}{2pt}
% \res
\resizebox{\linewidth}{!}{
    \begin{tabular}{l cc cc cc}
        \hline
            \multirow{2}{*}{Watermarks} & \multicolumn{2}{c}{\textbf{MNIST}} & \multicolumn{2}{c}{\textbf{CIFAR-10}}& \multicolumn{2}{c}{\textbf{CIFAR-100}}   \\ 
            % \hline
            \cmidrule(r){2-3} \cmidrule(r){4-5}\cmidrule(r){6-7}
                        & Clean / WM & Thre. & Clean / WM & Thre. & Clean / WM & Thre.  \\ 
            \hline
            \textbf{\textit{Content}}     & 99.08 / 100 & 0.0083  & 93.77 / 100 & 0.1044  & 76.27 / 100   & 0.0244   \\ 
            \textbf{\textit{Noise}}       & 99    / 100 & 0.1057  & 94.13 / 100 & 0.4697  & 76.67 / 100   & 0.0151   \\ 
            \textbf{\textit{Unrelated}}   & 98.94 / 100 & 0.088   & 93.65 / 100 & 0.0931  & 75.93 / 100   & 0.0305   \\ 
            \textbf{\textit{Piracy}}      & 99.22 / 100 & 0.2807  & 93.31 / 100 & 0.1966  & 73.69 / 100   & 0.1672   \\ 
            \textbf{\textit{EWE}}         & 98.86 / 100 & 0.0025  & 93.97 / 100 & 0.045   & 67.73 / 100   & 0.0059   \\ 
            
            \textbf{\textit{Adi}}         & 98.82 / 100 & 0.149   & 93.91 / 100 & 0.1425  & 75.56 / 100   & 0.0176   \\ 
            \textbf{\textit{AFS}}         & 98.73 / 100 & 0.1506  & 92.96 / 100 & 0.2103  & 65.76 / 100   & 0.0813   \\ 
            \textbf{\textit{EW}}          & 98.55 / 100 & 0.0025  & 92.89 / 100 & 0.1335  & 75.98 / 100   & 0.0088   \\ 
            \textbf{\textit{Blind}}       & 98.64 / 100 & 0.1432  & 91.18 / 100 & 0.1452  & 67.22 / 100   & 0.0266    \\
            \textbf{\textit{Mark}}        & 98.57 / 100 & 0.0025  & 93.97 / 100 & 0.1381  & 65.25 / 100    & 0.0125    \\
            \hline
    \end{tabular}
}
\end{table}

\subsection{Performance of Watermarked Models}\label{app:target-models}
The performance of the well-trained DNN models protected by ten different black-box watermarks on three benchmark datasets is shown in Table~\ref{tab:target_models}.

\subsection{Details of Baseline Removal Attacks}\label{app:removal}
Here we provide some implementation details of the baseline removal attacks.
For the \textit{Pruning} attack, we set the weight prune ratio to 80\%.
For the \textit{Fine-pruning} attack, we set the neuron prune ratio to 20\%.
For the \textit{Finetuning} attack, we start with a learning rate of 0.05 as in \cite{libo2019leveraging}, followed by a gradually decreasing learning rate schedule.
% For the \textit{Regularization} attack, we reproduce it based on the open-source repository\footnote{\url{https://github.com/CodeSubmission642/WatermarkRobustness}}.
For \textit{Regularization} and \textit{Distraction}, we strictly follow the original papers.
For \textit{Laundering}, we follow the main method in \cite{Aiken2020Laundering} and implement it with a certain watermark knowledge. We use the original labels for watermarks using non-OOD data, and least-likely labels for watermarks using OOD data.

% \subsection{Black-box DNN Watermarks and Backdoors}\label{app:bb_wm}

\subsection{Details of the Compared Recovering Methods}\label{app:comp_recover_impt}

During ablation studies in Section~\ref{sec:inv_comp}, to better control other variables and perform a fair comparison on the recovering methods, we assume the attacker has got the ground-truth knowledge of the target label settings (i.e., fixed-class with target class 6 for \textit{Content} and \textit{EWE}, and non-fixed for \textit{Adi}) and attempts to perform an adaptive removal attack.

To further ensure the comparison fairness, we implement the improved \textsc{Dehydra} with the splitting algorithm omitted.
For a fixed-class watermark, we specifically unlearn the recovered samples corresponding to its target class (or the poisoned samples applied with a trigger towards the target class).
For a non-fixed-class watermark, we unlearn the recovered samples from all classes (or the poisoned samples applied with different triggers towards all classes).

\section{More Experiment Results}

% \subsection{Results of Exclusive Salient Neurons Analysis on Non-fixed-class Watermarks}
\subsection{Salient Neurons Analysis of Recovered Samples}\label{app:salient_neuron}
In this section, we follow the settings in Section~\ref{sec:unlearn2} and validate the hypothesis of normal data dominance on all the ten black-box watermarks.
Concretely, for the ten batches of normal samples $X_0, X_1, \dots, X_9$ and one batch of watermark samples $X_{wm}$, we capture the salient neurons of each batch, e.g. $SN(X_i)$, and calculate its intersection ratio w.r.t. those of the recovered samples $SN(B_6)$.
Note that for each watermarked model, the sum of intersection ratios of all the classes could be larger than 1, since different classes may share a few common salient neurons.

As is shown in Figure~\ref{fig:neuron_cross}, both the fixed-class and non-fixed-class watermarks have witnessed the noisy activations, as the recovered samples would activate salient neurons of samples from classes other than 6.
Moreover, $\lvert SN(B_6)\rvert \cap \lvert SN(X_6)\rvert$ is always larger than$\lvert SN(B_6)\rvert \cap \lvert SN(X_{wm})\rvert$, implying that the most consistent activations are correlated with the normal samples of class 6.
Note that for non-fixed-class watermarks, the gap still holds but is less evident, possibly due to their more complex input-output pairing relation, which leaves the watermark data more entangled with the normal data.
% this does not conflict with the goal of the watermark recovering stage. In fact,
% samples close to the real watermark data still exists in the inverted batch (even for the \textit{Piracy} watermark in Table~\ref{fig:neuron_cross}), but less evident in a batch-wise statistical view, compared to those close to the normal data.

\begin{figure}[tb]
\centering
\vspace{-0.1in}
\includegraphics[width=0.7\linewidth]{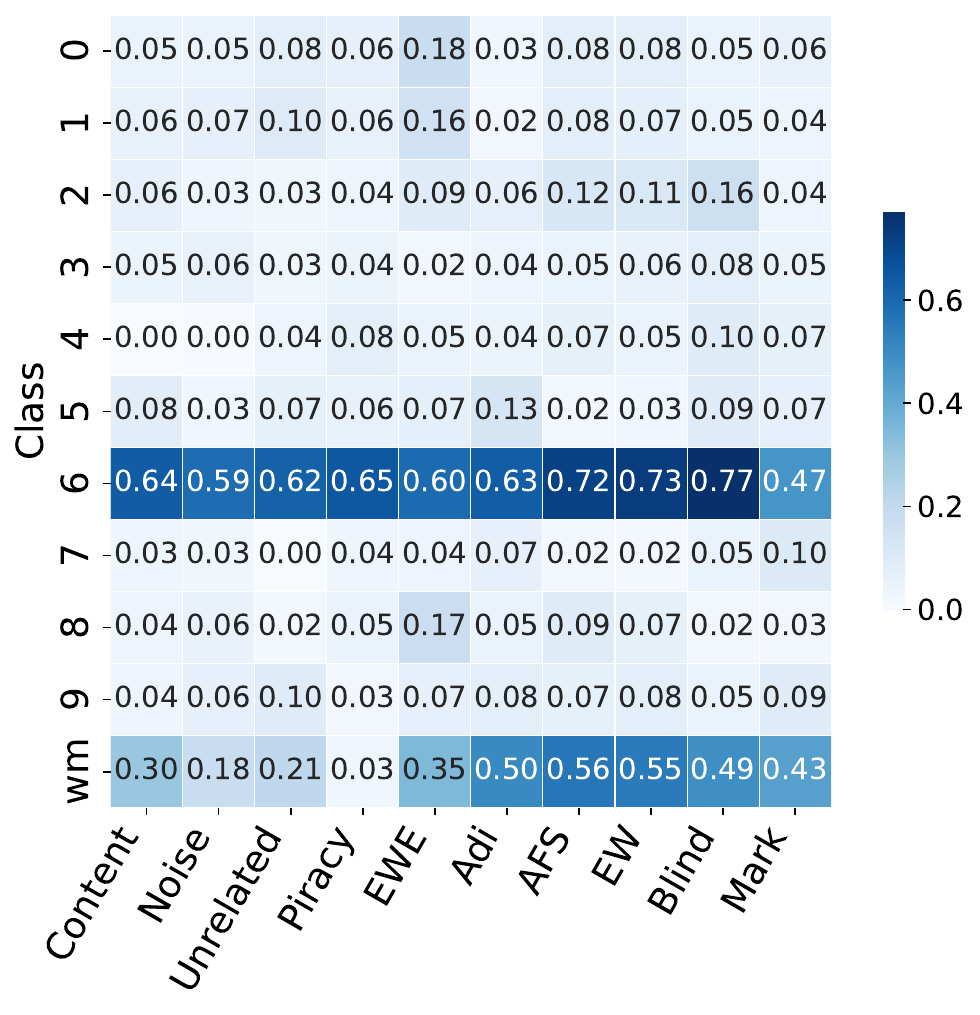}
\vspace{-0.1in}
\caption{Class-wise intersection ratio between the salient neurons of the real samples and the recovered samples.}
\label{fig:neuron_cross}
\vspace{-0.15in}
\end{figure}

\subsection{Attack Results on MNIST}\label{app:atk-mnist}

Note that MNIST is a relatively simple task, and the LeNet-5~\cite{lecun1998lenet} structure we select has no batch normalization layers.
In this case, we simply set the term $\mathcal{R}_{bn}(B)=0$ during watermark recovering.
Both the basic and improved \textsc{Dehydra} still perform well, illustrating the generality of our methods.

We select 1000 samples for the in-distribution setting, and 2000 samples from EMNIST~\cite{cohen2017emnist} for the transfer setting.
The attack results on different data settings are shown in Table~\ref{tab:in_distribution_mnist}, Table~\ref{tab:transfer_mnist}, and  Table~\ref{tab:data_free_mnist}.

\begin{table*}[htbp]
\caption{Comparison of removal attacks against black-box watermarks under the in-distribution setting on MNIST.
x / y denotes the clean accuracy / rescaled watermark accuracy.
% x/y denotes clean accuracy/rescaled watermark accuracy. The left five columns are attack results of fixed-label watermarks, while the right five are of non-fixed-label watermarks.
}
\label{tab:in_distribution_mnist}
\vspace{-0.1in}
\centering
\small
\setlength{\tabcolsep}{2.5pt}
\begin{tabular}{lcccccccccc}
\hline
\textbf{Attacks} & \textbf{\textit{Content}} & \textbf{\textit{Noise}} & \textbf{\textit{Unrelated}} & \textbf{\textit{Piracy}} & \textbf{\textit{EWE}} & \textbf{\textit{Adi}} & \textbf{\textit{AFS}} & \textbf{\textit{EW}} & \textbf{\textit{Blind}} & \textbf{\textit{Mark}} \\
\hline
\textit{None}                & 99.1 / 100.0& 99.0 / 100.0& 98.9 / 100.0& 99.2 / 100.0& 98.9 / 100.0& 98.8 / 100.0& 98.7 / 100.0& 98.5 / 100.0& 98.6 / 100.0& 98.6 / 100.0 \\
\hline
\textit{Fine-pruning}        & 96.6 / 53.6& 97.0 / 47.4& 97.7 / 100.0& 98.4 / 75.0& 97.8 / 100.0& 96.0 / 52.4& 95.6 / 50.0& 97.4 / 54.4& 97.3 / 49.8& 96.0 / 53.4 \\
% \specialrule{0em}{2pt}{1pt}
\textit{Finetuning}          & 97.7 / 60.7& 97.2 / 45.2& 97.4 / 99.5& 99.1 / 100.0& 97.4 / 87.0& 95.3 / 48.3& 97.3 / 58.8& 95.3 / 54.4& 95.9 / 49.2& 95.8 / 52.9 \\
% \specialrule{0em}{2pt}{1pt}
\textit{Regularization}      & 97.4 / 50.6& 97.6 / 45.8& 97.2 / 53.9& 97.3 / 41.6& 97.7 / 51.4& 97.4 / 55.3& 97.0 / 50.6& 95.7 / 54.4& 95.9 / 51.0& 96.0 / 53.9 \\
% \specialrule{0em}{2pt}{1pt}
\textit{Distraction}         & 90.4 / 52.6& 87.3 / 44.6& 90.8 / 45.7& 99.2 / 100.0& 95.0 / 49.9& 97.9 / 55.9& 85.8 / 50.6& 85.5 / 55.9& 94.4 / 51.6& 93.9 / 54.4 \\
% \specialrule{0em}{2pt}{1pt}
\textit{Laundering}          & 97.0 / 51.6& 97.7 / 46.9& 96.0 / 45.2& 98.6 / 36.7& 96.0 / 53.4& 95.4 / 48.3& 96.6 / 50.0& 95.8 / 55.4& 96.9 / 49.8& 96.2 / 53.4 \\
\hline

\textsc{Dehydra}\textsubscript{Basic}      & 96.7 / 50.6& 96.7 / 44.6& 93.8 / 45.2& 96.3 / 37.4& 97.1 / 49.9& 96.7 / 49.5& 96.9 / 46.4& 95.9 / 53.9& 96.1 / 49.2& 96.9 / 52.9 \\
% \specialrule{0em}{2pt}{1pt}
\textsc{Dehydra}\textsubscript{Improved}   % & 96.7 / 50.6& 96.7 / 44.6& 93.8 / 45.2& 96.3 / 37.4& 97.1 / 49.9& 96.7 / 49.5& 96.9 / 46.4& 95.9 / 53.9& 96.1 / 49.2& 96.9 / 52.9 \\ 
& 98.7 / 52.1& 99.0 / 45.8& 98.7 / 45.2& 98.8 / 36.0& 98.3 / 51.4& 97.2 / 46.5& 97.0 / 47.0& 96.4 / 53.9& 96.5 / 49.8 & 97.1 / 52.9 \\
\hline
\end{tabular}
% \footnotesize\raggedright{*The baseline \pxd{in this scenario has slightly different settings}.}
\end{table*}

\begin{table*}[htbp]
\caption{Comparison of removal attacks against black-box watermarks under the transfer setting on MNIST.
% x/y denotes clean accuracy/rescaled watermark accuracy, and values behind ± report the standard deviation in 5 repetitive tests.
}
\label{tab:transfer_mnist}
\vspace{-0.1in}
\centering
\small
\setlength{\tabcolsep}{2.5pt}
\centering

\begin{tabular}{c cccccccccc}
\hline
\textbf{Attacks} & \textbf{\textit{Content}} & \textbf{\textit{Noise}} & \textbf{\textit{Unrelated}} & \textbf{\textit{Piracy}} & \textbf{\textit{EWE}} & \textbf{\textit{Adi}} & \textbf{\textit{AFS}} & \textbf{\textit{EW}} & \textbf{\textit{Blind}} & \textbf{\textit{Mark}} \\
\hline
    \textit{None}    & 99.1 / 100.0& 99.0 / 100.0& 98.9 / 100.0& 99.2 / 100.0& 98.9 / 100.0& 98.8 / 100.0& 98.7 / 100.0& 98.5 / 100.0& 98.6 / 100.0& 98.6 / 100.0 \\
\hline
    \textit{Regularization}
            & 73.1 / 75.8& 61.0 / 60.9& 65.5 / 99.5& 67.8 / 91.0& 77.8 / 49.9& 72.0 / 58.9& 62.6 / 58.2& 72.2 / 59.9& 80.3 / 53.9& 74.7 / 57.4 \\
% \specialrule{0em}{2pt}{1pt}
\textit{Distraction}
            & 85.9 / 52.1& 91.8 / 44.1& 96.2 / 81.4& 99.2 / 100.0& 98.9 / 100.0& 95.9 / 61.8& 87.5 / 60.0& 88.2 / 53.9& 98.5 / 90.1& 94.8 / 62.4 \\
\hline
\textsc{Dehydra}\textsubscript{Basic} 
            & 83.2 / 54.1& 96.3 / 46.3& 95.4 / 51.2& 96.9 / 42.3& 97.0 / 53.4& 94.8 / 61.2& 90.9 / 55.3& 90.3 / 75.4& 91.8 / 64.4& 90.8 / 84.0 \\
% \specialrule{0em}{2pt}{1pt}
\textsc{Dehydra}\textsubscript{Improved} 
            & 96.8 / 53.1& 98.6 / 44.6& 97.3 / 45.2& 98.1 / 36.0& 97.9 / 49.9& 95.8 / 60.6& 90.6 / 51.1& 92.3 / 62.4& 91.9 / 62.7& 91.1 / 61.9 \\
\hline
\end{tabular}
% \footnotesize\raggedright{*The PST-Out attack has a slightly different scenario, as described in Table \ref{tab:in_distribution}.}
\end{table*}
\begin{table}[tb]
\caption{Comparison of removal attacks against black-box watermarks under the data-free setting on MNIST.
% x/y denotes clean accuracy/rescaled watermark accuracy, and values behind ± report the standard deviation in 5 repetitive tests.
}
\label{tab:data_free_mnist}
\vspace{-0.1in}
\centering
\small
\setlength{\tabcolsep}{1.5pt}
\begin{tabular}{c cccccc}
\hline
\textbf{Attacks} & \textbf{\textit{Content}} & \textbf{\textit{Noise}} & \textbf{\textit{Unrelated}} & \textbf{\textit{Piracy}} & \textbf{\textit{EWE}} \\
\hline
    \textit{None}        & 99.1 / 100.0& 99.0 / 100.0& 98.9 / 100.0& 99.2 / 100.0& 98.9 / 100.0 \\
\hline
    \textit{Pruning}     & 98.9 / 100.0& 98.3 / 100.0& 98.8 / 100.0& 99.1 / 100.0& 98.5 / 99.5 \\
\hline
\makecell{\textsc{Dehydra}\\(Improved)}
                & \makecell{93.2 / 49.6} & \makecell{97.8 / 44.1} & \makecell{89.3 / 47.9} & \makecell{97.3 / 36.0} & \makecell{94.5 / 49.9} \\

% \shortstack[c]{MIRA (Improved)}
%                 & 93.2 / 49.6 & 97.8 / 44.1 & 89.3 / 47.9 & 97.3 / 36.0 & 94.5 / 49.9 \\

\hline
\end{tabular}
\end{table}

\subsection{Attack Results on CIFAR-100}\label{app:atk-cifar100}

For the CIFAR-100 dataset, we select 1000 samples for the in-distribution setting, and 2000 samples from CIFAR-10 for the transfer setting.
The attack results on different data settings are shown in Table~\ref{tab:in_distribution_100}, Table~\ref{tab:transfer_100}, and  Table~\ref{tab:data_free_100}.

\begin{table*}[!tb]
\caption{Comparison of removal attacks against black-box watermarks under the in-distribution setting on CIFAR-100.
% x/y denotes clean accuracy/rescaled watermark accuracy. The left five columns are attack results of fixed-label watermarks, while the right five are of non-fixed-label watermarks.
}
\label{tab:in_distribution_100}
\vspace{-0.1in}
\centering
\small
\setlength{\tabcolsep}{2.5pt}
\begin{tabular}{c cccccccccc}
\hline
\textbf{Attacks} & \textbf{\textit{Content}} & \textbf{\textit{Noise}} & \textbf{\textit{Unrelated}} & \textbf{\textit{Piracy}} & \textbf{\textit{EWE}} & \textbf{\textit{Adi}} & \textbf{\textit{AFS}} & \textbf{\textit{EW}} & \textbf{\textit{Blind}} & \textbf{\textit{Mark}} \\
\hline
\textit{None}                & 76.3 / 100.0& 76.7 / 100.0& 75.9 / 100.0& 73.7 / 100.0& 67.7 / 100.0& 75.6 / 100.0& 65.8 / 100.0& 76.0 / 100.0& 67.2 / 100.0& 65.2 / 100.0 \\
\hline
\textit{Fine-pruning}        & 56.3 / 53.9& 58.6 / 54.3& 58.4 / 48.4& 7.6 / 40.0& 63.3 / 95.5& 57.5 / 56.2& 51.2 / 70.6& 52.6 / 61.7& 10.6 / 49.1& 51.3 / 61.5 \\
% \specialrule{0em}{2pt}{1pt}
\textit{Finetuning}          & 67.3 / 100.0& 68.5 / 100.0& 67.9 / 100.0& 12.5 / 40.0& 60.1 / 99.0& 69.5 / 99.5& 63.2 / 98.9& 68.5 / 96.5& 31.2 / 59.4& 63.4 / 98.0 \\

% \specialrule{0em}{2pt}{1pt}
\textit{Regularization}      & 58.2 / 100.0& 59.2 / 100.0& 60.4 / 100.0& 60.5 / 40.0& 25.9 / 49.7& 57.1 / 100.0& 31.9 / 69.0& 59.5 / 96.5& 36.1 / 70.7& 34.7 / 52.4 \\

% \specialrule{0em}{2pt}{1pt}
\textit{Distraction}         & 56.1 / 48.7& 50.1 / 49.2& 60.9 / 48.4& 58.7 / 46.0& 64.7 / 100.0& 59.8 / 51.1& 63.7 / 99.5& 56.6 / 50.6& 27.0 / 49.1& 64.4 / 85.8 \\
% \specialrule{0em}{2pt}{1pt}
\textit{Laundering}          & 67.6 / 49.8& 68.0 / 49.2& 67.8 / 48.4& 46.7 / 40.6& 47.3 / 49.7& 58.9 / 71.0& 60.1 / 96.7& 66.9 / 93.9& 33.3 / 49.1& 62.1 / 91.9 \\

\hline

\textsc{Dehydra}\textsubscript{Basic}        & 72.6 / 48.7& 63.2 / 57.9& 63.0 / 55.6& 65.8 / 44.8& 64.5 / 49.7& 62.1 / 57.2& 65.2 / 57.5& 61.2 / 57.1& 54.5 / 58.4& 64.0 / 60.5 \\
% \specialrule{0em}{2pt}{1pt}
\textsc{Dehydra}\textsubscript{Improved}      & 75.1 / 48.7& 73.9 / 50.2& 73.1 / 53.1& 71.0 / 41.8& 67.1 / 49.7& 64.8 / 60.3& 65.4 / 54.8& 62.8 / 59.6& 65.1 / 55.8& 64.5 / 58.5 \\
\hline
\end{tabular}
% \footnotesize\raggedright{*The baseline \pxd{in this scenario has slightly different settings}.}
\end{table*}

\begin{table*}[!tb]
\caption{Comparison of removal attacks against black-box watermarks under the transfer setting on CIFAR-100.
% x/y denotes clean accuracy/rescaled watermark accuracy, and values behind ± report the standard deviation in 5 repetitive tests.
}
\label{tab:transfer_100}
\vspace{-0.1in}
\centering
\small
\setlength{\tabcolsep}{2.5pt}
\begin{tabular}{c cccccccccc}
\hline
\textbf{Attacks} & \textbf{\textit{Content}} & \textbf{\textit{Noise}} & \textbf{\textit{Unrelated}} & \textbf{\textit{Piracy}} & \textbf{\textit{EWE}} & \textbf{\textit{Adi}} & \textbf{\textit{AFS}} & \textbf{\textit{EW}} & \textbf{\textit{Blind}} & \textbf{\textit{Mark}} \\
\hline
    \textit{None}    & 76.3 / 100.0& 76.7 / 100.0& 75.9 / 100.0& 73.7 / 100.0& 67.7 / 100.0& 75.6 / 100.0& 65.8 / 100.0& 76.0 / 100.0& 67.2 / 100.0& 65.2 / 100.0 \\
\hline
    \textit{Regularization}
            & 59.3 / 100.0& 53.9 / 100.0& 52.6 / 100.0& 48.4 / 100.0& 42.3 / 49.7& 60.3 / 100.0& 52.2 / 90.2& 57.1 / 95.5& 43.8 / 88.2& 49.3 / 52.4 \\
% \specialrule{0em}{2pt}{1pt}
\textit{Distraction}
            & 76.3 / 90.3& 76.4 / 100.0& 76.0 / 99.5& 72.3 / 52.6& 64.9 / 100.0& 56.3 / 50.1& 63.1 / 97.8& 59.4 / 57.6& 66.7 / 98.5& 63.8 / 77.7 \\
\hline
\textsc{Dehydra}\textsubscript{Basic} 
            & 72.7 / 48.7& 56.6 / 51.3& 62.3 / 58.2& 64.1 / 50.2& 56.2 / 59.8& 59.1 / 58.3& 58.4 / 54.8& 59.5 / 59.1& 55.7 / 54.8& 55.4 / 60.0 \\
% \specialrule{0em}{2pt}{1pt}
\textsc{Dehydra}\textsubscript{Improved} 
            & 74.0 / 48.7& 75.2 / 49.7& 70.0 / 58.2& 71.0 / 44.8& 62.2 / 59.3& 55.1 / 52.7& 58.8 / 49.9& 61.9 / 58.6& 56.1 / 54.3& 61.1 / 58.5 \\
\hline
\end{tabular}
% \footnotesize\raggedright{*The PST-Out attack has a slightly different scenario, as described in Table \ref{tab:in_distribution}.}
\end{table*}

\begin{table}[tb]
\caption{Comparison of removal attacks against black-box watermarks under the data-free setting on CIFAR-100.
% x/y denotes clean accuracy/rescaled watermark accuracy, and values behind ± report the standard deviation in 5 repetitive tests.
}
\label{tab:data_free_100}
\vspace{-0.1in}
\centering
\footnotesize
\setlength{\tabcolsep}{1.5pt}
\centering
\begin{tabular}{c cccccc}
\hline
\textbf{Attacks} & \textbf{\textit{Content}} & \textbf{\textit{Noise}} & \textbf{\textit{Unrelated}} & \textbf{\textit{Piracy}} & \textbf{\textit{EWE}} \\
\hline
    \textit{None}        & 76.3 / 100.0& 76.7 / 100.0& 75.9 / 100.0& 73.7 / 100.0& 67.7 / 100.0 \\
\hline
    \textit{Pruning}     & 66.4 / 100.0& 66.9 / 100.0& 58.7 / 100.0& 26.7 / 40.0& 60.2 / 100.0 \\
\hline
% \shortstack[c]{MIRA\\(Improved)}
%                 & 70.0 / 50.8& 72.2 / 51.8& 61.1 / 58.7& 71.1 / 41.8& 57.5 / 60.3 \\

\makecell{\textsc{Dehydra}\\(Improved)} & 

\makecell{70.0 / 50.8}& \makecell{ 72.2 / 51.8} & \makecell{61.1 / 58.7 }& \makecell{71.1 / 41.8}& \makecell{57.5 / 60.3} \\

% \multirow{2}{*}{MIRA\\(Improved)} & 70.0 / 50.8& 72.2 / 51.8& 61.1 / 58.7& 71.1 / 41.8& 57.5 / 60.3 \\

\hline
\end{tabular}
\end{table}
% \subsection{Comparison with Other Recovering Methods} \label{app:comp_recover}\phantom{}

\section{Proof of Watermark Smoothness Discrepancy}\label{sec:proof}

\subsection{Problem Definition and Theorems}
We formulate the \textit{watermark smoothness discrepancy} hypothesis in Theorem \ref{th:wm_smooth_complete} on mixture Gaussian distributions w.r.t. a linear model $f(x) = \text{sign} (\langle \boldsymbol{w}, \boldsymbol{x} \rangle +b)$.
These settings are widely used in previous analytical studies for trustworthy machine learning due to the intractable complexity of analyzing DNNs \cite{ma2022tradeoff, xu2021robust}.

Concretely, suppose a \louis{balanced} binary classification dataset $\mathcal{D}_{nor} = (\boldsymbol{x},y) \in \mathbb{R}^d \times \{\pm 1\}$, with underlying distributions $\mathcal{N}\left(\boldsymbol{\mu_+}, \sigma^2 I\right)$ and $\mathcal{N}\left(\boldsymbol{\mu_-}, \sigma^2 I\right)$ for the two classes.
% \begin{equation}
%     y=\left\{\begin{array} {ll} 
% { + 1 , p = 0.5 } \\
% { - 1 , p = 0.5 , }
% \end{array} \boldsymbol{x} \sim \left\{\begin{array}{ll}
% \mathcal{N}\left(\boldsymbol{\mu_+}, \sigma^2 I\right) \text { if } y=+1 \\
% \mathcal{N}\left(\boldsymbol{\mu_-}, \sigma^2 I\right) \text { if } y=-1.
% \end{array}\right.\right.
% \end{equation}
The watermark data, $\mathcal{D}_{wm} = \mathcal{N}(\boldsymbol{\mu_{wm}}, \sigma_{wm}^2 I) \times \{+1\}$, have a fixed target class ``+1'' and are sampled with ratio $p$ compared to the normal dataset during training.
Also, we define the input smoothness $S_I$ and parameter smoothness $S_P$ to be the performance of model under input-level and parameter-level Gaussian perturbations respectively.
Then we have the following theorem:

\begin{theorem}\label{th:wm_smooth_complete}
% If the watermark distribution satisfies the condition $\frac{p}{0.5+p}(\sigma_{wm}^2-\sigma^2) + \frac{0.5p}{(0.5+p)^2}(\boldsymbol{\mu_{wm}}-\boldsymbol{\mu_+})^2 >0$
% and the sampling ratio of watermark data $p > 0.5 \sqrt{\frac{0.5}{0.5+p} + \frac{p}{0.5+p}\frac{\sigma_{wm}^2}{\sigma^2}+\frac{0.5p}{(0.5+p)^2}\frac{(\boldsymbol{\mu_{wm}}-\boldsymbol{\mu_+})^2}{\sigma^2}} - 0.5$,
Consider the optimal linear model $f^*(x) = \text{sign} (\langle \boldsymbol{w^*}, \boldsymbol{x} \rangle +b^*)$ obtained by minimizing the misclassification risk on $\mathcal{D}_{nor} \cup p\mathcal{D}_{wm}$.
If the watermark data satisfies (1) $\frac{p}{0.5+p}(\sigma_{wm}^2-\sigma^2) + \frac{0.5p}{(0.5+p)^2}(\boldsymbol{\mu_{wm}}-\boldsymbol{\mu_+})^2 >0$
and (2) $p > 0.5 \sqrt{\frac{0.5}{0.5+p} + \frac{p}{0.5+p}\frac{\sigma_{wm}^2}{\sigma^2}+\frac{0.5p}{(0.5+p)^2}\frac{(\boldsymbol{\mu_{wm}}-\boldsymbol{\mu_+})^2}{\sigma^2}} - 0.5$,
then we have $S_{I,+1}(f^*) > S_{I,-1}(f^*)$ and $S_{P,+1}(f^*) > S_{P,-1}(f^*)$.
\end{theorem}

\noindent \textbf{Remark.} \textit{The theorem intuitively indicates that, for a fixed-class watermark, the target class will exhibit a higher smoothness (both input-level and parameter-level) compared to other classes if the watermark samples are (1) diverse internally and different from the normal ones in the target class, and (2) sampled much more frequently.
% watermark samples which are either diverse or different from the normal ones in the target class (the first condition in the theorem), together with a much higher sampling ratio (the second condition), would induce a higher smoothness of the target class.
% Note that the second requirement is a little complicated, but we can find an upper bound for the expression on the right-hand side, that is $0.5 \sqrt{\frac{0.5}{0.5+p} + \frac{\sigma_{wm}^2}{\sigma^2}+0.25\frac{(\boldsymbol{\mu_{wm}}-\boldsymbol{\mu_+})^2}{\sigma^2}} - 0.5$, which is monotonically decreasing w.r.t. $p$.
During proof, we also find the additive input-level and parameter-level noise would further enlarge the class-wise discrepancies.}

% \louis{Intuitively, the first requirement in the above theorem favors watermark samples that are either diverse or different from the normal ones in the target class.
% The second requirement seems a little complicated, but \lyf{it is easy to find an upper bound for the expression on the right-hand side, that is $0.5 \sqrt{\frac{0.5}{0.5+p} + \frac{\sigma_{wm}^2}{\sigma^2}+0.25\frac{(\boldsymbol{\mu_{wm}}-\boldsymbol{\mu_+})^2}{\sigma^2}} - 0.5$, which is monotonically decreasing w.r.t. $p$.
% Therefore, a larger sampling ratio of watermark samples should be preferred} for watermark smoothness discrepancy.}

\subsection{Proof Sketch}

We provide the proof sketch for the watermark smoothness discrepancy here.
The complete proof is shown in the next section.

\begin{proofsketch}
We first consider a simpler mixture Gaussian distribution, where $x\sim\mathcal{N}(\boldsymbol{1}, \sigma_+^2 I)$ for class ``+1'' with prior probability $\alpha$ and $x\sim\mathcal{N}(-\boldsymbol{1}, \sigma_-^2 I)$ for class ``-1'' with prior probability $1-\alpha$.
Under this setting, the parameter of the optimal model $f^*(x) = \text{sign} (\langle \boldsymbol{w^*}, \boldsymbol{x} \rangle +b^*)$ can be derived by solving the equation where the derivative of the misclassification risk equals zero, i.e. $\frac{dR(f^*)}{df^*}=0$.

Next, we prove two lemmas on $f^*$'s class-wise discrepancies for input smoothness and parameter smoothness, respectively.
Specifically, for input smoothness, \textit{if $\sigma_+ > \sigma_-$ and $\frac{\alpha\sigma_-}{(1-\alpha)\sigma_+} > 1$, then we have $S_{I,+1}(f^*) > S_{I,-1}(f^*)$}.
This means that if class ``+1'' has a larger variance and is sampled more heavily, then $f^*$ should perform better on class ``+1'' under input-level perturbations, compared with that on class ``-1''.
A similar conclusion holds for parameter smoothness.
During the derivation, we also find that the variance of the additive perturbations is positively correlated with the class-wise discrepancies.

Finally, for a watermarked model trained jointly on $\mathcal{D}_{nor} \cup p\mathcal{D}_{wm}$, we show that the training data belonging to class ``+1'', including $\mathcal{N}\left(\boldsymbol{\mu_+}, \sigma^2 I\right)$ and $\mathcal{N}(\boldsymbol{\mu_{wm}}, \sigma_{wm}^2 I)$, actually forms a new mixture Gaussian distribution.
Therefore, this problem can be easily reduced to the lemmas' settings above, by translating, rotating and scaling the coordinate system.
Substituting the concrete parameters of watermark data into the conditions $\sigma_+ > \sigma_-$ and $\frac{\alpha\sigma_-}{(1-\alpha)\sigma_+} > 1$ yields the final conditions for watermark smoothness discrepancy as in Theorem \ref{th:wm_smooth_complete}.
\end{proofsketch}

\subsection{Complete Proof}
We first consider a binary classification problem on $\mathcal{D}=(\boldsymbol{x},y) \in \mathbb{R}^d \times \{\pm 1\}$ with classes ``+1'' and ``-1'', each following a Gaussian distribution.
Without loss of generality, we assume the distribution of $\mathcal{D}$ as:
\begin{equation}
    y=\left\{\begin{array} { l l } 
{+ 1 , p = \alpha } \\
{ - 1 , p = 1 - \alpha , }
\end{array}  \boldsymbol{x} \sim \left\{\begin{array}{ll}
\mathcal{N}\left(\boldsymbol{1}, \sigma_{+}^2 I\right) \text{ if } y=+1 \\
\mathcal{N}\left(-\boldsymbol{1}, \sigma_{-}^2 I\right) \text{ if } y=-1,
\end{array}\right.\right.
\end{equation}
where $\alpha$ is the prior probability for class ``+1'', $\boldsymbol{1} = (1,1,\dots,1)$ is a $d$-dim vector and $I$ is a $d$-dim identity matrix.
This simple setting will be adapted to a more complex one in Section \ref{sec:proof4}, with additional watermark data incorporated during training.

% The linear model takes the following form:
% \begin{equation}
%     f(x) = \text{sign} (\langle \boldsymbol{w}, \boldsymbol{x} \rangle +b),
% \end{equation}
% where $\boldsymbol{w}\in \mathbb{R}^d$, $b\in \mathbb{R}$, $\text{sign}(\cdot)$ returns 1 when the input scalar is greater than or equal to 0, and -1 otherwise.

We firstly derive the parameter form of the optimal model obtained during training, and analyze the class-wise smoothness discrepancies of the model during testing subsequently.

\subsubsection{Optimal Model}\label{sec:optimal_model}
% \textcolor{red}{\small
% This subsection shares similar derivations and results with the paper \textit{On the tradeoff between robustness and fairness}.
% Should we cite them, and directly use their theorem?
% There are several distinctions between our work and theirs:
% \begin{enumerate}
%     \item We made a simpler assumption on the Gaussian form with the mean $\boldsymbol{1}$, instead of their $\boldsymbol{\mu}_{+}$ (the robustness discrepancy cannot be easily derived without assuming $\sigma_+=\sigma_-$ in their paper).
%     \item We fixed an error in the solution (Equation \ref{eq:optimal_solution}) to the optimal model equation.
%     \item We used a more meaningful variable $K$.
% \end{enumerate}
% }

Denote $f^*(x) = \text{sign} (\langle \boldsymbol{w^*}, \boldsymbol{x} \rangle +b^*)$ to be the optimal model.
Due to the assumed Gaussian distributions with origin symmetry, we can prove $w_1^*=w_2^*=\dots=w_d^* > 0$ for $\boldsymbol{w^*}$ by contradiction.
For simplicity, we can also define $\eta^*=\frac{b^*}{w_1^*}$, which is the parameter to solve for $f^*$.

Based on the properties of Gaussian distributions, the empirical risk to minimize for class ``+1'' is:
\begin{equation}
\begin{aligned}
    R_{+1}(f^*) & =\mathbb{P}\left(\sum_{i=1}^d w_i^* x_i+b^*<0\right)=\mathbb{P}\left(w_1^* \sum_{i=1}^d x_i+b^*<0\right) \\
    % & =\mathbb{P}\left(\frac{w_1^* \sum_{i=1}^d\left(x_i-\mu_{+}\right)}{\sqrt{w_1^{* 2} \sum_{i=1}^d \sigma_{+}^2}}<\frac{-b-d w_1^* \mu_{+}}{\sqrt{w_1^{* 2} \sum_{i=1}^d \sigma_{+}^2}}\right) \\
    & =\Phi\left(-\frac{b^*+d w_1^*}{\sqrt{d} w_1^* \sigma_{+}}\right)
    = \Phi\left(-\frac{\eta^*+d}{\sqrt{d} \sigma_{+}}\right),
\end{aligned}
\end{equation}
and the empirical risk for class ``-1'' is:
\begin{equation}
\begin{aligned}
    R_{-1}(f^*) & =\mathbb{P}\left(\sum_{i=1}^d w_i^* x_i+b^*>0\right)=\mathbb{P}\left(w_1^* \sum_{i=1}^d x_i+b^*>0\right) \\
    & = 1 - \Phi\left(\frac{-b^* + d w_1^*}{\sqrt{d} w_1^* \sigma_{-}}\right)
    = \Phi\left(\frac{\eta^*-d}{\sqrt{d} \sigma_{-}}\right).
\end{aligned}
\end{equation}
% \begin{equation}
%     R^{-1}(f^*) = \Phi\left(\frac{b^*-d w_1^*}{\sqrt{d} w_1^* \sigma_{-}}\right)
% \end{equation}
Therefore, the total objective function for $f^*$ is:
\begin{equation}
    R\left(f^*\right)=
    \alpha \Phi\left(-\frac{\eta^*+d}{\sqrt{d} \sigma_{+}}\right)
    + (1-\alpha) \Phi\left(\frac{\eta^*-d}{\sqrt{d} \sigma_{-}}\right).
    \label{eq:liner_objective}
\end{equation}
% \begin{equation}
%     R\left(f^*\right)=
%     \alpha \Phi\left(-\frac{b^*+d w_1^*}{\sqrt{d} w_1^* \sigma_{+}}\right)
%     +(1-\alpha) \Phi\left(\frac{b^*-d w_1^*}{\sqrt{d} w_1^* \sigma_{-}}\right)
% \end{equation}

We can find the optimal $\eta^*$ minimizing the above objective function by taking $\frac{d R(f^*)}{d\eta^*} = 0$, which can be expanded as:
\begin{equation}
    \begin{aligned}
        &\alpha \frac{1}{\sqrt{2 \pi}} \exp \left(-\frac{1}{2}\left(\frac{\eta^*+d}{\sqrt{d} \sigma_{+}}\right)^2\right) \frac{-1}{\sqrt{d} \sigma_{+}}+ \\
        &(1-\alpha) \frac{1}{\sqrt{2 \pi}} \exp \left(-\frac{1}{2}\left(\frac{\eta^*-d}{\sqrt{d} \sigma_{-}}\right)^2\right) \frac{1}{\sqrt{d} \sigma_{-}}=0.
    \end{aligned}
    \label{eq:equation_for_optimal}
\end{equation}
The solution to this equation is:
\begin{equation}
    \eta^* = d \frac{\sigma_{+}^2+\sigma_{-}^2
    -2\sigma_{+} \sigma_{-} 
    \sqrt{1 - \frac{K}{2d} \left(\sigma_{+}^2-\sigma_{-}^2\right)}}
    {\sigma_{+}^2-\sigma_{-}^2},
\label{eq:optimal_solution}
\end{equation}
where $K=\log(\frac{\alpha\sigma_-}{(1-\alpha)\sigma_+})$.
Note that we drop another solution to Equation \ref{eq:equation_for_optimal} since its decision boundary does not lie between the two Gaussians.
In this way, we derive the parameter $\eta^*$ of the optimal model $f^*$ for $\mathcal{D}$.

\subsubsection{Input Smoothness Discrepancy}
We first study the input smoothness $S_{I,+1}(f^*)$ and $S_{I,-1}(f^*)$ of the two classes, i.e. the performance of model on the class samples under input-level Gaussian perturbations $\mathcal{N}(\boldsymbol{0}, \sigma_I^2 I)$.

\begin{lemma}\label{th:input_smooth}
Consider the linear model $f^*(x) = \text{sign} (\langle \boldsymbol{w^*}, \boldsymbol{x} \rangle +b^*)$ obtained by minimizing Equation \ref{eq:liner_objective} on mixture Gaussian distribution $\mathcal{D}$.
If $\sigma_+ > \sigma_-$ and $\frac{\alpha\sigma_-}{(1-\alpha)\sigma_+} > 1$, then we have $S_{I,+1}(f^*) > S_{I,-1}(f^*)$.
\end{lemma}

\begin{proof}
Denote $\epsilon_I$ to be the noise added to each input element.
The input smoothness of class ``+1'' and ``-1'' can be expressed as follows:
% \begin{equation}
% \begin{aligned}
%     S_{I,+1}(f^*) & =\mathbb{P}\left(\sum_{i=1}^d w_i^* (x_i+\epsilon_I)+b^*>0\right) \\
%     & = \Phi\left(\frac{\eta^* + d}{\sqrt{d(\sigma_{+}^2 + \sigma_I^2)}}\right),
% \end{aligned}
% \end{equation}
\begin{equation}
S_{I,+1}(f^*) =\mathbb{P}\left(\sum_{i=1}^d w_i^* (x_i+\epsilon_I)+b^*>0\right)
= \Phi\left(\frac{\eta^* + d}{\sqrt{d(\sigma_{+}^2 + \sigma_I^2)}}\right),
\end{equation}
% \begin{equation}
% \begin{aligned}
%     S_{I,-1}(f^*) & =\mathbb{P}\left(\sum_{i=1}^d w_i^* (x_i+\epsilon_I)+b^*<0\right) \\
%     & = \Phi\left(\frac{-\eta^* + d}{\sqrt{d(\sigma_{-}^2 + \sigma_I^2)}}\right).
% \end{aligned}
% \end{equation}
\begin{equation}
S_{I,-1}(f^*) =\mathbb{P}\left(\sum_{i=1}^d w_i^* (x_i+\epsilon_I)+b^*<0\right)
= \Phi\left(\frac{-\eta^* + d}{\sqrt{d(\sigma_{-}^2 + \sigma_I^2)}}\right).
\end{equation}
Here we can find that
\begin{equation}
\begin{aligned}
    &\frac{\Phi^{-1}(S_{I,+1}(f^*))}{\sqrt{d}} - \frac{\Phi^{-1}(S_{I,-1}(f^*))}{\sqrt{d}}
    = \frac{1+\frac{\eta^*}{d}}{\sqrt{\sigma_{+}^2 + \sigma_I^2}} - \frac{1 - \frac{\eta^*}{d}}{\sqrt{\sigma_{-}^2 + \sigma_I^2}} \\
    &= \frac{\sqrt{\sigma_{-}^2 + \sigma_I^2} \cdot (2\sigma_+^2-2\sigma_+\sigma_-\sqrt{1 - \frac{K}{2d} \left(\sigma_{+}^2-\sigma_{-}^2\right)})}
    {\sqrt{\sigma_{+}^2 + \sigma_I^2} \cdot \sqrt{\sigma_{-}^2 + \sigma_I^2} \cdot (\sigma_{+}^2-\sigma_{-}^2)} \\
    &- \frac{\sqrt{\sigma_{+}^2 + \sigma_I^2} \cdot (-2\sigma_-^2+2\sigma_+\sigma_-\sqrt{1 - \frac{K}{2d} \left(\sigma_{+}^2-\sigma_{-}^2\right)})}
    {\sqrt{\sigma_{+}^2 + \sigma_I^2} \cdot \sqrt{\sigma_{-}^2 + \sigma_I^2} \cdot (\sigma_{+}^2-\sigma_{-}^2)} \\
    &>_{(1)} \frac{\sqrt{\sigma_{-}^2 + \sigma_I^2} \cdot (2\sigma_+^2-2\sigma_+\sigma_-) - \sqrt{\sigma_{+}^2 + \sigma_I^2} \cdot (-2\sigma_-^2+2\sigma_+\sigma_-)}
    {\sqrt{\sigma_{+}^2 + \sigma_I^2} \cdot \sqrt{\sigma_{-}^2 + \sigma_I^2} \cdot (\sigma_{+}^2-\sigma_{-}^2)} \\
    % &= \frac{2\sigma_+\sqrt{\sigma_{-}^2 + \sigma_I^2} - 2\sigma_-\sqrt{\sigma_{+}^2 + \sigma_I^2}}
    % {\sqrt{\sigma_{+}^2 + \sigma_I^2} \cdot \sqrt{\sigma_{-}^2 + \sigma_I^2} \cdot (\sigma_{+}+\sigma_{-})} \\
    &= \frac{2 \sigma_+\sigma_- (\sqrt{1 + \frac{\sigma_I^2}{\sigma_-^2}} - \sqrt{1 + \frac{\sigma_I^2}{\sigma_+^2}})}
    {\sqrt{\sigma_{+}^2 + \sigma_I^2} \cdot \sqrt{\sigma_{-}^2 + \sigma_I^2} \cdot (\sigma_{+}+\sigma_{-})} >_{(2)} 0.
\end{aligned}
\end{equation}
Therefore, we can prove $S_{I,+1}(f^*) > S_{I,-1}(f^*)$.

Note that the first inequality $>_{(1)}$ above exploits $\sqrt{1 - \frac{K}{2d} \left(\sigma_{+}^2-\sigma_{-}^2\right)} < 1$ since $\sigma_+ > \sigma_-$ and $K=\log(\frac{\alpha\sigma_-}{(1-\alpha)\sigma_+}) > 0$.
The second inequality $>_{(2)}$ holds due to the extra noise $\mathcal{N}(\boldsymbol{0}, \sigma_{I}^2 I)$, \textbf{which further enlarges the class-wise discrepancy}. 
In the extreme case of $\sigma_I=0$, this inequality turns to an equality, and the smoothness discrepancy degenerates to natural performance discrepancy.

\end{proof}

\subsubsection{Parameter Smoothness Discrepancy}
We also study the parameter smoothness $S_{P,+1}(f^*)$ and $S_{P,-1}(f^*)$ here, i.e. the performance of model on the class samples under parameter-level (weights and biases) Gaussian perturbations $\mathcal{N}(\boldsymbol{0}, \sigma_P^2 I)$.
Our finding is similar to Lemma \ref{th:input_smooth}.

\begin{lemma}\label{th:param_smooth}
Consider the linear model $f^*(x) = \text{sign} (\langle \boldsymbol{w^*}, \boldsymbol{x} \rangle +b^*)$ obtained by minimizing Equation \ref{eq:liner_objective} on mixture Gaussian distribution $\mathcal{D}$.
If $\sigma_+ > \sigma_-$ and $\frac{\alpha\sigma_-}{(1-\alpha)\sigma_+} > 1$, then we have $S_{P,+1}(f^*) > S_{P,-1}(f^*)$.
\end{lemma}

\begin{proof}
Since the product of two Gaussian random variables is hard to process, we approximately study the parameter smoothness on the centers of $\mathcal{N}(\boldsymbol{1}, \sigma_+^2I)$ and $\mathcal{N}(-\boldsymbol{1}, \sigma_-^2I)$, respectively.
Denote $\epsilon_P$ to be the noise added to each parameter element.
The parameter smoothness of class ``+1'' and ``-1'' can be expressed as follows:
\begin{equation}
\begin{aligned}
    S_{P,+1}(f^*) & =\mathbb{P}\left(\sum_{i=1}^d (w_i^* + \epsilon_P) \cdot \boldsymbol{1}+(b^*+\epsilon_P)>0\right) \\
    & = \Phi\left(\frac{b^* + dw_1^*}{\sqrt{d+1}\sigma_P}\right)
    = \Phi\left(\frac{w_1^* (d + \eta^*)}{\sqrt{d+1}\sigma_P}\right),
\end{aligned}
\end{equation}
\begin{equation}
\begin{aligned}
    S_{P,-1}(f^*) & =\mathbb{P}\left(\sum_{i=1}^d (w_i^* + \epsilon_P) \cdot (-\boldsymbol{1})+(b^*+\epsilon_P)<0\right) \\
    & = \Phi\left(\frac{-b^* + dw_1^*}{\sqrt{d+1}\sigma_P}\right)
    = \Phi\left(\frac{w_1^* (d - \eta^*)}{\sqrt{d+1}\sigma_P}\right).
\end{aligned}
\end{equation}

Note that $w_1^*$ is proved to be positive in Section \ref{sec:optimal_model}.
Since $\sigma_+ > \sigma_-$ and $K=\log(\frac{\alpha\sigma_-}{(1-\alpha)\sigma_+}) > 0$, we have $\sqrt{1 - \frac{K}{2d} \left(\sigma_{+}^2-\sigma_{-}^2\right)} < 1$ and hence $\eta^* > d \cdot \frac{\sigma_{+}^2+\sigma_{-}^2
    -2\sigma_{+} \sigma_{-}\cdot 1}
    {\sigma_{+}^2-\sigma_{-}^2}
    = d \cdot \frac{(\sigma_{+}-\sigma_-)^2}
    {\sigma_{+}^2-\sigma_{-}^2}
    > 0$.
Therefore, we can prove $S_{P,+1}(f^*) > S_{P,-1}(f^*)$.

\end{proof}

\subsubsection{Adaptation to a Watermarked Model}\label{sec:proof4}
Finally, we study the smoothness property of a linear model with fixed-class watermark (target class is ``+1'') in a more realistic setting.
For a normal dataset $\mathcal{D}_{nor} = (\boldsymbol{x},y) \in \mathbb{R}^d \times \{\pm 1\}$, suppose the distribution is:
% suppose that class ``+1'' and ``-1'' have the same sample sizes as well as the distribution variances:
\begin{equation}
    y=\left\{\begin{array} { l l } 
{ + 1 , p = 0.5 } \\
{ - 1 , p = 0.5 , }
\end{array} \boldsymbol{x} \sim \left\{\begin{array}{ll}
\mathcal{N}\left(\boldsymbol{\mu_+}, \sigma^2 I\right) \text { if } y=+1 \\
\mathcal{N}\left(\boldsymbol{\mu_-}, \sigma^2 I\right) \text { if } y=-1,
\end{array}\right.\right.
\end{equation}
and the watermark data, $\mathcal{D}_{wm} = \mathcal{N}(\boldsymbol{\mu_{wm}}, \sigma_{wm}^2 I) \times \{+1\}$, are sampled with ratio $p$ compared to the normal dataset during training.

\begin{theorem}\label{th:wm_smooth}
Consider the linear model $f^*(x) = \text{sign} (\langle \boldsymbol{w^*}, \boldsymbol{x} \rangle +b^*)$ obtained by minimizing Equation \ref{eq:liner_objective} on $\mathcal{D}_{nor} \cup p\mathcal{D}_{wm}$.
% If the watermark distribution satisfies the condition $\frac{p}{0.5+p}(\sigma_{wm}^2-\sigma^2) + \frac{0.5p}{(0.5+p)^2}(\boldsymbol{\mu_{wm}}-\boldsymbol{\mu_+})^2 >0$
% and the sampling ratio of watermark data $p > 0.5 \sqrt{\frac{0.5}{0.5+p} + \frac{p}{0.5+p}\frac{\sigma_{wm}^2}{\sigma^2}+\frac{0.5p}{(0.5+p)^2}\frac{(\boldsymbol{\mu_{wm}}-\boldsymbol{\mu_+})^2}{\sigma^2}} - 0.5$,
If the watermark data satisfies $\frac{p}{0.5+p}(\sigma_{wm}^2-\sigma^2) + \frac{0.5p}{(0.5+p)^2}(\boldsymbol{\mu_{wm}}-\boldsymbol{\mu_+})^2 >0$
and $p > 0.5 \sqrt{\frac{0.5}{0.5+p} + \frac{p}{0.5+p}\frac{\sigma_{wm}^2}{\sigma^2}+\frac{0.5p}{(0.5+p)^2}\frac{(\boldsymbol{\mu_{wm}}-\boldsymbol{\mu_+})^2}{\sigma^2}} - 0.5$,
then we have $S_{I,+1}(f^*) > S_{I,-1}(f^*)$ and $S_{P,+1}(f^*) > S_{P,-1}(f^*)$.
\end{theorem}

\begin{proof}
The training data belonging to class ``+1'' forms a mixture Gaussian distribution $\mathcal{N}(\boldsymbol{\mu_{mix}}, \sigma_{mix}^2 I)$ with
\begin{equation}
\begin{aligned}
    \boldsymbol{\mu_{mix}} &=  \frac{0.5}{0.5+p}\boldsymbol{\mu_+} + \frac{p}{0.5+p}\boldsymbol{\mu_{wm}}, \\
    \sigma_{mix}^2 &= \frac{0.5}{0.5+p}\sigma^2 + \frac{p}{0.5+p}\sigma_{wm}^2 + \frac{0.5p}{(0.5+p)^2}(\boldsymbol{\mu_{wm}} - \boldsymbol{\mu_+})^2.
\end{aligned}
\end{equation}
The model $f^*$ are jointly trained on $\mathcal{N}(\boldsymbol{\mu_{mix}}, \sigma_{mix}^2 I) \times \{+1\}$ and $\mathcal{N}(\boldsymbol{\mu_{-}}, \sigma^2 I) \times \{-1\}$, with prior probability $\alpha=\frac{0.5+p}{1+p}$ for class ``+1''.

We can reduce this problem to the lemmas above by translating, rotating the coordinate system and scaling by $\frac{2\sqrt{d}}{\|\boldsymbol{\mu_{mix}} - \boldsymbol{\mu_-}\|}$,
to project the distribution $\mathcal{N}(\boldsymbol{\mu_{mix}}, \sigma_{mix}^2 I)$ to $\mathcal{N}(\boldsymbol{1}, \frac{4d}{(\boldsymbol{\mu_{mix}} - \boldsymbol{\mu_-})^2}\sigma_{mix}^2 I)$
and $\mathcal{N}(\boldsymbol{\mu_{-}}, \sigma^2 I)$ to $\mathcal{N}(-\boldsymbol{1}, \frac{4d}{(\boldsymbol{\mu_{mix}} - \boldsymbol{\mu_-})^2}\sigma^2 I)$.

In this way, the condition $\frac{p}{0.5+p}(\sigma_{wm}^2-\sigma^2) + \frac{0.5p}{(0.5+p)^2}(\boldsymbol{\mu_{wm}}-\boldsymbol{\mu_+})^2 >0$ yields
\begin{equation}
    \sigma_{mix}^2 > \sigma^2,
\end{equation}
and $p > 0.5 \sqrt{\frac{0.5}{0.5+p} + \frac{p}{0.5+p}\frac{\sigma_{wm}^2}{\sigma^2}+\frac{0.5p}{(0.5+p)^2}\frac{(\boldsymbol{\mu_{wm}}-\boldsymbol{\mu_+})^2}{\sigma^2}} - 0.5$ yields
\begin{equation}
    \frac{0.5+p}{0.5} > \frac{\sigma_{mix}}{\sigma}.
\end{equation}
Finally, we obtain $\frac{2\sqrt{d}}{\|\boldsymbol{\mu_{mix}} - \boldsymbol{\mu_-}\|}\sigma_{mix} > \frac{2\sqrt{d}}{\|\boldsymbol{\mu_{mix}} - \boldsymbol{\mu_-}\|}\sigma$
and $\frac{\alpha\cdot\frac{2\sqrt{d}}{\|\boldsymbol{\mu_{mix}} - \boldsymbol{\mu_-}\|}\sigma}{(1-\alpha)\cdot\frac{2\sqrt{d}}{\|\boldsymbol{\mu_{mix}} - \boldsymbol{\mu_-}\|}\sigma_{mix}} > 1$. Applying Lemma \ref{th:input_smooth} and \ref{th:param_smooth}, we can conclude that $S_{I,+1}(f^*) > S_{I,-1}(f^*)$ and $S_{P,+1}(f^*) > S_{P,-1}(f^*)$.
\end{proof}

\end{document}